\documentclass[a4paper,11pt,authoryear]{article}
\newcommand{\drop}[1]{}
\usepackage{natbib}
\usepackage{multicol}
\usepackage{color}
\usepackage{xspace}
\usepackage{enumerate}

\makeatletter
\def\bs{\expandafter\@gobble\string\\}
\def\lb{\expandafter\@gobble\string\{}
\def\rb{\expandafter\@gobble\string\}}
\def\@pdfauthor{C.V.Radhakrishnan}
\def\@pdftitle{elsarticle.cls -- A documentation}
\def\@pdfsubject{Document formatting with elsarticle.cls}
\def\@pdfkeywords{LaTeX, Elsevier Ltd, document class}

\DeclareRobustCommand{\LaTeX}{L\kern-.26em%
        {\sbox\z@ T%
         \vbox to\ht\z@{\hbox{\check@mathfonts
           \fontsize\sf@size\z@
           \math@fontsfalse\selectfont
          A\,}%
         \vss}%
        }%
     \kern-.15em%
    \TeX}
\makeatother

\usepackage{geometry}
\usepackage{xcolor}
\usepackage{graphicx}
\usepackage{lineno,hyperref}
\usepackage{float}
\usepackage{bbold}
\usepackage{amsmath}
\usepackage{amssymb}
\usepackage{amsfonts}
\usepackage{mathrsfs}
\usepackage{mattens}
\usepackage{amsthm}
\usepackage{upgreek}
\usepackage{textcomp}
\usepackage{placeins}
\usepackage{gensymb}
\usepackage{subcaption}
\usepackage{textcomp}
\usepackage{multicol}
\usepackage{todonotes}
\usepackage{listings}
\usepackage{mathtools}
\usepackage{soul}
\usepackage{stackengine}
\usepackage{caption}
\usepackage{subcaption}
\usepackage{array}
\usepackage{enumerate}
\usepackage{siunitx}
\usepackage{fix-cm}
\usepackage{stackengine}

\def\barromannew#1{\sbox0{#1}\dimen0=\dimexpr\wd0+1pt\relax
  \makebox[\dimen0]{\rlap{\vrule width\dimen0 height 0.075ex depth 0.0ex}%
    \rlap{\vrule width\dimen0 height\dimexpr\ht0+0.05ex\relax 
            depth\dimexpr-\ht0+0.09ex\relax}%
    \kern.5pt#1\kern.5pt}}

\def\barroman#1{\sbox0{#1}\dimen0=\dimexpr\wd0+1pt\relax
  \makebox[\dimen0]{\rlap{\vrule width\dimen0 height 0.06ex depth 0.06ex}%
    \rlap{\vrule width\dimen0 height\dimexpr\ht0+0.03ex\relax 
            depth\dimexpr-\ht0+0.09ex\relax}%
    \kern.5pt#1\kern.5pt}}
\newcommand{\calE}{\mathcal{E}}%
\newcommand{\calM}{\mathcal{M}}%
\newcommand{\calP}{\mathcal{P}}%
\newcommand{\calR}{\mathcal{R}}%
\newcommand{\calS}{\mathcal{S}}%
\newcommand{\calV}{\mathcal{V}}%
\newcommand{\mfC}{\boldsymbol{\mathsfit C}}%
\newcommand{\sbbC}{\scriptscriptstyle\boldsymbol{\mathsfit{C}}}%
\newcommand{\sbbA}{\osixs{A}}%

\newcommand{\bfB}{\boldsymbol{B}}%
\newcommand{\bfe}{\boldsymbol{e}}\newcommand{\bfE}{\boldsymbol{E}}%
\newcommand{\bfF}{\boldsymbol{F}}%
\newcommand{\bfG}{\boldsymbol{G}}%
\newcommand{\bfH}{\boldsymbol{H}}%
\newcommand{\bfn}{\boldsymbol{n}}\newcommand{\bfN}{\boldsymbol{N}}%
\newcommand{\bfo}{\boldsymbol{o}}\newcommand{\bfO}{\boldsymbol{O}}%
\newcommand{\bfP}{\boldsymbol{P}}%
\newcommand{\bfr}{\boldsymbol{r}}\newcommand{\bfR}{\boldsymbol{R}}%
\newcommand{\bfS}{\boldsymbol{S}}%
\newcommand{\bfT}{\boldsymbol{T}}%
\newcommand{\bfu}{\boldsymbol{u}}\newcommand{\bfU}{\boldsymbol{U}}%
\newcommand{\bfv}{\boldsymbol{v}}%
\newcommand{\bfw}{\boldsymbol{w}}%
\newcommand{\bfx}{\boldsymbol{x}}%
\newcommand{\bfSigma}{\boldsymbol{\mathit{\Sigma}}}

\newcommand{\iPi}{\mathit{\Pi}}

\newcommand{\iOmega}{\mathit{\Omega}}

\newcommand{\bfPi}{\boldsymbol{\iPi}}
\newcommand{\bfOm}{\boldsymbol{\iOmega}}
\newcommand{\bfzero}{\boldsymbol{0}}

\DeclareMathOperator{\tr}{tr}
\DeclareMathOperator{\trans}{\scriptscriptstyle\top\mskip-2mu}

\newcommand{\idem}{\boldsymbol{\mathit{1}}}

\newcommand{\er}{\bfe_r}
\newcommand{\ep}{\bfe_\varphi}
\newcommand{\et}{\bfe_\vartheta}

\newcommand{\Cp}{\cos\varphi}
\newcommand{\Sp}{\sin\varphi}
\newcommand{\Ct}{\cos\vartheta}
\newcommand{\St}{\sin\vartheta}

\newcommand{\Apar}{A_{\scriptscriptstyle\parallel}}
\newcommand{\Aperp}{A_{\scriptscriptstyle\perp}}

\newcommand{\Npar}{N_{\scriptscriptstyle\parallel}}
\newcommand{\Nperp}{N_{\scriptscriptstyle\perp}}
\newcommand{\Nspin}{N_{\scriptscriptstyle\circlearrowleft}}
\newcommand{\Spar}{S_{\scriptscriptstyle\parallel}}
\newcommand{\Sperp}{S_{\scriptscriptstyle\perp}}

\newcommand{\da}{\,\text{d}a}

\newcommand{\dr}{\,\text{d}r}

\newcommand{\dvarphi}{\,\text{d}\varphi}
\newcommand{\dvar}{\,\text{d}\vartheta}

\newcommand{\srr}{S_{rr}}

\DeclareMathAlphabet{\mathsfit}{T1}{\sfdefault}{\mddefault}{\sldefault}
\SetMathAlphabet{\mathsfit}{bold}{T1}{\sfdefault}{\bfdefault}{\sldefault} 
\input pdf-trans
\newbox\qbox
\def\usecolor#1{\csname\string\color@#1\endcsname\space}
\def\colsplithelp#1#2 #3\relax{%
  \edef\tmpB{\tmpB#1#2 }%
  \ifnum `#1>`9\relax\def\tmpC{#3}\else\colsplithelp#3\relax\fi
}

\input pdf-trans
\newbox\qbox
\def\usecolor#1{\csname\string\color@#1\endcsname\space}
\newcommand\bordercolor[1]{\colsplit{1}{#1}}
\newcommand\fillcolor[1]{\colsplit{0}{#1}}
\newcommand\colsplit[2]{\colorlet{tmpcolor}{#2}\edef\tmp{\usecolor{tmpcolor}}%
  \def\tmpB{}\expandafter\colsplithelp\tmp\relax%
  \ifnum0=#1\relax\edef\fillcol{\tmpB}\else\edef\bordercol{\tmpC}\fi}
\def\colsplithelp#1#2 #3\relax{%
  \edef\tmpB{\tmpB#1#2 }%
  \ifnum `#1>`9\relax\def\tmpC{#3}\else\colsplithelp#3\relax\fi
}
\newcommand\outline[1]{\leavevmode%
  \def\maltext{#1}%
  \setbox\qbox=\hbox{\maltext}%
   \boxgs{Q q 2 Tr \thickness\space w \fillcol\space \bordercol\space}{}%
  \copy\qbox%
}
\newcommand\osix[2][1]{%
\stackengine{0pt}{\def\thickness{.15}\outline{$\boldsymbol{\mathsfit{#2}}$}}{\kern.1pt\outline{$\boldsymbol{\mathsfit{#2}}$}}{O}{l}{F}{F}{L}}
\bordercolor{black} 
\fillcolor{white}
\def\thickness{.15}
\newcommand\osixs[2][1]{%
\stackengine{0pt}{\def\thickness{.15}\outline{$\boldsymbol{\scriptscriptstyle{\mathsfit{#2}}}$}}{\kern.1pt\outline{$\boldsymbol{\scriptscriptstyle{\mathsfit{#2}}}$}}{O}{l}{F}{F}{L}}
\bordercolor{black} 
\fillcolor{white}
\def\thickness{.15}

\newcommand\olf[2][1]{%
\stackengine{0pt}{\def\thickness{.15}\outline{$\boldsymbol{#2}$}}{\kern.1pt\outline{$\boldsymbol{#2}$}}{O}{l}{F}{F}{L}}
\bordercolor{black} 
\fillcolor{white}
\def\thickness{.15}

\newcommand{\sym}{\text{sym}\mskip2mu}
\newcommand{\sk}{\text{skw}\mskip2mu}
\newcommand{\Sym}{\text{Sym}\mskip2mu}
\newcommand{\symm}{\mathbb{s}\mathbb{y}\mathbb{m}\mskip2mu}
\newcommand{\skk}{\mathbb{s}\mathbb{k}\mathbb{w}\mskip2mu}
\newcommand{\xg}{\text{Grad}\mskip2mu}
\newcommand{\xd}{\text{Div}\mskip2mu}
\newcommand{\xc}{\text{Curl}\mskip2mu}
\newcommand{\dv}{\,\text{d}v}

\newcommand{\bfmu}{\boldsymbol{\mu}}

\newcommand{\be}{\begin{equation}}
\newcommand{\ee}{\end{equation}}
\newcommand{\ba}{\begin{aligned}}
\newcommand{\ea}{\end{aligned}}

\definecolor{newred}{RGB}{180,20,5}

\definecolor{newgreen}{RGB}{1,129,30}

\definecolor{newblue}{RGB}{40,100,250}

\usepackage[english]{babel}
\newtheorem{theorem}{Theorem}

\newtheorem{proposition}[theorem]{Proposition}

\makeatletter
\newcommand{\thickhline}{%
    \noalign {\ifnum 0=`}\fi \hrule height 1pt
    \futurelet \reserved@a \@xhline
}
\newcolumntype{"}{@{\hskip\tabcolsep\vrule width 1pt\hskip\tabcolsep}}
\makeatother

\newcommand{\cb}[1]{{\color{blue}#1}}

\usepackage{eqparbox}

\newcommand{\footremember}[2]{%
    \footnote{#2}
    \newcounter{#1}
    \setcounter{#1}{\value{footnote}}%
}

\title{Derivation, characterization, and application of complete orthonormal sequences for representing general three-dimensional states of residual stress} 
\author{%
  Sankalp Tiwari\footremember{alley}{Mechanics and Materials Unit, Okinawa Institute of Science and Technology. Email : sankalp.tiwari@oist.jp}%
  \and Eliot Fried\footremember{trailer}{Mechanics and Materials Unit, Okinawa Institute of Science and Technology. Email: eliot.fried@oist.jp}%
  }
  \date{}

\begin{document}
\maketitle
\begin{abstract}
    Residual stresses are self-equilibrated stresses on unloaded bodies. Owing to their complex origins, it is useful to develop functions that can be linearly combined to represent {\em any} sufficiently regular residual stress field. In this work, we develop orthonormal sequences that span the set of all square-integrable residual stress fields on a given three-dimensional region. These sequences are obtained by extremizing the most general quadratic, positive-definite functional of the stress gradient on the set of all sufficiently regular residual stress fields subject to a prescribed normalization condition; each such functional yields a sequence. For the special case where the sixth-order coefficient tensor in the functional is homogeneous and isotropic and the fourth-order coefficient tensor in the normalization condition is proportional to the identity tensor, we obtain a three-parameter subfamily of sequences. Upon a suitable parameter normalization, we find that the viable parameter space corresponds to a semi-infinite strip. For a further specialized spherically symmetric case, we obtain analytical expressions for the sequences and the associated Lagrange multipliers. Remarkably, these sequences change little across the entire parameter strip. To illustrate the applicability of our theoretical findings, we employ three such spherically symmetric sequences to accurately approximate two standard residual stress fields. Our work opens avenues for future exploration into the implications of different sequences, achieved by altering both the spatial distribution and the material symmetry class of the coefficient tensors, toward specific objectives.
\end{abstract}  

\section{Introduction}
Let $\calR$ be an open, bounded region in three-dimensional point space $\calE$. Let $\bfn$ denote the unit normal, directed outward from $\calR$, on $\partial\calR$. A residual stress field $\bfS$ on $\calR$ is a symmetric second-order tensor field defined such that
\be
\left.\ba
\xd\mskip1mu\bfS=\bf0 \qquad &\text{on} \qquad \calR,\\[4pt]
\bfS\bfn=\bf0 \qquad &\text{on} \qquad \partial\calR,
\ea \mskip3mu\right\}
\label{intro}
\ee
where $\xd$ denotes the divergence operator on $\calR$. In other words, $\bfS$ is a self-equilibrated distribution of stress in an unloaded body that occupies the region $\calR$. 

Residual stresses are ubiquitous, manifesting most commonly in manufactured components and biological tissues. Virtually all manufacturing processes give rise to residual stresses that may adversely affect performance, potentially leading to premature failure. Conversely, through processes like the shot-peening of metals and the tempering of glass, residual stresses are sometimes judiciously introduced to improve properties and performance \citep{schajer2013practical}. Similarly, in the relatively new field of 3D printing, the mass accretion process can be designed to create optimal residual stress profiles tailored to specific applications \citep{zurlo2017printing}. 

Examples of biological systems that use residual stresses to their advantage abound. For instance, leaves require residual stresses to open \citep{oliver2016morphing}, flowers are residually stressed for seed dispersal \citep{evangelista2011mechanics}, blood arteries are residually stressed to prevent large stress concentrations \citep{chuong1986residual} and achieve better functionality ~\citep{sigaeva2019anisotropic}, and residual stresses that develop when growing wood cells contract longitudinally while expanding transversely fortify tree trunks against strong winds \citep{gordon}. Such biological systems have inspired the development of biomimetic devices with self-shaping properties for applications like drug delivery \citep{fernandes2012self} and self-assembly of lithographically structured microcontainers \citep{leong2008thin}, among others. 

The system \eqref{intro} has an infinite number of solutions. To uniquely determine any particular solution $\bfS$ requires the provision of additional information. Such information is usually obtained by introducing a functional constitutive relation between $\bfS$ and the deformation. However, residual stresses often have complex origins, and it may be challenging to ascertain the precise deformation history and material properties of a residually stressed body. From that viewpoint, developing sequences that span the space of {\em all} $\bfS$ satisfying \eqref{intro}, independent of their physical or theoretical origins, is useful. Ideally, any such sequence should depend solely on the region $\calR$; in particular, it should be independent of the deformation history and material properties of any given body.

\citet{tiwari2020basis} obtained one such sequence. That sequence spans the set of all square-integrable residual stress fields on $\calR$. Its elements are the stationary points, henceforth referred to as `extremizers', of the functional 
\be
E_0(\bfS)=\frac{1}{2}\int_{\calR}|\xg\bfS|^2\dv
\label{Eintro}
\ee
over the set 
\be
\calS_0=\left\{\bfT:\bfT\in\Sym,\mskip2mu\xd\bfT=\bfzero,\mskip2mu\bfT\bfn|_{\partial\calR}=\bfzero,\mskip2mu\int_{\calR}|\bfT|^2\dv<\infty,\mskip2mu E_0(\bfT)<\infty\right\},
\label{setS0}
\ee
where $\bfn$ is the unit normal field on $\partial\calR$ and `Sym' denotes the set of all symmetric second-order tensor fields on $\calR$, subject to the normalization condition
\be
\int_{\calR}|\bfS|^2\dv=1.
\label{tiwarinorm}
\ee
It is important to note that the extremization problem studied in \citet{tiwari2020basis} is tacitly predicated on a nondimensionalization.

Sequences that span general residual stress fields have many potential applications. For example, they can be used to obtain finite-dimensional approximations of the desired accuracy of a given residual stress field, as shown by \citet{tiwari2020basis}. Additionally, they can aid in interpolating the residual stress in a manufactured component determined experimentally at a finite number of points, as shown by \citet{tiwari2022thesis}. This approach is to be distinguished from alternatives based on one-dimensional interpolation, using splines, polynomials, Fourier series, etc., usually employed by practitioners, exemplified by the work of \citet{prime2007}. Furthermore, \citet{tiwari2024} recently utilized the sequence obtained in their earlier work to solve traction boundary-value problems in linear elasticity using stress-based variational principles. Finally, in the application of theories in which the residual stress appears explicitly in the constitutive equation (e.g., \citet{hoger1986determination}) or in the development of metamaterials requiring complex residual stress patterns (e.g., \citet{danescu2013spherical}), such sequences may provide a suitable finite-dimensional framework for the associated optimization problem. Generally speaking, the provision of such a sequence furnishes a vocabulary that can be used to discuss, describe, and characterize arbitrary residual stresses without knowledge of the physical mechanisms that may have caused them.

In this paper, we generalize the results in \cite{tiwari2020basis} by establishing that, for spatially varying sixth- and fourth-order coefficient tensors $\osix{A}$ and $\mfC$, respectively, satisfying certain symmetry and positive-definiteness conditions presented in Section \ref{sec:EL}, extremization of {\em each} quadratic functional $E$ defined by
\be
E(\bfS)=\frac{1}{2}\int_{\calR}\xg\bfS\cdot\osix{A}\mskip2mu[\xg\bfS]\dv
\label{E}
\ee
over the set
\be
\calS=\left\{\bfT:\bfT\in\Sym,\mskip2mu\xd\bfT=\bfzero,\mskip2mu\bfT\bfn|_{\partial\calR}=\bfzero,\mskip2mu\int_{\calR}\bfT\cdot\mfC[\bfT]\dv<\infty,\mskip2mu E(\bfT)<\infty\right\}
\label{setS}
\ee 
subject to the normalization condition 
\be
\int_{\calR}\bfS\cdot\mfC[\bfS]\dv=\varsigma\text{vol}\mskip2mu(\calR),
\label{normalization}
\ee
where $\varsigma>0$ is a prescribed constant, yields a sequence of extremizers spanning the set of all square-integrable residual stress fields; the notations $\osix{H}[\boldsymbol{\calM}]$ and $\boldsymbol{\mathsfit{D}}[\bfN]$ represent the action of a sixth-order tensor $\osix{H}$ on a third-order tensor $\boldsymbol{\calM}$ and the action of a fourth-order tensor $\boldsymbol{\mathsfit{D}}$ on a second-order tensor $\bfN$, respectively. Thus, we obtain a {\em family} of residual stress bases, which includes, in particular, the sequence obtained in \citet{tiwari2020basis}.

If $\osix{A}$ is homogeneous and isotropic and $\mfC$ is proportional to the fourth-order identity tensor, we obtain a three-parameter subfamily of sequences. Notably, for parameters corresponding to the choice tacitly employed in the work of \citet{tiwari2020basis}, the Euler--Lagrange equation has features reminiscent of the Stokes equation for an incompressible Newtonian fluid \citep{gurtin1973} and the time-independent Schr\"odinger equation \citep{dirac1930}.

For a further specialized spherically symmetric case, wherein the analytical determination of extremizers is tractable, we remarkably find that the corresponding sequences change only slightly in the vast landscape of viable functionals. Furthermore, for a specific functional choice, the Euler--Lagrange equation reduces to the Helmholtz equation for the extremizing stress. Lastly, through illustrative examples, we find that the spanning properties of different members of the spherically symmetric family of sequences are essentially identical, indicating that there is no evident reason to favor any particular sequence, including that selected by \citet{tiwari2020basis}.

The primary contribution of this paper is that the generality of the results established herein affords great leeway in the choice of the coefficient tensors $\osix{A}$ and $\mfC$. This flexibility extends both to the spatial distributions and material symmetry classes of $\osix{A}$ and $\mfC$, thereby providing a rich landscape of diverse complete sequences. We are optimistic that future work will explore these choices to obtain sequences with desired characteristics complementing existing works. For instance, \citet{zurlo2017printing} present a prescription for mass accretion in 3D printing that produces a desired residual stress in  an arbitrarily shaped body; our work provides a finite-dimensional setting to determine these residual stresses, thereby facilitating the development of tailored accretion protocols.

A secondary contribution lies in the explicit determination of the viable parameter space for homogeneous, isotropic choices of $\osix{A}$ and in demonstrating that even within the small collection of spherically symmetric sequences, constituting a tiny portion of the totality of complete sequences, the interplay of the various parameters leads to interesting phenomena. 

This paper is arranged as follows. In Section \ref{sec:bvp}, we consider, for a given region $\calR$, the first variation of $E$ over the set of all residual stress fields $\bfS$ satisfying a certain normalization condition and the property $E(\bfS)<\infty$. Each choice of viable coefficient tensors $\osix{A}$ and $\mfC$ in the resulting boundary-value problem yields a sequence of extremizers; therefore, we obtain a family of sequences. Toward the end of Section \ref{sec:bvp}, we obtain the dimensionless version of the boundary-value problem. The subsequent analysis and calculations undertaken in the paper focus on this dimensionless problem. In Section \ref{sec:props}, we find several useful properties of the members of the aforementioned family of sequences, the most important being that each member spans the set of all square-integrable residual stress fields on $\calR$. In Section \ref{sec:iso}, we consider the special case in which $\osix{A}$ is homogeneous and isotropic and $\mfC$ is the fourth-order identity tensor, obtaining a three-parameter subfamily of sequences. We next find the viable collection of parameters that ensure the positive definiteness of the integrand of $E$ as defined in \eqref{E}. In Section \ref{sec:spherical}, we obtain analytical solutions for the spherically symmetric sequences of Section \ref{sec:iso} and the accompanying Lagrange multiplier fields, thus enabling comparisons between different members of the family within this specialized case. In Section \ref{sec:examples}, we present two illustrative examples of fitting linear elastic, spherically symmetric residual stress fields with the sequences computed in Section \ref{sec:spherical}. Finally, we summarize our findings in Section \ref{sec:conclusions}.

\section{Variation of \boldmath{$E$}} \label{sec:bvp}
\subsection{Euler--Lagrange equation. Natural boundary condition} \label{sec:EL}
In this section, we consider the problem of extremizing the functional $E$, defined by \eqref{E}, over the set $\calS$, defined by \eqref{setS}, subject to the normalization condition \eqref{normalization}.

We assume that the sixth-order coefficient tensor $\osix{A}$ in \eqref{E} possesses the symmetry conditions
\be
\boldsymbol{\calM}_2\cdot\osix{A}\mskip2mu[\boldsymbol{\calM}_1]=\boldsymbol{\calM}_1\cdot\osix{A}\mskip2mu[\boldsymbol{\calM}_2] \qquad \text{and} \qquad \osix{A}\mskip2mu[\xg\bfN]=\osix{A}\mskip2mu[\xg(\bfN^{\trans})],
\label{Asymmetries}
\ee
for all third-order tensors $\boldsymbol{\calM}_1$ and $\boldsymbol{\calM}_2$ and second-order tensors $\bfN$. Moreover, we assume that the fourth-order coefficient tensor $\mfC$ possesses the symmetry conditions
\be
\bfN_2\cdot\mfC[\bfN_1]=\bfN_1\cdot\mfC[\bfN_2] \qquad \text{and} \qquad \mfC[\bfN_1]=\mfC[\bfN_1^{\trans}],
\label{Csymmetries}
\ee
for all second-order tensors $\bfN_1$ and $\bfN_2$. We furthermore assume that $\osix{A}$ and $\mfC$ satisfy the inequalities 
\be
\xg\bfS(\bfx)\cdot\osix{A}\mskip2mu(\bfx)[\xg\bfS(\bfx)]>0 \qquad \text{and} \qquad \bfS(\bfx)\cdot\mfC(\bfx)[\bfS(\bfx)]>0
\label{Aposdef}
\ee
for all $\bfx$ in $\calR$ and for all non-zero $\bfS$ in $\calS$. By \eqref{Aposdef}$_1$, we see that the functional $E$ in \eqref{E} is positive-definite for all $\bfS$ satisfying \eqref{intro}. We will see in Section \ref{sec:props}, and \ref{app:existence} referenced therein, that the positive definiteness of $E$ and the assumption \eqref{Aposdef}$_2$ suffice to ensure the existence of solutions of the extremization problems considered in this paper.

Notice that the coefficient tensor $\mfC$ possesses the same symmetry and positive definiteness conditions \eqref{Csymmetries} and \eqref{Aposdef}$_2$, respectively, as the fourth-order stiffness and compliance tensors of linear elasticity. However, we emphasize that despite these similarities, the extremization problem described by the conditions \eqref{E}--\eqref{normalization} and its solutions are independent of any functional constitutive relation. In particular, {\em each} sequence of extremizers spans the set of {\em all} square-integrable residual stress fields, not just those with linear elastic origins.

To obtain the conditions on an extremizer of the functional $E$ introduced in \eqref{E}, we note that an admissible variation $\bfU:=\delta\bfS$ of $\bfS$ must be symmetric and, with reference to \eqref{intro}$_{1,2}$, \eqref{normalization}, and \eqref{Csymmetries}$_1$, must satisfy
\be
\left.
\begin{aligned}
\xd\bfU=\bf0\qquad&\text{on}\qquad\calR,
\\[4pt]
\bfU\bfn=\bf0\qquad&\text{on}\qquad\partial\calR,
\\[4pt]
\end{aligned}
\mskip3mu\right\}
\label{admiss1}
\ee
and
\be
\int_{\calR}\bfU\cdot\mfC[\bfS]\dv=0.
\label{admiss2}
\ee
Preliminary to deriving the consequences of stipulating that the first variation $\delta E$ of $E$ be stationary, we establish some useful corollaries of \eqref{admiss1}--\eqref{admiss2}. To begin, let $\bfmu$ be a differentiable vector field on $\calR$. Then, by \eqref{admiss1} and the divergence theorem,
\begin{align}
\int_{\calR}(\sym\xg\bfmu)\cdot\bfU\dv
&=\int_{\calR}\xg\bfmu\cdot\bfU\dv
\notag\\[4pt]
&=\int_{\calR}(\xd(\bfU^{\trans}\mskip-1mu\bfmu)-\bfmu\cdot\xd\bfU)\dv
\notag\\[4pt]
&=\int_{\partial\calR}\bfmu\cdot\bfU\bfn\da
\notag\\[4pt]
&=0.
\label{corr1}
\end{align}
Next, by \eqref{admiss1}$_2$ and the symmetry of $\bfU$, there exists a symmetric tensor field $\bfB$ defined on $\partial\calR$ such that
\be
\bfU=\bfP\mskip-2mu\bfB\mskip-1mu\bfP
\qquad\text{on}\qquad\partial\calR,
\label{corr2}
\ee
where $\bfP$, as defined by
\be
\bfP=\idem-\bfn\otimes\bfn
\qquad\text{on}\qquad\partial\calR,
\label{P}
\ee
is the perpendicular projector on $\partial\calR$, $\idem$ being the second-order identity tensor.

On varying \eqref{E}, applying the divergence theorem, and invoking \eqref{Asymmetries}$_1$, \eqref{admiss2}, \eqref{corr1}, and \eqref{corr2}, the requirement that $\delta E=0$ yields
\be
\int_{\calR}(\xg(\osix{A}\mskip2mu[\xg\bfS])[\idem]+\lambda\mfC[\bfS]-\text{sym}\mskip2mu\xg\bfmu)\cdot\bfU\dv
=-\int_{\partial\calR}((\osix{A}\mskip2mu[\xg\bfS])\bfn)\cdot\bfU\da,
\label{deltaE}
\ee
where $\lambda$ is a constant scalar. Next, on choosing $\bfU$ to be compactly supported about a point interior to $\calR$,  \eqref{deltaE} yields the Euler--Lagrange equation
\be
-\xg(\osix{A}\mskip2mu[\xg\bfS])[\idem]+\text{sym}\mskip2mu\xg\bfmu=\lambda\mfC[\bfS]
\qquad\text{on}\qquad\calR.
\label{EL0}
\ee
Furthermore, in view of \eqref{EL0}, the stationarity condition \eqref{deltaE} reduces to
\be
\int_{\partial\calR}((\osix{A}\mskip2mu[\xg\bfS])\bfn)\cdot\bfU\da=0.
\label{deltaEreduced}
\ee
On invoking the representation \eqref{corr2} on $\partial\calR$ of $\bfU$ and using the symmetry of $\bfP$, we obtain
\be
\int_{\partial\calR}(\bfP((\osix{A}\mskip2mu[\xg\bfS])\bfn)\bfP)\cdot\bfB\da
=0.
\label{intnbc}
\ee
On choosing $\bfB$ to be compactly supported about a point on $\partial\calR$, \eqref{intnbc} yields the natural boundary condition
\be
\bfP((\osix{A}\mskip2mu[\xg\bfS])\bfn)\bfP=\bfO,
\label{nbc}
\ee
where $\bfO$ denotes the zero second-order tensor. Thus, in the sense of \citet{gurtin2002interface}, the tangential component of the tensor field $(\osix{A}\mskip2mu[\xg\bfS])\bfn$ defined on $\partial\calR$ vanishes.

We next wish to determine the constant Lagrange multiplier $\lambda$. Toward that end, we compute the scalar product of \eqref{EL0} with $\bfS$, integrate the resulting identity over $\calR$, and invoke \eqref{corr1} to find that
\be
-\int_{\calR}(\xg(\osix{A}\mskip2mu[\xg\bfS])[\idem])\cdot\bfS\dv=\lambda\int_{\calR}\bfS\cdot\mfC[\bfS]\dv.
\label{step1}
\ee
Applying the divergence theorem and invoking \eqref{normalization}, it follows, from \eqref{step1}, that
\be
-\int_{\partial\calR}((\osix{A}\mskip2mu[\xg\bfS])\bfn)\cdot\bfS\da+\int_{\calR}\xg\bfS\cdot\osix{A}\mskip2mu[\xg\bfS]\dv=\lambda\varsigma\text{vol}(\calR).
\label{step2}
\ee
With reference to \eqref{intro}$_2$ and \eqref{nbc}, the boundary term in \eqref{step2} drops out, yielding
\be
\int_{\calR}\xg\bfS\cdot\osix{A}\mskip2mu[\xg\bfS]\dv=\lambda\varsigma\text{vol}(\calR)
\label{step3}
\ee
and, thus, we find that $\lambda$ is given in terms of $\bfS$ by
\be
\lambda=\frac{1}{\varsigma\text{vol}(\calR)}\int_{\calR}\xg\bfS\cdot\osix{A}\mskip2mu[\xg\bfS]\dv.
\label{Gamma}
\ee
On using \eqref{Gamma} in \eqref{EL0} and invoking the definition \eqref{E} of $E$, the Euler--Lagrange equation \eqref{EL0} takes the form
\be
-\xg(\osix{A}\mskip2mu[\xg\bfS])[\idem]+\text{sym}\mskip2mu\xg\bfmu=\frac{2E}{\varsigma\text{vol}(\calR)}\bfS.
\label{EL}
\ee

Elimination of $\lambda$ leaves only the variables $\bfS$ and $\bfmu$ in the boundary-value problem consisting of the second-order partial-differential equation \eqref{EL}, the constraint embodied by the first-order partial-differential equation \eqref{intro}$_1$, the boundary conditions \eqref{intro}$_2$ and \eqref{nbc}, and the normalization condition \eqref{normalization}. Granted that $\bfS$ is known, it can be used in \eqref{E} to obtain $E$.

Incidentally, it is straightforward to see that if the pair $\bfS$ and $\bfmu$ is a solution to the constrained boundary-value problem described above, then the pair $-\bfS$ and $-\bfmu$ is also a solution.

\subsection{Dimensionless boundary-value problem} \label{sec:nondim}
If we stipulate that $E$, introduced in \eqref{E}, has physical dimensions 
\be
[E]=\frac{\textsf{M}\textsf{L}^2}{\textsf{T}^2},
\ee
the coefficient tensor $\osix{A}$ must accordingly have physical dimensions
\be
[\osix{A}]=\frac{\textsf{L}^3\textsf{T}^2}{\textsf{M}}.
\label{dimA}
\ee
Similarly, on stipulating that $\int_{\calR}\bfS\cdot\mfC[\bfS]\dv$ has the same physical dimensions as $E$, it follows that $\varsigma$ has dimensions of force per unit area or, equivalently mass per unit length per unit time squared, that is:
\be
[\varsigma]=\frac{\textsf{M}}{\textsf{L}\textsf{T}^2}.
\label{varsigmadim}
\ee
Moreover, from \eqref{normalization} and \eqref{varsigmadim}, we find that the coefficient tensor $\mfC$ must have physical dimensions
\be
[\mfC]=\frac{\textsf{L}\textsf{T}^2}{\textsf{M}}.
\label{dimC}
\ee
Furthermore, from the Euler--Lagrange equation \eqref{EL0}, it follows the physical dimensions of $\lambda$ and $\bfmu$ must be
\be
[\lambda]=\frac{[\osix{A}][\bfS]}{\textsf{L}^2}=1 \qquad \text{and} \qquad [\bfmu]=\frac{\textsf{M}[\osix{A}]}{\textsf{L}^2\textsf{T}^2}=\textsf{L}.
\label{dim}
\ee
Based on \eqref{dimA}--\eqref{dim}, we introduce dimensionless counterparts $\bfx^*$, $\bfS^*(\bfx^*)$, $\osix{A}^*(\bfx^*)$, $\mfC^*(\bfx^*)$, and $\bfmu^*(\bfx^*)$ of $\bfx$, $\bfS(\bfx)$, $\osix{A}(\bfx)$, $\mfC(\bfx)$, and $\bfmu(\bfx)$ through
\be
\left.
\begin{gathered}
\bfx^*=\frac{\bfx}{L}, \qquad \bfS^*(\bfx^*)=\frac{\bfS(\bfx)}{\varsigma},\qquad \osix{A}^*(\bfx^*)=\frac{\osix{A}(\bfx)\varsigma}{L^2}, 
\\[4pt]
\mfC^*(\bfx^*)=\mfC(\bfx)\varsigma, \qquad \text{and} \qquad \bfmu^*(\bfx^*)=\frac{\bfmu(\bfx)}{L}, 
\end{gathered}\mskip4mu\right\}
\label{change}
\ee
where we take $L$ to be given by
\be
L=(k_0\text{vol}\mskip2mu(\calR))^{\frac{1}{3}}, 
\label{Ldef}
\ee
with $k_0$ being a prescribed dimensionless constant. Letting $\xg^*$ denote the dimensionless counterpart of the gradient arising from the change \eqref{change}$_1$ of independent variable, we thus obtain a dimensionless version, 
\be
-\xg^*(\osix{A}^*\xg^*\bfS^*)[\idem]+\text{sym}\mskip2mu\xg^*\bfmu^*=\lambda\mfC^*[\bfS^*],
\label{ELdimlessbeta}
\ee
of the Euler--Lagrange equation \eqref{EL0}. Moreover, with the introduction of 
\be
\dv^*=\frac{\dv}{L^3} \qquad \text{and} \qquad E^*=\frac{E}{L^3\varsigma},
\ee
\eqref{E} transforms to 
\be
E^*(\bfS^*)=\frac{1}{2}\int_{\calR^*} \xg^*\bfS^*\cdot\osix{A}^*[\xg^*\bfS^*]\dv^*.
\label{Enondimast}
\ee
Finally, the normalization condition of \eqref{normalization} takes the form
\be
\int_{\calR^*}\bfS^*\cdot\mfC^*[\bfS^*]\dv^*=1.
\label{normnondim}
\ee

We hereinafter consider only the dimensionless entities. To reduce clutter, we drop the asterisks accompanying the various quantities and operators so that the constrained dimensionless boundary-value problem consists of
\drop{
\begin{subequations}
\begin{alignat}{2}
    \left.\begin{aligned}
        -\xg(\osix{A}\mskip2mu[\xg\bfS])[\idem]+\text{sym}\mskip2mu\xg\bfmu&=\lambda\cb{\mfC}\bfS  \\[4pt]
    \xd\bfS&=\bf0
    \label{bvp_nondim_calR}
    \end{aligned}\mskip3mu\right\}\text{on}\qquad \calR, &\\[4pt]
 \left.\begin{aligned}
        \bfS\bfn&=\bf0  \\[4pt]
    \bfP((\osix{A}\mskip2mu[\xg\bfS])\bfn)\bfP&=\bf0
    \label{bvp_nondim_partialcalR}
    \end{aligned}\mskip3mu\right\}\mskip30mu\text{on}\mskip25mu \partial\calR,&\\[4pt]
    \int_{\calR}\cb{(\mfC[\bfS])}\cdot\bfS\dv=1.\mskip120mu &
    \label{bvp_nondim_norm}
 \end{alignat}
 \label{bvp_nondim}
 \end{subequations}
 }
\begin{equation}
\left.
\begin{aligned}
\left.
\begin{aligned}
-\xg(\osix{A}\mskip2mu[\xg\bfS])[\idem]+\sym\xg\bfmu
&=\lambda\boldsymbol{\mathsfit C}[\bfS],
\\[4pt]
\xd\bfS&=\bf0,
\end{aligned}
\mskip4mu\right\}
&\quad\text{on}\quad\calR,
\\[4pt]
\left.
\begin{aligned}
\bfS\bfn&=\bf0,
\\[4pt]
\bfP((\osix{A}\mskip2mu[\xg\bfS])\bfn)\bfP&=\bfO,
\end{aligned}
\mskip4mu\right\}
&\quad\text{on}\quad\partial\calR,
\\[4pt]
\int_{\calR}\bfS\cdot\boldsymbol{\mathsfit C}[\bfS]\dv=1.
\hspace{40pt}&
\end{aligned}
\mskip4mu\right\}
\label{bvp_nondim}
\end{equation} 
The conditions in \eqref{bvp_nondim} are the Euler--Lagrange equation \eqref{bvp_nondim}$_1$, the equilibrium condition \eqref{bvp_nondim}$_2$, the traction boundary condition \eqref{bvp_nondim}$_3$, the natural boundary condition \eqref{bvp_nondim}$_4$, and the normalization condition \eqref{bvp_nondim}$_5$. Additionally, from \eqref{Enondimast}, the dimensionless functional underlying the Euler--Lagrange equation \eqref{bvp_nondim}$_1$ and the natural boundary condition \eqref{bvp_nondim}$_4$ is
\be
E(\bfS)=\frac{1}{2}\int_{\calR} \xg\bfS\cdot\osix{A}\mskip2mu[\xg\bfS]\dv.
\label{E_nondim}
\ee

We mention some noteworthy points from the boundary-value problem \eqref{bvp_nondim}. The Euler--Lagrange equation \eqref{bvp_nondim}$_1$ and the equilibrium condition \eqref{bvp_nondim}$_2$ yield six and three scalar differential equations, respectively, while the traction boundary condition \eqref{bvp_nondim}$_3$ and the natural boundary condition \eqref{bvp_nondim}$_4$ yield three scalar boundary conditions each. The unknowns entering \eqref{bvp_nondim} are the stress $\bfS$ (six scalar fields), the multiplier $\bfmu$ (three scalar fields) which may be interpreted as the reaction to the constraint \eqref{bvp_nondim}$_2$, and the constant scalar multiplier $\lambda$, which may be interpreted as the reaction to the normalization condition \eqref{bvp_nondim}$_5$. Thus, although \eqref{bvp_nondim}$_{1,2}$ constitute nine scalar conditions and \eqref{bvp_nondim}$_{3,4}$ constitute only six scalar conditions, the three additional conditions associated with the need to determine $\bfmu$ are supplied by \eqref{bvp_nondim}$_{1,2}$. In this regard, it is important to recognize that a boundary condition for $\bfmu$ is not needed. The information needed to determine $\lambda$ is supplied by the normalization condition \eqref{bvp_nondim}$_5$, which, on the face of it, appears to be the only nonlinear condition in the system \eqref{bvp_nondim}. However, noting the dependence \eqref{Gamma} of $\lambda$ on $\bfS$, we conclude that the Euler--Lagrange equation \eqref{bvp_nondim}$_1$ is also non-linear. Additionally, we see that $\bfmu$ appears only through its gradient and, thus, can generally be determined only up to an additive constant. Finally, $\lambda$ is related to the dimensionless functional $E$, from \eqref{E}, \eqref{Gamma}, and \eqref{E_nondim}, by
\be
\lambda=2E.
\label{lambda2E}
\ee
We thus conclude from the positive definiteness of $E$ that
\be
\lambda>0.
\label{lambda>0}
\ee

For a given choice of $\osix{A}$ and $\mfC$ that are viable in the sense that they satisfy \eqref{Asymmetries}--\eqref{Aposdef}, the boundary-value problem \eqref{bvp_nondim} generally admits more than one solution. Since $\lambda$ is positive by \eqref{lambda>0}, we can readily arrange the stress solutions and Lagrange multiplier fields corresponding to given $\osix{A}$ and $\mfC$ in sequences $(\bfS_N)$ and $(\bfmu_N)$, respectively, ordered by increasing $(\lambda_N)$. For subsequent reference, we denote the $N^{\text{th}}$ stress solution, henceforth referred to as `extremizer', and Lagrange multiplier field as $\bfS_N$ and $\bfmu_N$, respectively. At this point, the index set from which $N$ derives its values remains ambiguous; specifically, it is uncertain whether the index set encompasses the entirety of the natural numbers $\mathbb{N}$ or constitutes a subset thereof. We will see in Subsection \ref{sec:inf} that the index set is $\mathbb{N}$.

\section{Properties of the extremizers} \label{sec:props}
Motivated by the functional $E$ in \eqref{E_nondim} and the normalization condition \eqref{bvp_nondim}$_5$, we introduce the functionals $\langle\cdot,\cdot\rangle_{\sbbA}:\calS\times\calS\to\mathbb{R}$ and $\langle\cdot,\cdot\rangle_{\sbbC}:\calS\times\calS\to\mathbb{R}$ defined by
\be
\langle\bfT_1,\bfT_2\rangle_{\sbbA}=\int_{\calR}\xg\bfT_2\cdot\osix{A}\mskip2mu[\xg\bfT_1]\dv \qquad \text{and} \qquad \langle\bfT_1,\bfT_2\rangle_{\sbbC}=\int_{\calR}\bfT_2\cdot\mfC[\bfT_1]\dv 
\label{ACscalar}
\ee
for all $\bfT_1$ and $\bfT_2$ in $\calS$, where the set $\calS$ is defined in \eqref{setS}. By the symmetry, bilinearity, and positive definiteness of $\osix{A}$ from \eqref{Asymmetries}$_1$, \eqref{ACscalar}$_1$, and \eqref{Aposdef}$_1$, respectively, $\langle\cdot,\cdot\rangle_{\sbbA}$ is a scalar product since, for any scalars $\eta_1,\eta_2$, and any $\bfT_1,\bfT_2,\bfT_3$ in $\calS$,
\be
\left.
\ba
&\langle\bfT_1,\bfT_2\rangle_{\sbbA}=\langle\bfT_2,\bfT_1\rangle_{\sbbA},
\\[4pt]
&\langle\eta_1\mskip1mu\bfT_1+\eta_2\mskip1mu\bfT_2,\bfT_3\rangle_{\sbbA}=\eta_1\mskip1mu\langle\bfT_1,\bfT_3\rangle_{\sbbA}+\eta_2\mskip1mu\langle\bfT_2,\bfT_3\rangle_{\sbbA},
\\[4pt]
&\langle\bfT_1,\bfT_1\rangle_{\sbbA}>0\qquad\text{if}\qquad\bfT_1\neq\bfO.
\ea\mskip3mu\right\}
\ee
Similarly, by the symmetry, bilinearity, and positive definiteness of $\mfC$ from \eqref{Csymmetries}$_1$, \eqref{ACscalar}$_2$, and \eqref{Aposdef}$_2$, respectively, $\langle\cdot,\cdot\rangle_{\sbbC}$ is a scalar product since, for any scalars $\eta_1,\eta_2$, and any $\bfT_1,\bfT_2,\bfT_3$ in $\calS$,
\be
\left.
\ba
&\langle\bfT_1,\bfT_2\rangle_{\sbbC}=\langle\bfT_2,\bfT_1\rangle_{\sbbC},
\\[4pt]
&\langle\eta_1\mskip1mu\bfT_1+\eta_2\mskip1mu\bfT_2,\bfT_3\rangle_{\sbbC}=\eta_1\mskip1mu\langle\bfT_1,\bfT_3\rangle_{\sbbC}+\eta_2\mskip1mu\langle\bfT_2,\bfT_3\rangle_{\sbbC},
\\[4pt]
&\langle\bfT_1,\bfT_1\rangle_{\sbbC}>0\qquad\text{if}\qquad\bfT_1\neq\bfO.
\ea\mskip3mu\right\}
\ee
The norms derived from the scalar products \eqref{ACscalar}$_1$ and \eqref{ACscalar}$_2$ are denoted as $\|\cdot\|_{\sbbA}$ and $\|\cdot\|_{\sbbC}$, respectively, and are defined by
\be
\|\bfT\|_{\sbbA}=\langle\bfT,\bfT\rangle_{\sbbA}^{\frac{1}{2}} \qquad \text{and} \qquad \|\bfT\|_{\sbbC}=\langle\bfT,\bfT\rangle_{\sbbC}^{\frac{1}{2}}
\label{normsAC}
\ee
for any $\bfT$ in $\calS$. We show in Propositions \ref{hatEH1} and \ref{CL2} in \ref{app:equiv} that on the set $\calS$, the norm $\|\cdot\|_{\sbbA}$ is equivalent to the $H^1(\calR)$ norm and the norm $\|\cdot\|_{\sbbC}$ is equivalent to the $L^2(\calR)$ norm, respectively.

Recalling the definition \eqref{setS} of $\calS$, we denote its completion in the $L^2(\calR)$ norm as $\bar{\calS}$. By the equivalence of the $L^2(\calR)$ norm with the norm $\|\cdot\|_{\sbbC}$ from Proposition \ref{CL2}, the completion of $\calS$ in the norm $\|\cdot\|_{\sbbC}$ is also equal to $\bar{\calS}$. Notice, in particular, that discontinuous residual stress fields are elements of $\bar{\calS}$, but are not elements of $\calS$.

We next wish to establish that the sequence $(\bfS_N)$ of extremizers forms a basis in the norm $\|\cdot\|_{\sbbC}$ for $\bar{\calS}$ and in the norm $\|\cdot\|_{\sbbA}$ for $\calS$. Toward this objective, we devote Subsections \ref{sec:ortho} through \ref{sec:inf} to deriving several useful properties of the extremizers. Then, in Subsections \ref{sec:span} and \ref{sec:H1span}, we use those properties to establish the sought results .

\subsection{Extremizers are orthonormal with respect to the scalar product induced by \boldmath{$\mfC$}} \label{sec:ortho}
Let $\bfS_p$, $\bfmu_p$, and $\lambda_p$ be a solution triplet of the boundary-value problem described by \eqref{bvp_nondim}, meaning, in particular, that
\be
    -\xg(\osix{A}\mskip2mu[\xg\bfS_p])[\idem]+\text{sym}\mskip2mu\xg\bfmu_p=\lambda_p\mskip1mu\mfC[\bfS_p].
\label{S1S2}
\ee
Let $\bfSigma$ be a residual stress field on $\calR$. Consider the $L^2(\calR)$ scalar product of \eqref{S1S2} with $\bfSigma$:
\be
-\int_{\calR}(\xg(\osix{A}\mskip2mu[\xg\bfS_p])[\idem])\cdot\bfSigma\dv+\int_{\calR}(\text{sym}\mskip2mu\xg\bfmu_p)\cdot\bfSigma\dv=\lambda_p\int_{\calR}\bfSigma\cdot\mfC[\bfS_p]\dv.
\label{ip1}
\ee
By the arguments in $\eqref{corr1}$, the second-term on the left-hand side of \eqref{ip1} vanishes. We use the divergence theorem to re-write the first term on the left-hand side of \eqref{ip1} as follows:
\be
-\int_{\calR}(\xg(\osix{A}\mskip2mu[\xg\bfS_p])[\idem])\cdot\bfSigma\dv
=-\int_{\partial\calR} ((\osix{A}\mskip2mu[\xg\bfS_p])\bfn)\cdot\bfSigma\dv+\int_{\calR}\xg\bfSigma\cdot\osix{A}\mskip2mu[\xg\bfS_p]\dv.
\label{plug1}
\ee
Next, by \eqref{intro}$_2$ and the symmetry of $\bfSigma$, there exists a symmetric tensor field $\bfB$ defined on $\partial\calR$ such that
\be
\bfSigma=\bfP\mskip-2mu\bfB\mskip-1.5mu\bfP
\qquad\text{on}\qquad\partial\calR,
\label{bt3}
\ee
where $\bfP$ is the perpendicular projector on $\partial\calR$, and is given by \eqref{P}. Moreover, from symmetry of $\bfP$ and the natural boundary condition \eqref{bvp_nondim}$_4$, it follows that
\be
((\osix{A}\mskip2mu[\xg\bfS_p])\bfn)\cdot(\bfP\mskip-2mu\bfB\mskip-1.5mu\bfP)=(\bfP((\osix{A}\mskip2mu[\xg\bfS_p])\bfn)\bfP)\cdot\bfB=0.
\label{Pn}
\ee
Thus, the integral over $\partial\calR$ in \eqref{plug1} vanishes and \eqref{ip1} reduces to 
\be
\int_{\calR}\xg\bfSigma\cdot\osix{A}\mskip2mu[\xg\bfS_p]\dv=\lambda_p\int_{\calR}\bfSigma\cdot\mfC[\bfS_p]\dv.
\label{bt4a}
\ee

Let $\bfS_q$, $\bfmu_q$, and $\lambda_q$ be another solution triplet of the boundary-value problem described by \eqref{bvp_nondim}, with $q\neq p$. Since $\bfS_q$ is a residual stress field, it follows from \eqref{bt4a} that
\be
\int_{\calR}\xg\bfS_q\cdot\osix{A}\mskip2mu[\xg\bfS_p]\dv=\lambda_p\int_{\calR}\bfS_q\cdot\mfC[\bfS_p]\dv.
\label{bt4}
\ee
Repeating the arguments from \eqref{S1S2} through \eqref{bt4} with $p$ and $q$ interchanged yields
\be
\int_{\calR}\xg\bfS_p\cdot\osix{A}\mskip2mu[\xg\bfS_q]\dv=\lambda_q\int_{\calR}\bfS_p\cdot\mfC[\bfS_q]\dv.
\label{bt5}
\ee
By the assumptions \eqref{Asymmetries}$_{1,2}$, \eqref{bt5} can be written as
\be
\int_{\calR}\xg\bfS_q\cdot\osix{A}\mskip2mu[\xg\bfS_p]\dv=\lambda_q\int_{\calR}\bfS_q\cdot\mfC[\bfS_p]\dv.
\label{bt6}
\ee
Subtracting \eqref{bt6} from \eqref{bt4} yields
\be
(\lambda_p-\lambda_q)\int_{\calR}\bfS_q\cdot\mfC[\bfS_p]\dv=0.
\ee

Consider, first, the case $\lambda_p\neq\lambda_q$, whence
\be
\int_{\calR}\bfS_q\cdot\mfC[\bfS_p]\dv=0.
\label{L2orth}
\ee
Thus, $\bfS_p$ and $\bfS_q$ are orthogonal with respect to the scalar product $\langle\cdot,\cdot\rangle_{\sbbC}$ defined by \eqref{ACscalar}$_2$. By \eqref{bvp_nondim}$_5$, we find, furthermore, that they are {\em orthonormal} with respect to the scalar product $\langle\cdot,\cdot\rangle_{\sbbC}$.

Next, consider the case $\lambda_p=\lambda_q$ but $\bfS_p\neq\bfS_q$. Accordingly, the relation \eqref{L2orth} obtained from the consideration $\lambda_p\neq\lambda_q$ need not hold. However, it is straightforward to check that for real scalars $\eta_p$ and $\eta_q$, $\eta_p\bfS_p+\eta_q\bfS_q$ is also a solution of the boundary-value problem \eqref{bvp_nondim}, with the corresponding Lagrange multiplier field being $\eta_p\bfmu_p+\eta_q\bfmu_q$, as long as $\eta_p$ and $\eta_q$ satisfy the condition
\be
\eta_p^2+\eta_q^2+2\eta_p\eta_q\int_{\calR}\bfS_q\cdot\mfC[\bfS_p]\dv=1
\label{ensure}
\ee
which ensures that $\eta_p\bfS_p+\eta_q\bfS_q$ has unit $\|\cdot\|_{\sbbC}$ norm. In other words, any element in the subset of $\calS$ spanned by $\bfS_p$ and $\bfS_q$ with unit $\|\cdot\|_{\sbbC}$ norm is a solution. Thus, $\bfS_p$ and $\bfS_q$ can be {\em chosen} such that they are orthonormal in the scalar product induced by $\mfC$. 

Finally, if $\lambda_p=\lambda_q$ and $\bfS_p=\bfS_q$ but $\bfmu_p\neq\bfmu_q$ then, from \eqref{bvp_nondim}$_1$, 
\be
\sym\xg\bfmu_p=\sym\xg\bfmu_q,
\ee
and there is no distinction between these two cases. 

In summary, the extremizers corresponding to given viable coefficient tensors $\osix{A}$ and $\mfC$ are mutually orthonormal with respect to the scalar product $\langle\cdot,\cdot\rangle_{\sbbC}$ induced by $\mfC$.

\subsection{Extremizers are orthonormal with respect to the scalar product induced by \boldmath{$\osix{A}$}} \label{sec:varphi}
On invoking \eqref{L2orth}, the identity \eqref{bt6} for $p\neq q$ reduces to
\be
\int_{\calR}\xg\bfS_q\cdot\osix{A}\mskip2mu[\xg\bfS_p]\dv=0
\label{red}
\ee
or, by the definition \eqref{ACscalar}$_1$ of the scalar product $\langle\cdot,\cdot\rangle_{\sbbA}$, to
\be
\langle\bfS_p,\bfS_q\rangle_{\sbbA}=0.
\label{ipE}
\ee
Thus, the extremizers are not only orthonormal with respect to the scalar product $\langle\cdot,\cdot\rangle_{\sbbC}$ induced by $\mfC$ but also orthogonal in the scalar product $\langle\cdot,\cdot\rangle_{\sbbA}$ induced by $\osix{A}$. 

\subsection{Extremizers are infinitely many} \label{sec:inf}
We next establish that for given $\osix{A}$ and $\mfC$ satisfying \eqref{Asymmetries}--\eqref{Aposdef}, there are infinitely many linearly independent extremizers. We argue by contradiction. Assume that there only $N_0$ linearly independent extremizers $\bfS_N, N=1,2,\dots,N_0$. For brevity, we introduce the set 
\be
\mathbb{N}_{N_0}=\{m\in\mathbb{N}:m=1,2,\dots,N_0\} 
\label{NN0}
\ee
of all natural numbers from 1 through $N_0$.

Recalling the set $\calS$ defined by \eqref{setS}, consider the non-dimensional problem of extremizing $E$, defined by \eqref{E_nondim},
over the set 
\be
\calS_{N_0\scriptscriptstyle{\perp}}=\left\{\bfT:\bfT\in\calS,\mskip2mu\int_{\calR}\bfS_N\cdot\mfC[\bfT]\dv=0,\mskip1mu N\in\mathbb{N}_{N_0}\right\},
\label{SN_orth}
\ee
subject to the condition
\be
\int_{\calR}\bfS\cdot\mfC[\bfS]\dv=1.
\label{bvp_orth}
\ee
We show in \ref{app:existence} that a solution to the above extremization problem exists.

To find the equations satisfied by an extremizer $\bfS$ of the above problem, we notice that an admissible variation $\bfU$ of $\bfS$ must, in addition to being symmetric and complying with \eqref{admiss1}--\eqref{corr2}, satisfy 
\be
\int_{\calR}\bfU\cdot\mfC[\bfS_N]\dv=0, \qquad N\in\mathbb{N}_{N_0}.
\label{orthpN}
\ee
Following the procedure detailed in Section \ref{sec:bvp}, we find, after incorporating \eqref{orthpN}, that $\bfS$ satisfies the Euler--Lagrange equation 
\be
-\xg(\osix{A}\mskip2mu[\xg\bfS])[\idem]+\text{sym}\mskip2mu\xg\bfmu=\lambda\mfC[\bfS]+\sum_{N=1}^{N_0} \nu_N\mskip1mu\mfC[\bfS_N],
\label{EL_orth}
\ee
where $\nu_N, N\in\mathbb{N}_{N_0}$, are constant scalars. Moreover, since the additional constraints 
\be
\int_{\calR}\bfS_N\cdot\mfC[\bfS]\dv=0, \qquad N\in\mathbb{N}_{N_0},
\label{SSN}
\ee
have no bearing on the boundary conditions, $\bfS$ satisfies the natural boundary condition \eqref{bvp_nondim}$_4$.

We compute the scalar product of \eqref{EL_orth} with a given $\bfS_p$, $p\in\mathbb{N}_{N_0}$, and integrate the resulting identity over $\calR$ to obtain
\be
-\int_{\calR}(\xg(\osix{A}\mskip2mu[\xg\bfS])[\idem])\cdot\bfS_p\dv+\int_{\calR}(\text{sym}\mskip2mu\xg\bfmu)\cdot\bfS_p\dv=\lambda\int_{\calR}\bfS_p\cdot\mfC[\bfS]\dv+\sum_{N=1}^{N_0}\nu_N\int_{\calR}\bfS_p\cdot\mfC[\bfS_N]\dv.
\label{EL_orthint}
\ee
By the divergence theorem, the second term on the left-hand side of \eqref{EL_orthint} drops out. By \eqref{SSN}, the first term on the right-hand side of \eqref{EL_orthint} drops out. Furthermore, by \eqref{L2orth} and \eqref{bvp_nondim}$_5$, the second term on the right-hand side of \eqref{EL_orthint} contributes only $\nu_p$. By an application of the divergence theorem on its first term on the left-hand side, \eqref{EL_orthint} then becomes
\be
-\int_{\partial\calR} ((\osix{A}\mskip2mu[\xg\bfS])\bfn)\cdot\bfS_p\da +\int_{\calR}\xg\bfS_p\cdot\osix{A}\mskip2mu[\xg\bfS]\dv=\nu_p.
\label{infimany2}
\ee
Since $\bfS$ satisfies the same natural boundary condition as $\bfS_p$, the arguments from \eqref{bt3} through \eqref{bt4a}, which led to the realization that the integral over $\partial\calR$ vanishes, still hold, so that \eqref{infimany2} reduces to
\be
\int_{\calR}\xg\bfS_p\cdot\osix{A}\mskip2mu[\xg\bfS]\dv=\nu_p.
\label{infimany3a}
\ee
By the assumption \eqref{Asymmetries}$_1$, \eqref{infimany3a} can be equivalently written as 
\be
\int_{\calR}\xg\bfS\cdot\osix{A}\mskip2mu[\xg\bfS_p]\dv=\nu_p.
\label{infimany3}
\ee
Since $\bfS$ is a residual stress field, \eqref{bt4a}, with $\bfSigma=\bfS$, applies, thence yielding
\be
\nu_p=\int_{\calR}\bfS\cdot\mfC[\bfS_p]\dv.
\ee
Upon invoking the symmetry \eqref{Csymmetries}$_1$ of $\mfC$ and the relation \eqref{SSN}, we find that
\be
\nu_p=0.
\label{nup1}
\ee
Since $p$ is arbitrary, we conclude that $\nu_N=0, N\in\mathbb{N}_{N_0}$. Inserting these zeros in \eqref{EL_orth}, we find that $\bfS$ satisfies the same equations as those satisfied by $\bfS_N,N\in\mathbb{N}_{N_0}$, and additionally satisfies \eqref{SSN}.
This implies that $\bfS$ must either be equal to one of the $\bfS_N,N\in\mathbb{N}_{N_0}$, or, a linear combination thereof with unit $\|\cdot\|_{\sbbC}$ norm; moreover, it is orthogonal to each $\bfS_N,N\in\mathbb{N}_{N_0}$, with respect to the scalar product $\langle\cdot,\cdot\rangle_{\sbbC}$. The provisional assumption that there are only finitely many extremizers $N_0$ is thus contradictory, and we conclude that there are infinitely many extremizers $\bfS_N, N\in\mathbb{N}$. 

In the next subsection, we prove that $\bfS_N, N\in\mathbb{N}$, span $\bar{\calS}$.

\subsection{Extremizers form an orthonormal basis for \boldmath{$\bar{\calS}$} in the norm induced by \boldmath{$\mfC$}}\label{sec:span}
Again, we argue by contradiction. Assume that the set
\be
\tilde{\calS}_{\scriptscriptstyle{\perp}}=\left\{\bfT:\bfT\in\calS,\mskip2mu\int_{\calR}\bfS_N\cdot\mfC[\bfT]\dv=0,\mskip1mu N\in\mathbb{N}\right\},
\label{Shat_orth}
\ee
where $\calS$ is defined by \eqref{setS}, is non-empty. Consider the problem of extremizing the functional $E$, defined by \eqref{E_nondim}, over $\tilde{\calS}_{\scriptscriptstyle{\perp}}$, subject to the condition
\be
\int_{\calR}\bfS\cdot\mfC[\bfS]\dv=1.
\label{bvp_orthinf}
\ee
An admissible variation $\bfU$ of an extremizer $\bfS$ must, in addition to being symmetric and complying with \eqref{admiss1}--\eqref{corr2}, satisfy 
\be
\int_{\calR}\bfU\cdot\mfC[\bfS_N]\dv=0, \qquad N\in\mathbb{N}.
\label{orthpinf}
\ee
Arguing now in the same spirit as that in Subsection \ref{sec:inf}, we find that $\bfS$, guaranteed to exist by the arguments in \ref{app:existence}, lies in the span of $\bfS_N, N\in\mathbb{N}$. Moreover, since $\bfS$ belongs to the set $\tilde{\calS}_{\scriptscriptstyle{\perp}}$ defined in \eqref{Shat_orth}, it is orthogonal to each of the $\bfS_N, N\in\mathbb{N}$, with respect to the scalar product $\langle\cdot,\cdot\rangle_{\sbbC}$. Our assumption that $\tilde{\calS}_{\scriptscriptstyle{\perp}}$ is non-empty is thus contradictory and we conclude that $\tilde{\calS}_{\scriptscriptstyle{\perp}}$ is empty. This implies that there is no element in $\calS$ that is orthogonal to each $\bfS_N, N\in\mathbb{N}$, with respect to the scalar product $\langle\cdot,\cdot\rangle_{\sbbC}$. Thus, $\bfS_N, N\in\mathbb{N}$ span $\calS$. Finally, since each element of $\bar{\calS}$ is arbitrarily close in the norm $\|\cdot\|_{\sbbC}$ to some element of $\calS$, we conclude that $\bfS_N, N\in\mathbb{N}$, form an orthonormal basis for $\bar{\calS}$ with respect to the norm $\|\cdot\|_{\sbbC}$. 

Since, by Proposition \ref{CL2}, the norm $\|\cdot\|_{\sbbC}$ is equivalent to the $L^2(\calR)$ norm on $\calS$, we obtain, furthermore, the corollary that $\bfS_N, N\in\mathbb{N}$, form a basis for $\bar{\calS}$ with respect to the $L^2(\calR)$ norm.

As mentioned at the beginning of Section \ref{sec:bvp}, discontinuous residual stress fields are elements of $\bar{\calS}$ but are not elements of $\calS$. The foregoing result shows that even a discontinuous residual stress field can be represented as a linear combination of $\bfS_N, N\in\mathbb{N}$. This feature of our theory is illustrated in Subsection \ref{sec:shrink}, where we consider a shrink-fit residual stress field.

\subsection{Extremizers form an orthogonal basis for \boldmath{$\calS$} in the norm induced by \boldmath{$\osix{A}$}}\label{sec:H1span}
Notice, first, that the norm $\|\cdot\|_{\sbbA}$, given by \eqref{normsAC}$_1$, is equivalent to the $H^1(\calR)$ norm on $\calS$ (see Proposition \ref{hatEH1} in \ref{app:equiv}). Moreover, the mappings 
\be
\bfT\to\xd\bfT \qquad \text{on} \qquad \calR
\qquad\qquad\text{and}\qquad\qquad
\bfT\to\bfT\bfn \quad \qquad \text{on} \qquad \partial\calR
\ee
are continuous in the $H^1(\calR)$ norm. As a result, $\calS$ is a complete space when equipped with the norm $\|\cdot\|_{\sbbA}$. 

We prove the claim stated in the heading of this subsection by contradiction. Assume that there is a non-zero element $\tilde{\bfS}$ in $\calS$ such that 
\be
\langle\tilde{\bfS},\bfS_N\rangle_{\sbbA}=0, \qquad N\in\mathbb{N}.
\label{assume}
\ee
Invoking the definitions \eqref{ACscalar}$_{1,2}$ of the scalar products $\langle\cdot,\cdot\rangle_{\sbbA}$ and $\langle\cdot,\cdot\rangle_{\sbbC}$ and \eqref{bt4a}, it follows that
\be
\langle\tilde{\bfS},\bfS_N\rangle_{\sbbA}=\lambda_N\langle\tilde{\bfS},\bfS_N\rangle_{\sbbC}, \qquad N\in\mathbb{N},
\ee
which, in conjunction with \eqref{assume}, implies that
\be
\lambda_N\langle\tilde{\bfS},\bfS_N\rangle_{\sbbC}=0, \qquad N\in\mathbb{N}.
\ee
Since $\lambda_N>0, N\in\mathbb{N}$, we find that
\be
\langle\tilde{\bfS},\bfS_N\rangle_{\sbbC}=0, \qquad N\in\mathbb{N}.
\ee
However, as established in Subsection \ref{sec:span}, $\bfS_N, N\in\mathbb{N}$, span $\calS$ in the norm $\|\cdot\|_{\sbbC}$, and we obtain a contradictory conclusion that $\tilde{\bfS}=\bf0$. Thus, $\bfS_N, N\in\mathbb{N}$, form a basis for $\calS$ with respect to the norm $\|\cdot\|_{\sbbA}$. Using this result in conjunction with \eqref{ipE}, we conclude that $\bfS_N, N\in\mathbb{N}$, form an {\em orthogonal} basis for $\calS$ with respect to the norm $\|\cdot\|_{\sbbA}$.

Since the norm $\|\cdot\|_{\sbbA}$ is equivalent to the $H^1(\calR)$ norm on $\calS$ by Proposition \ref{hatEH1}, we obtain, furthermore, the corollary that $\bfS_N, N\in\mathbb{N}$, form a basis for $\calS$ with respect to the $H^1(\calR)$ norm.

In summary, the sequence of extremizers corresponding to given viable coefficient tensors $\osix{A}$ and $\mfC$ forms an orthonormal basis for $\bar{\calS}$ in the norm $\|\cdot\|_{\sbbC}$ induced by $\mfC$. Moreover, it also forms an orthogonal basis for $\calS$ in the norm $\|\cdot\|_{\sbbA}$ induced by $\osix{A}$. This way, we have a family of residual stress bases.

\section{The case of homogeneous, isotropic coefficient tensor}\label{sec:iso}
\citet{olive2013symmetry} showed  that a sixth-order tensor $\osix{A}$ has 17 symmetry classes. Across different symmetry classes, the number of independent components in $\osix{A}$ varies from 5 in the isotropic case to 171 in the totally anisotropic case. In this section, we consider the simplest case of $\osix{A}$ being isotropic. Furthermore, we assume that $\osix{A}$ is homogeneous on $\calR$. Additionally, recognizing that we are dealing with the dimensionless version of the problem, we take the coefficient tensor $\mfC$ to be the fourth-order identity tensor. 
\subsection{Euler--Lagrange equation. Natural boundary condition}
For brevity, we denote the integrand of $E$, given by \eqref{E_nondim}, by
\be
U(\xg\mskip1mu\bfT;\bfx)=\xg\bfT(\bfx)\cdot\osix{A}\mskip2mu(\bfx)[\xg\bfT(\bfx)]
\label{Wfirsttime}
\ee
and we suppress the $\bfx$ dependence henceforth. For isotropic coefficient tensor $\osix{A}$ and symmetric second-order tensor field $\bfT$, a trivial extension of the arguments due to \citet{mindlin1964} shows that $U$ must admit a representation of the form
\begin{multline}
U(\xg\mskip1mu\bfT)=\frac{1}{2}(a_1|\xg(\tr\mskip1mu\bfT)|^2+a_2|\xg\mskip1mu\bfT|^2+a_3\xg\mskip1mu\bfT\cdot (\xg\mskip1mu\bfT)^{\trans}\\
+a_4|\xd\bfT|^2+a_5\xg(\tr\mskip1mu\bfT)\cdot\xd\bfT),
\label{simpli2a}
\end{multline}
where $a_1$ through $a_5$ are scalar coefficients varying with position in $\calR$, and where the transpose $\boldsymbol{\calM}^{\trans}$ of a third-order tensor $\boldsymbol{\calM}$ is defined such that
\be
\boldsymbol{\calM}^{\trans}\cdot(\bfu\otimes\bfv\otimes\bfw)=\boldsymbol{\calM}\cdot(\bfv\otimes\bfw\otimes\bfu)
\label{transpose}
\ee
for all vectors $\bfu$, $\bfv$, and $\bfw$. Since $\xd\bfT=\bf0$ in our case, we set $a_4$ and $a_5$ to zero. Furthermore, to make the dependence of the first term in \eqref{simpli2a} on $\xg\mskip1mu\bfT$ explicit, we define the trace of a third-order tensor $\boldsymbol{\calM}$ such that
\be
\tr\boldsymbol{\calM}\cdot\bfv=(\boldsymbol{\calM}\bfv)\cdot\idem
\label{trdef}
\ee
for all vectors $\bfv$. Then, 
\begin{equation}
(\tr(\xg\mskip1mu\bfT))\cdot\bfv =((\xg\mskip1mu\bfT)\bfv)\cdot\idem
=\xg\mskip1mu\bfT\cdot(\idem\otimes\bfv)
=\xg(\bfT\cdot\idem)\cdot\bfv
=\xg(\tr\bfT)\cdot\bfv.
\label{trgrads}
\end{equation}
Thus, $\xg(\tr\bfT)=\tr(\xg\mskip1mu\bfT)$ and \eqref{simpli2a} simplifies to
\be
U(\xg\mskip1mu\bfT)=\frac{1}{2}(a_1|\tr(\xg\mskip1mu\bfT)|^2+a_2|\xg\mskip1mu\bfT|^2+a_3\xg\mskip1mu\bfT\cdot(\xg\mskip1mu\bfT)^{\trans}).
\label{simpli2}
\ee

We further assume that the coefficient tensor $\osix{A}$ is homogeneous on $\calR$, whereby the coefficients $a_1$ through $a_3$ in \eqref{simpli2} are constants. Moreover, we take the coefficient tensor $\mfC$ to be the fourth-order identity tensor. Then, for $U$ of the form \eqref{simpli2}, the Euler--Lagrange equation \eqref{bvp_nondim}$_1$ takes the form
\be
-a_1(\Delta\tr\bfS)\idem-a_2\Delta\bfS+\text{sym}\mskip2mu\xg\bfmu=\lambda\bfS
\label{EL_iso}
\ee
and the natural boundary condition \eqref{bvp_nondim}$_4$ takes the form
\be
a_1((\xg(\tr\bfS))\cdot\bfn)\bfP+a_2\bfP((\xg\bfS)\bfn)\bfP+a_3 \bfP((\xg\bfS)^{\trans}\bfn)\bfP=\bfO.
\label{natbc_iso}
\ee
Although the parameters $a_1$ and $a_2$ appear in both \eqref{EL_iso} and \eqref{natbc_iso}, $a_3$ appears only in \eqref{natbc_iso}. This is essentially because the condition $\xd\bfT=\bf0$ implies that
\be
   \xg\bfT\cdot(\xg\bfT)^{\trans}=\xd(\bfT(\xg\bfT)^{\trans}).
\label{gradst}
\ee
Consequently, while extremizing the specialization of $E$ obtained by using \eqref{Wfirsttime} in \eqref{E_nondim}, the term with coefficient $a_3$ in \eqref{simpli2} goes over to the boundary upon using the divergence theorem and, thus, is a null Lagrangian.

\subsection{Parameter space for a positive-definite \boldmath{$U$}}
We noted in Section \ref{sec:span}, and \ref{app:existence} referenced therein, that \eqref{Aposdef} or, equivalently, the condition
\be
U(\xg\bfT)>0,
\label{U>0}
\ee
where $\bfT$ is a residual stress field, is sufficient for the existence of extremizers. We now find how this condition translates to those on the parameters $a_1$ through $a_3$ for $\osix{A}$ isotropic. 

To find these conditions, it is helpful to write $U$ as a linear combination of positive-definite terms. Toward that objective, we employ Auffray's decomposition \citep{auffray2013geometrical} of the gradient of a symmetric second-order tensor into two orthogonal parts. The first part $\symm\xg\bfT$ is given by
\be
\symm\xg\bfT=\frac{1}{3}(\xg\bfT+(\xg \bfT)^{\trans}+((\xg\bfT)^{\trans})^{\trans}),
\label{symmdef}
\ee
where $(\xg \bfT)^{\trans}$ and $((\xg\bfT)^{\trans})^{\trans}$ are defined through \eqref{transpose}. Notice that $\symm\xg\bfT$ is a totally-symmetric third-order tensor, that is:
\be
\symm\xg\bfT=(\symm\xg\bfT)^{\trans}=((\symm\xg\bfT)^{\trans})^{\trans}.
\ee
The second, remaining, part $\skk\xg\bfT$ of $\xg\bfT$ is given by
\be
\skk\xg\bfT=\xg\bfT-\symm\xg\bfT.
\label{skkdef}
\ee
Note that we have used the blackboard bold font styles $\symm$ and $\skk$ in \eqref{symmdef} and \eqref{skkdef} to distinguish those transformations from the operators `$\sym$' and `$\sk$' used earlier that deliver the symmetric and skew parts, respectively, of a second-order tensor. 

Using the orthogonality property $\symm(\xg\bfT)\cdot\skk(\xg\bfT)=0$, it follows that
\be
\ba
|\xg\mskip1mu\bfT|^2&=\xg\mskip1mu\bfT\cdot\xg\mskip1mu\bfT\\[4pt]
&=(\symm\xg\mskip1mu\bfT+\mskip1mu\skk\xg\mskip1mu\bfT)\cdot(\symm\xg\mskip1mu\bfT+\mskip1mu\skk\xg\mskip1mu\bfT)\\[4pt]
&=\symm\xg\mskip1mu\bfT\cdot\symm\xg\mskip1mu\bfT+\skk\xg\mskip1mu\bfT\cdot\skk\xg\mskip1mu\bfT\\[4pt]
&=|\symm\xg\mskip1mu\bfT|^2+|\skk\xg\mskip1mu\bfT|^2.
\ea
\label{gtgt}
\ee
Similarly, using $\symm\xg\bfT=(\symm\xg\bfT)^{\trans}$ and the identity 
\be
\boldsymbol{\calM}_1^{\trans}\cdot\boldsymbol{\calM}_2^{\trans}=\boldsymbol{\calM}_1\cdot\boldsymbol{\calM}_2,
\label{M1M2}
\ee
which is valid for all third-order tensors $\boldsymbol{\calM}_1$ and $\boldsymbol{\calM}_2$, we find that
\be
\ba
\xg\mskip1mu\bfT\cdot(\xg\mskip1mu\bfT)^{\trans}&=(\symm\xg\mskip1mu\bfT+\mskip1mu\skk\xg\mskip1mu\bfT)\cdot(\symm\xg\mskip1mu\bfT+\mskip1mu\skk\xg\mskip1mu\bfT)^{\trans}\\[4pt]
&=(\symm\xg\mskip1mu\bfT+\mskip1mu\skk\xg\mskip1mu\bfT)\cdot((\symm\xg\mskip1mu\bfT)^{\trans}+\mskip1mu(\skk\xg\mskip1mu\bfT)^{\trans})\\[4pt]
&=(\symm\xg\mskip1mu\bfT+\mskip1mu\skk\xg\mskip1mu\bfT)\cdot(\symm\xg\mskip1mu\bfT+\mskip1mu(\skk\xg\mskip1mu\bfT)^{\trans})\\[4pt]
&=\symm\xg\mskip1mu\bfT\cdot\symm\xg\mskip1mu\bfT+\symm\xg\mskip1mu\bfT\cdot(\skk\xg\mskip1mu\bfT)^{\trans}\\
&\quad+\skk\xg\mskip1mu\bfT\cdot(\skk\xg\mskip1mu\bfT)^{\trans}\\[4pt]
&=|\symm\xg\mskip1mu\bfT|^2+(\symm\xg\mskip1mu\bfT)^{\trans}\cdot(\skk\xg\mskip1mu\bfT)^{\trans}+\skk\xg\mskip1mu\bfT\cdot(\skk\xg\mskip1mu\bfT)^{\trans}\\[4pt]
&=|\symm\xg\mskip1mu\bfT|^2+\skk\xg\mskip1mu\bfT\cdot(\skk\xg\mskip1mu\bfT)^{\trans}.
\ea
\label{skto}
\ee

We next wish to express $\skk\xg\mskip1mu\bfT\cdot(\skk\xg\mskip1mu\bfT)^{\trans}$ as a linear combination of positive-definite terms. To that end, we use the following characterization of $\skk\xg\bfT$:
\be
\skk\xg\bfT=\frac{1}{3}((\xc\bfT)\times + ^t((\xc\bfT)\times)).
\label{skeq}
\ee
The various operations in \eqref{skeq} are defined as follows. For a third-order tensor $\boldsymbol{\calM}$, $^{t}\boldsymbol{\calM}$ is the third-order tensor defined such that
\be
^{t}\boldsymbol{\calM}\cdot(\bfu\otimes\bfv\otimes\bfw)=\boldsymbol{\calM}\cdot(\bfv\otimes\bfu\otimes\bfw)
\ee
for all vectors $\bfu$, $\bfv$, and $\bfw$. For a second-order tensor $\bfN$, $\bfN\times$ is the third-order tensor defined such that, for any vector $\bfv$,
\be
(\bfN\times)\bfv=\bfN(\bfv\times),
\ee
where the second-order tensor $\bfv\times$ is defined by
\be
(\bfv\times)\bfw=\bfv\times\bfw
\ee
for any vector $\bfw$. Finally, for a second-order tensor $\bfN$, $\xc\bfN$ is the second-order tensor defined through
\be
(\xc\bfN)\times=\xg\bfN-((\xg\bfN)^{\trans})^{\trans}.
\ee
Let 
\be
\bfT_c=(\xc\bfT)\times.
\label{Cdef}
\ee
Then, using the identity 
\be
^{t}(\bfN\times)=-(\bfN\times)^{\trans}
\ee
holding for all second-order tensors $\bfN$ and the identity
\be
((\boldsymbol{\calM}^{\trans})^{\trans})^{\trans}=\boldsymbol{\calM}
\ee
holding for all third-order tensors $\boldsymbol{\calM}$, and invoking \eqref{M1M2}, \eqref{skeq}, and \eqref{Cdef}, it follows that
\be
\ba
\skk\xg\bfT\cdot(\skk\xg\bfT)^{\trans}&=\frac{1}{9}(\bfT_c+^t\bfT_c)\cdot(\bfT_c^{\trans}+(^t\bfT_c)^{\trans})\\[4pt]
&=\frac{1}{9}(\bfT_c\cdot\bfT_c^{\trans}-\bfT_c\cdot(\bfT_c^{\trans})^{\trans}-\bfT_c^{\trans}\cdot\bfT_c^{\trans}+\bfT_c^{\trans}\cdot(\bfT_c^{\trans})^{\trans})\\[4pt]
&=\frac{1}{9}(\bfT_c\cdot\bfT_c^{\trans}-\bfT_c^{\trans}\cdot\bfT_c-\bfT_c\cdot\bfT_c+\bfT_c\cdot\bfT_c^{\trans})\\[4pt]
&=-\frac{1}{9}(\bfT_c\cdot\bfT_c-\bfT_c\cdot\bfT_c^{\trans}).
\ea
\label{MMT}
\ee
Similarly,
\be
\ba
\skk\xg\bfT\cdot\skk\xg\bfT&=\frac{1}{9}(\bfT_c+^t\bfT_c)\cdot(\bfT_c+^t\bfT_c)\\[4pt]
&=\frac{1}{9}(\bfT_c\cdot\bfT_c+\bfT_c\cdot^t\bfT_c+\bfT_c\cdot^t\bfT_c+^t\bfT_c\cdot^t\bfT_c)\\[4pt]
&=\frac{1}{9}(\bfT_c\cdot\bfT_c-\bfT_c\cdot\bfT_c^{\trans}- \bfT_c^{\trans}\cdot\bfT_c+\bfT_c^{\trans}\cdot\bfT_c^{\trans})\\[4pt]
&=\frac{2}{9}(\bfT_c\cdot\bfT_c-\bfT_c\cdot\bfT_c^{\trans}).
\ea
\label{MM}
\ee
By \eqref{MMT} and \eqref{MM}, it follows that
\be
\skk\xg\bfT\cdot(\skk\xg\bfT)^{\trans}=-\frac{1}{2}\mskip1mu|\skk\xg\bfT|^2.
\label{sktoskt}
\ee
Substituting \eqref{sktoskt} in \eqref{skto} yields
\be
\xg\mskip1mu\bfT\cdot(\xg\mskip1mu\bfT)^{\trans}=|\symm\xg\mskip1mu\bfT|^2-\frac{1}{2}\mskip1mu|\skk\xg\bfT|^2.
\label{gtgtt}
\ee
Using \eqref{gtgt} and \eqref{gtgtt} in \eqref{simpli2}, it follows that
\be
U(\xg\mskip1mu\bfT)=\frac{1}{2}\Big(a_1 |\tr(\xg\mskip1mu\bfT)|^2+(a_2+a_3)|\symm\xg\mskip1mu\bfT|^2+\Big(a_2-\frac{a_3}{2}\Big)|\skk\xg\mskip1mu\bfT|^2\Big).
\label{W2}
\ee
By \eqref{W2}, the set of viable parameters for positive definiteness of $U$ is
\be
\calV_0=\big\{(a_1,a_2,a_3):a_1>0,\mskip4mu a_2+a_3>0,\mskip4mu 2c_2-a_3>0\big\}.
\label{viable0}
\ee
Notice, however, that in deriving the viable parameter set $\calV_0$ for the $U$ given by \eqref{simpli2}, we have not invoked the property $\xd\mskip1mu\bfT=\bf0$. Hence, the conditions in the definition of $\calV_0$ in \eqref{viable0} are sufficient but not necessary for positive definiteness of $U$. In other words, the set $\calV_0$ is contained in a set for which $U$ is positive-definite with the additional constraint $\xd\bfT=\bf0$. 

To determine how the constraint $\xd\bfT=\bf0$ enlarges the set of viable parameters for positive definiteness of $U$, we further decompose $\symm\xg\mskip1mu\bfT$, following \citet{auffray2013geometrical}, as 
\be
\symm\xg\mskip1mu\bfT=\bfH_1+\bfH_2,
\label{symdef}
\ee
where the third-order tensor fields $\bfH_1$ and $\bfH_2$ are given by
\begin{multline}
\bfH_1=\frac{1}{15}(\xg(\tr\bfT)\otimes\idem+\idem\otimes\xg(\tr\bfT)+(\xg(\tr\bfT)\otimes \idem)^{\trans})\\
+\frac{2}{15}(\xd\bfT\otimes\idem+\idem\otimes\xd\bfT+(\xd\bfT\otimes\idem)^{\trans})
\label{H1def}
\end{multline}
and
\be
\bfH_2=\symm\xg\mskip1mu\bfT-\bfH_1.
\label{H2def}
\ee
With reference to \eqref{H1def} and \eqref{H2def}, $\bfH_1$ and $\bfH_2$ satisfy  
\be
\bfH_1\cdot\bfH_2=0;
\label{H1H2dot}
\ee
furthermore, with reference to \eqref{trdef} and \eqref{H2def}, $\bfH_2$ satisfies
\be
\tr\bfH_2=\bf0. 
\label{trH2}
\ee
Similarly, $\skk\xg\mskip1mu\bfT$ can be decomposed as 
\be
\skk\xg\bfT=\bfH_3+\bfH_4,
\label{skdef}
\ee
where the third-order tensor fields $\bfH_3$ and $\bfH_4$ are given by
\be
\bfH_3=\frac{1}{3}(\sym(\xc\bfT)\times+^t(\sym (\xc\bfT)\times))
\label{H3def}
\ee
and
\be
\bfH_4=\skk\xg\mskip1mu\bfT-\bfH_3.
\ee
By \eqref{skeq}, $\bfH_4$ admits the alternative representation
\be
\bfH_4=\frac{1}
{3}(\sk(\xc\mskip1mu\bfT)\times+^t(\sk(\xc\mskip1mu\bfT)\times)).
\label{H4def}
\ee
Considerations analogous to those leading to \eqref{H1H2dot} and \eqref{trH2} yield
\be
\bfH_3\cdot\bfH_4=0
\label{H3H4dot}
\ee
and
\be
\tr\bfH_3=\bf0. 
\label{trH3}
\ee

Using the decompositions \eqref{symdef} and \eqref{skdef} in \eqref{W2} while taking note of the properties \eqref{H1H2dot}, \eqref{trH2}, \eqref{H3H4dot}, and \eqref{trH3}, we find that if $\osix{A}$ is isotropic, then the integrand of $E$, given by \eqref{E_nondim}, can be expressed as
\be
U=\frac{1}{2}\Big(a_1 (|\tr\bfH_1+\tr\bfH_4|^2)+(a_2+a_3)(|\bfH_1|^2+|\bfH_2|^2)+\Big(a_2-\frac{a_3}{2}\Big)(|\bfH_3|^2+|\bfH_4|^2)\Big).
\label{W3}
\ee
Since $\xd\mskip1mu\bfT=\bf0$, $\bfH_1$ in \eqref{H1def} reduces to
\be
\bfH_1=\frac{1}{15}(\xg(\tr\bfT)\otimes\idem+\idem\otimes\xg(\tr\bfT)+(\xg(\tr\bfT)\otimes \idem)^{\trans}).
\label{H1red}
\ee
Moreover, using the identity
\be
\sk(\xc\mskip1mu\bfN)=\frac{1}{2}(\xd\mskip1mu\bfN-\xg\mskip1mu(\tr\bfN))\times
\ee
which holds for any symmetric second-order tensor $\bfN$, it follows, from \eqref{H4def} and the constraint $\xd\mskip1mu\bfT=\bf0$, that
\be
\bfH_4=-\frac{1}{6}((\xg\mskip1mu(\tr\bfT)\times)\times+^t((\xg\mskip1mu(\tr\bfT)\times)\times)).
\label{H4red}
\ee
In view of the identities 
\be
\left.\ba
&(\bfv\otimes\idem)\cdot(\bfv\otimes\idem)=3|\bfv|^2,\\[4pt]
&(\idem\otimes\bfv)\cdot(\idem\otimes\bfv)=3|\bfv|^2,\\[4pt]
&(\bfv\otimes\idem)\cdot(\idem\otimes\bfv)=(\idem\otimes\bfv)\cdot(\bfv\otimes\idem)=|\bfv|^2,\\[4pt]
&(\bfv\otimes\idem)^{\trans}\cdot(\idem\otimes\bfv)=(\idem\otimes\bfv)\cdot(\bfv\otimes\idem)^{\trans}=|\bfv|^2,\\[4pt]
&((\bfv\times)\times)\cdot((\bfv\times)\times)=^t((\bfv\times)\times)\cdot^t((\bfv\times)\times)=4|\bfv|^2,\\[4pt]
& ^t((\bfv\times)\times)\cdot((\bfv\times)\times)=^t((\bfv\times)\times)\cdot((\bfv\times)\times)=2|\bfv|^2,
\ea\mskip3mu\right\}
\ee
which hold for any vector $\bfv$, we find that
\be
|\bfH_1|^2=\frac{1}{15}|\xg\mskip1mu\tr\bfT|^2, \qquad |\bfH_4|^2=\frac{1}{3}|\xg\mskip1mu\tr\bfT|^2=5|\bfH_1|^2.
\label{result1}
\ee
Furthermore, by the definition in \eqref{trdef} of trace of a third-order tensor, it can be shown that
\be
\left.\ba
\tr&(\xg(\tr\bfT)\otimes\idem)=\xg(\tr\bfT),\\[4pt]
\tr&(\idem\otimes\xg(\tr\bfT))=3\mskip1mu\xg(\tr\bfT),\\[4pt]
\tr&((\xg(\tr\bfT)\otimes\idem)^{\trans})=\xg(\tr\bfT),\\[4pt]
\tr&((\xg\mskip1mu(\tr\bfT)\times)\times)=-2\mskip1mu\xg\mskip1mu(\tr\bfT),\\[4pt]
\tr&(\mskip1mu ^{t}((\xg\mskip1mu(\tr\bfT)\times)\times))=-2\mskip1mu\xg\mskip1mu(\tr\bfT),
\ea\mskip3mu\right\}
\ee
whereby, from \eqref{H1red} and \eqref{H4red},
\be
|\tr\bfH_1|^2=\frac{1}{9}|\xg(\tr\bfT)|^2=\frac{5}{3}|\bfH_1|^2,\qquad|\tr\bfH_4|^2=\frac{4}{9}|\xg(\tr\bfT)|^2=\frac{20}{3}|\bfH_1|^2.
\label{result2}
\ee
Substituting \eqref{result1} and \eqref{result2} in \eqref{W3}, we see that
\be
U=\frac{3}{4}(10a_1+4a_2-a_3)|\bfH_1|^2+\frac{1}{2}(a_2+a_3)|\bfH_2|^2+\frac{1}{4}(2a_2-a_3)|\bfH_3|^2.
\label{W5}
\ee
Since $\bfH_1$, $\bfH_2$, and $\bfH_3$ are independent, by \eqref{W5}, the largest set $\tilde{\calV}$ of viable parameters $a_i$, $i=1,2,3$, for which $U$ is positive-definite is
\be
\tilde{\calV}=\big\{(a_1,a_2,a_3):10a_1+4a_2-a_3>0,\mskip4mu a_2+a_3>0,\mskip4mu 2a_2-a_3>0\big\}.
\label{viable1a}
\ee
In summary, for the case of isotropic $\osix{A}$, the conditions on $a_i$, $i=1,2,3$, in the definition \eqref{viable1a} of $\tilde{\calV}$ are necessary and sufficient for positive definiteness of $U$. As expected, $\calV_0$ defined by \eqref{viable0} is a proper subset of $\tilde{\calV}$. 

For further analysis, we consider the inequality $2a_1+a_2>0$ obtained by adding the first two conditions in the definition \eqref{viable1a} of $\tilde{\calV}$. We introduce
\be
\chi=2a_1+a_2,
\label{chidef}
\ee
and define the dimensionless constants
\be
\varrho=\frac{a_1}{\chi},\qquad \beta=\frac{a_2}{\chi},\qquad\text{and}\qquad\gamma=\frac{a_3}{\chi}.
\label{ctilde}
\ee
Notice, by \eqref{chidef} and \eqref{ctilde}$_{1,2}$, that
\be
\varrho=\frac{1-\beta}{2}.
\label{alphabeta}
\ee
Accordingly, by \eqref{simpli2} and \eqref{chidef}--\eqref{alphabeta}, it follows that
\be
U(\xg\mskip1mu\bfT)=\frac{\chi}{2}\Big(\frac{1-\beta}{2}|\tr(\xg\mskip1mu\bfT)|^2+
\beta|\xg\mskip1mu\bfT|^2+\gamma\xg\mskip1mu\bfT\cdot(\xg\mskip1mu\bfT)^{\trans}\Big).
\label{simpli2h}
\ee
Granted that 
\be
\chi>0,
\label{Gammadef}
\ee
the set of viable parameters for positive definiteness of $U$ can be equivalently written, by \eqref{viable1a}, \eqref{ctilde}, and \eqref{alphabeta}, as 
\be
\mathcal{V}=\big\{(\beta,\gamma):\beta+\gamma<5,\mskip4mu \beta+\gamma>0,\mskip4mu 2\beta-\gamma>0\big\}.
\label{viable}
\ee

The viable choices of $\beta$ and $\gamma$, obtained from \eqref{viable}, belong to the shaded region in Figure \ref{fig:c2c3b}. Notice that the semi-infinite strip of the viable points has its other finite edge at infinity in the fourth quadrant. Furthermore, since the inequalities in \eqref{viable} are strict, the points on the boundaries of this strip, given by the union of the sets $\partial\calV_i$, $i=1,2,3$, defined by
\be
\ba
&\partial\calV_1=\big\{(\beta,\gamma):\beta+\gamma=5,\mskip4mu \beta\ge 5/3\big\}, \\
&\partial\calV_2=\big\{(\beta,\gamma):\gamma=2\beta,\mskip4mu 0< \beta< 5/3\big\}, \\
&\partial\calV_3=\big\{(\beta,\gamma):\beta+\gamma=0,\mskip4mu \beta\ge 0\big\},
\ea
\label{Vbound}    
\ee
are {\em not} viable.
\begin{figure}[t!]
    \centering
         \includegraphics[width=0.65\textwidth]{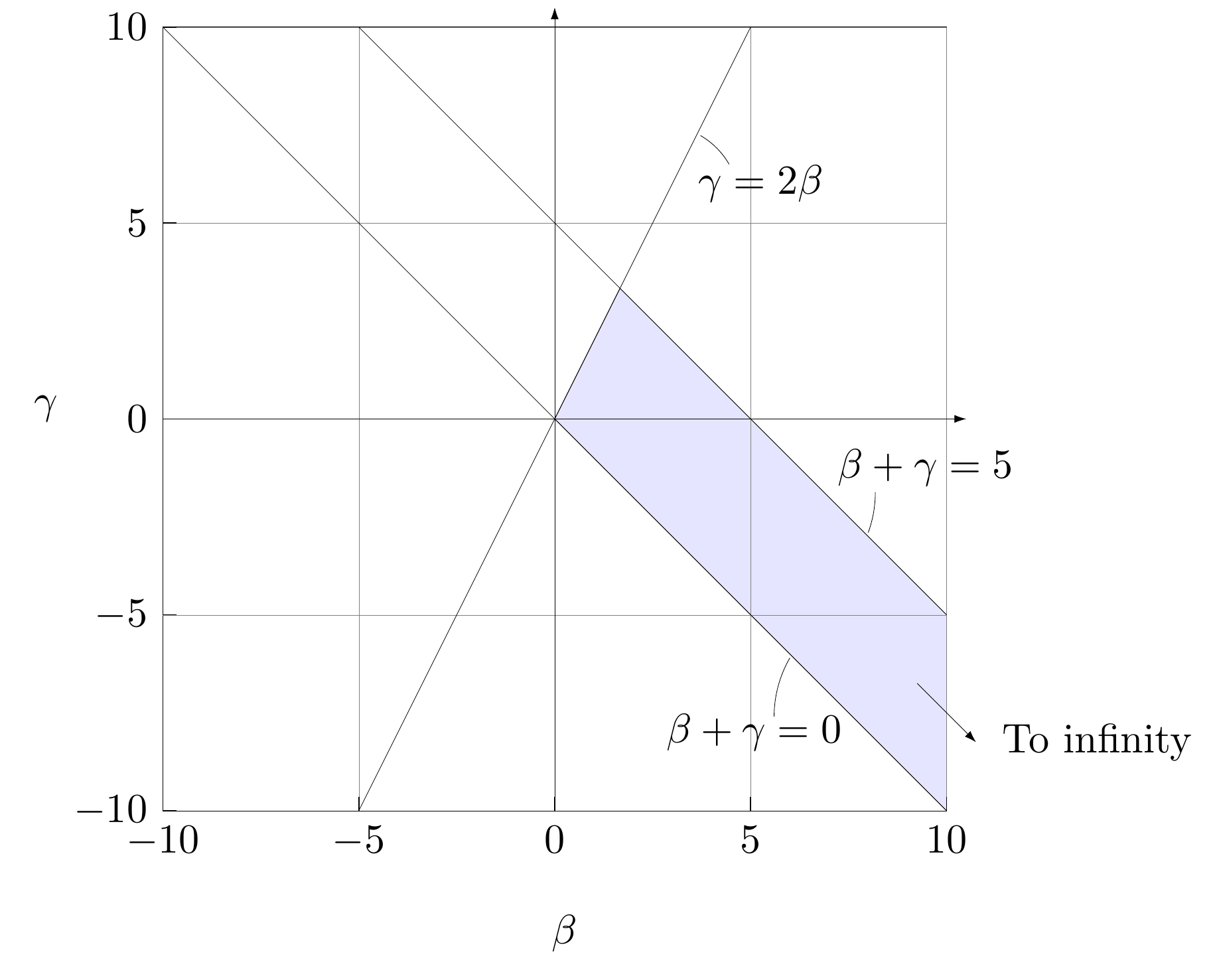}
    \caption{Semi-infinite strip, depicted in blue shade, of viable choices of the parameters $\beta$ and $\gamma$ for positive definiteness of $U$ for isotropic $\osix{A}$. The parameters corresponding to the boundaries of the strip are {\em not} viable.}
    \label{fig:c2c3b}
\end{figure}

To summarize, \eqref{simpli2h}, with $\chi$ satisfying \eqref{Gammadef}, and $\beta$ and $\gamma$ belonging to the set $\calV$ defined by \eqref{viable}, defines all positive-definite choices of $U$ for isotropic $\osix{A}$.

It is evident from \eqref{simpli2h} that the choice of $U$ adopted by Tiwari and Chatterjee, which is simply the integrand in \eqref{Eintro}, corresponds to $\chi=1$, $\beta=1$, and $\gamma=0$. By \eqref{Gammadef} and \eqref{viable}, we see that this particular choice of $U$ is positive-definite, as expected.

In terms of $\chi$, $\beta$, and $\gamma$, the dimensionless constrained boundary-value problem \eqref{bvp_nondim} takes, with reference to \eqref{EL_iso}, \eqref{natbc_iso}, \eqref{ctilde}, and \eqref{alphabeta}, the form
\begin{equation}
\left.
\begin{aligned}
\left.
\begin{aligned}
-\chi\frac{1-\beta}{2}(\Delta\tr\bfS)\idem- \chi\beta\Delta\bfS+\text{sym}\mskip2mu\xg\bfmu&=\lambda\bfS,
\\[4pt]
\xd\bfS&=\bf0,
\end{aligned}
\mskip4mu\right\}
&\quad\text{on}\quad\calR,
\\[4pt]
\left.
\begin{aligned}
\bfS\bfn&=\bf0,
\\[4pt]
\frac{1-\beta}{2}((\xg(\tr\bfS))\cdot\bfn)\bfP+\beta \bfP((\xg\bfS)\bfn)\bfP+\gamma \bfP((\xg\bfS)^{\trans}\bfn)\bfP&=\bfO,
\end{aligned}
\mskip4mu\right\}
&\quad\text{on}\quad\partial\calR,
\\[4pt]
\int_{\calR}|\bfS|^2\dv=1.
\hspace{40pt}&
\end{aligned}
\mskip4mu\right\}
\label{bvp_nondim_iso1}
\end{equation}
The only relation in \eqref{bvp_nondim_iso1} that depends explicitly on the parameter $\chi$ is the Euler--Lagrange equation \eqref{bvp_nondim_iso1}$_1$, which may be recast as
\be
-\frac{1-\beta}{2}(\Delta\tr\bfS)\idem- \beta\Delta\bfS+\text{sym}\mskip2mu\xg\tilde{\bfmu}=\tilde{\lambda}\bfS,
\label{ELre}
\ee
with
\be
\tilde{\bfmu}=\frac{\bfmu}{\chi} \qquad \text{and} \qquad \tilde{\lambda}=\frac{\lambda}{\chi}.
\label{redef}
\ee
We notice from \eqref{bvp_nondim_iso1} and \eqref{ELre} that for any choice of $\chi>0$, the dimensionless constrained boundary-value problem \eqref{bvp_nondim_iso1} is the same as that for $\chi=1$ upon redefining the Lagrange multipliers $\bfmu$ and $\lambda$ through \eqref{redef}. Thus, without loss of generality, we choose
\be
\chi=1.
\label{chi1}
\ee
Accordingly, the dimensionless constrained boundary-value problem \eqref{bvp_nondim_iso1} becomes
\begin{equation}
\left.
\begin{aligned}
\left.
\begin{aligned}
-\frac{1-\beta}{2}(\Delta\tr\bfS)\idem- \beta\Delta\bfS+\text{sym}\mskip2mu\xg\bfmu&=\lambda\bfS,
\\[4pt]
\xd\bfS&=\bf0,
\end{aligned}
\mskip4mu\right\}
&\quad\text{on}\quad\calR,
\\[4pt]
\left.
\begin{aligned}
\bfS\bfn&=\bf0,
\\[4pt]
\frac{1-\beta}{2}((\xg(\tr\bfS))\cdot\bfn)\bfP+\beta \bfP((\xg\bfS)\bfn)\bfP+\gamma \bfP((\xg\bfS)^{\trans}\bfn)\bfP&=\bfO,
\end{aligned}
\mskip4mu\right\}
&\quad\text{on}\quad\partial\calR,
\\[4pt]
\int_{\calR}|\bfS|^2\dv=1.
\hspace{40pt}&
\end{aligned}
\mskip4mu\right\}
\label{bvp_nondim_iso}
\end{equation}
Furthermore, by \eqref{E_nondim}, \eqref{Wfirsttime}, \eqref{simpli2h}, and \eqref{chi1}, the dimensionless functional $E$ takes the form
\be
E=\frac{1}{2}\int_{\calR}\Big(\frac{1-\beta}{2}|\tr(\xg\mskip1mu\bfT)|^2+
\beta|\xg\mskip1mu\bfT|^2+\gamma\mskip2mu\xg\mskip1mu\bfT\cdot(\xg\mskip1mu\bfT)^{\trans}\Big)\dv.
\label{E_nondim_iso}
\ee

Notice that the relations \eqref{bvp_nondim_iso}$_{1,2}$ are partial-differential equations; hence, it is not generally possible to find analytical solutions to the system \eqref{bvp_nondim_iso}. However, for the particular case in which $\calR$ is a spherical shell, it seems reasonable to explore whether the system supports spherically symmetric solutions. In the next section, we explore this possibility.

\section{Analytical solutions for the spherically symmetric case} \label{sec:spherical}
In Section \ref{sec:iso}, we derived the boundary-value problem satisfied by the extremizers for the special case in which the sixth-order tensor $\osix{A}$ entering the definition \eqref{E_nondim} of the functional $E$ is homogeneous and isotropic and the fourth-order tensor $\mfC$ in the normalization condition \eqref{bvp_nondim}$_5$ is the identity tensor. In this section, we further specialize the problem by seeking spherically symmetric solutions in the context of that specialization. Toward that end, we consider $\calR$ to be a spherical shell and stipulate that the solutions to the boundary-value problem \eqref{bvp_nondim_iso} be spherically symmetric too, thus reducing \eqref{bvp_nondim_iso} to a system of ordinary-differential and algebraic equations. In attempting to solve that system, we aim to derive analytical expressions for the spherically symmetric extremizers $\bfS_N,N\in\mathbb{N}$, and the corresponding Lagrange multiplier fields $\bfmu_N,N\in\mathbb{N}$.

\subsection{Reduction of the constrained boundary-value problem to a system of ordinary-differential and algebraic equations}
We consider $\calR$ to be a spherical shell having inner and outer radii $r_i$ and $r_o$, respectively, necessarily satisfying $r_o>r_i>0$. Furthermore, we select the center of the shell to be the origin $\bfo$ of point space $\calE$ in which it is embedded. Given a fixed right-handed orthonormal basis $\{\bfe_1,\bfe_2,\bfe_3\}$ for the translation space of $\calE$, we define spherical basis vectors through
\be
\left.
\begin{aligned}
\er&=\Sp\,(\Ct\,\bfe_1+\St\,\bfe_2)+\Cp\,\bfe_3,
\\[4pt]
\ep&=\Cp\,(\Ct\,\bfe_1+\St\,\bfe_2)-\Sp\,\bfe_3,
\\[4pt]
\et&=-\St\,\bfe_1+\Ct\,\bfe_2,
\end{aligned}
\mskip3mu\right\}
\qquad
r>0,
\quad
0\le\varphi\le\pi,
\quad
0\le\vartheta<2\pi.
\ee
Furthermore, we define associated second-order tensors $\bfPi$ and $\bfOm$ through
\be
\bfPi=\ep\otimes\ep+\et\otimes\et
=\idem-\er\otimes\er
\qquad\text{and}\qquad
\bfOm=\ep\otimes\et-\et\otimes\ep.
\label{PiOm}
\ee
It can then be shown that the tensor field $\bfR$ of the form
\be
\bfR(\zeta)=\er\otimes\er+\cos\zeta\,\bfPi
+\sin\zeta\,\bfOm,
\ee
represents a rotation about $\er$ by the angle $\zeta$.

A vector field $\bfmu$ is spherically symmetric if its components relative to the spherical basis $\{\er,\ep,\et\}$ depend at most on $r$ and it satisfies
\be
\bfR\bfmu=\bfmu
\label{Rmu=mu}
\ee
for all $\zeta$. It can then be shown that \eqref{Rmu=mu} holds for all $\zeta$ if and only if there exists a scalar-valued function $\mu$ of $r$ such that $\bfmu$ has the form
\be
\bfmu=\mu\bfe_r.
\label{murep}
\ee
Similarly, a second-order tensor field $\bfN$ is spherically symmetric if its components relative to $\{\er,\ep,\et\}$ depend at most on $r$ and it satisfies
\be
\bfR\bfN\mskip-2.5mu\bfR^{\trans}=\bfN
\label{RARt}
\ee
for all $\zeta$. It can then be shown that \eqref{RARt} holds for all $\zeta$ if and only if there exist scalar-valued functions $\Npar$, $\Nperp$, and $\Nspin$ of $r$ such that $\bfN$ has the form
\be
\bfN=\Npar\er\otimes\er
+\Nperp\bfPi
+\Nspin\bfOm.
\label{Arepgen}
\ee
Accordingly, a symmetric second-order tensor $\bfS$ that satisfies \eqref{RARt} must have the form
\be
\bfS=\Spar\er\otimes\er+\Sperp\bfPi
=(\Spar-\Sperp)\er\otimes\er+\Sperp\mskip-2mu\idem.
\label{Srep}
\ee

Considering the boundary-value problem \eqref{bvp_nondim_iso}, it is evident that we need to evaluate the quantities $\sym\xg\bfmu$, $\xg\bfS$, $\xd\bfS$, $\Delta\bfS$, and $\Delta(\tr\bfS)$. To begin, we see from \eqref{murep} that
\be
\xg\bfmu=\mu'\er\otimes\er+\frac{\mu}{r}\bfPi,
\label{Gradmu}
\ee
where a prime denotes a differentiation with respect to $r$. From \eqref{PiOm}$_1$ and \eqref{Gradmu}, it is evident that, as defined by \eqref{Gradmu}, $\xg\bfmu$ is symmetric.

Next, applying the gradient to \eqref{Srep} and using the identity
\begin{align}
\xg(\er\otimes\er)
&=\frac{1}{r}(
\er\otimes\mskip-2mu\bfPi+\ep\otimes\er\otimes\ep
+\et\otimes\er\otimes\et)
\notag\\[4pt]
&=\frac{1}{r}(
(\er\otimes\ep+\ep\otimes\er)\otimes\ep
+(\er\otimes\et+\et\otimes\er)\otimes\et),
\end{align}
%
%
we find that
\begin{align}
\xg\bfS&=\Spar'\er\otimes\er\otimes\er
+\Sperp'\bfPi\otimes\er
\notag\\[4pt]
&\qquad+\frac{\Spar-\Sperp}{r}(
(\er\otimes\ep+\ep\otimes\er)\otimes\ep
+(\er\otimes\et+\et\otimes\er)\otimes\et),
\label{GradS}
\end{align}
from which it follows that
\be
\xd\bfS=\Big(\Spar'+\frac{2(\Spar-\Sperp)}{r}\Big)\er.
\label{DivS}
\ee
Applying the divergence to \eqref{GradS}, we next find that 
\be
\Delta\bfS
=\Big(\Spar''+\frac{2}{r}\Spar'
-\frac{4(\Spar-\Sperp)}{r^2}\Big)\er\otimes\er
+\Big(\Sperp''+\frac{2}{r}\Sperp'
+\frac{2(\Spar-\Sperp)}{r^2}\Big)\bfPi.
\label{DelS}
\ee
Furthermore, since, by \eqref{Srep}, 
\be
\tr\bfS=\Spar+2\Sperp,
\label{trSsph}
\ee
we find that
\be
\Delta(\tr\bfS)=\Spar''+2\Sperp''+\frac{2(\Spar'+2\Sperp')}{r}.
\label{DeltaS}
\ee

Thus, granted that $\bfS$ and $\bfmu$ are spherically symmetric, the Euler--Lagrange equation \eqref{bvp_nondim_iso}$_1$ reduces, by \eqref{Gradmu}--\eqref{DeltaS}, to the system of ordinary-differential equations
\be
\left.
\begin{aligned}
-\frac{1-\beta}{2}\Big(\Spar''+2\Sperp''+\frac{2(\Spar'+2\Sperp')}{r}\Big)-\beta\Big(\Spar''+\frac{2}{r}\Spar'
-\frac{4(\Spar-\Sperp)}{r^2}\Big)+\mu'
&=\lambda\Spar, 
\\[4pt]
-\frac{1-\beta}{2}\Big(\Spar''+2\Sperp''+\frac{2(\Spar'+2\Sperp')}{r}\Big)-\beta\Big(\Sperp''+\frac{2}{r}\Sperp'
+\frac{2(\Spar-\Sperp)}{r^2}\Big)
+\frac{\mu}{r}&=\lambda\Sperp,
\end{aligned}
\mskip3mu\right\}
\label{EL_sph}
\ee
which are to be solved on $r_i<r<r_o$. Similarly, in view of the identities
\be
\left.
\begin{aligned}
(\xg(\tr\bfS))\cdot\bfn=(\Spar'+2\Sperp'),
\\[4pt]
(\xg\bfS)\er=\Spar'\er\otimes\er+\Sperp'\bfPi,
\\[4pt]
(\xg\bfS)^{\trans}\er=\Spar'\er\otimes\er+\frac{\Spar-\Sperp}{r}\bfPi,
\\[4pt]
\end{aligned}
\mskip3mu\right\}
\ee
the natural boundary condition \eqref{bvp_nondim_iso}$_4$ reduces to
\be
\frac{1-\beta}{2}(\Spar'+2\Sperp')+\beta\Sperp'
+\gamma\frac{\Spar-\Sperp}{r}=0, 
\label{nbc_sph}
\ee
which applies at $r=r_i$ and $r=r_o$. The remaining equations in the boundary-value problem, namely \eqref{bvp_nondim_iso}$_{2,3,5}$, reduce, respectively, to
\be
\left.
\ba
\Spar'+\frac{2(\Spar-\Sperp)}{r}&=0 \qquad \text{on} \qquad r_i< r< r_o,
\\[4pt]
\Spar&=0 \qquad \text{at} \qquad r=r_i \quad\text{and} \quad r=r_o,
\\[4pt]
\int_{r=r_i}^{r_o}(\Spar^2+2\Sperp^2)\dr&=\frac{1}{4\pi}.
\\[4pt]
\ea\mskip3mu\right\}
\label{bvp_sph}
\ee

The system consisting of \eqref{EL_sph}, \eqref{nbc_sph}, and \eqref{bvp_sph}, represents the spherically symmetric version of the constrained boundary-value problem \eqref{bvp_nondim_iso}. To confirm that we have the correct number of conditions, we first observe that the two second-order ordinary-differential equations \eqref{EL_sph} and the first-order ordinary-differential equation \eqref{bvp_sph}$_1$ can be recast into a system of five first-order ordinary-differential equations in the five scalar fields: $\Spar$, $\Spar'$, $\Sperp$, $\Sperp'$, and $\mu$. The conditions \eqref{nbc_sph} and \eqref{bvp_sph}$_2$ collectively constitute four boundary conditions. Another boundary condition is obtained by stipulating that the scalar fields $\Spar$, $\Sperp$, and $\Spar'$ are continuous at $r=r_i$, which implies that the condition \eqref{bvp_sph}$_1$ also holds at $r=r_i$. Consequently, by incorporating the normalization condition \eqref{bvp_sph}$_3$, we have the requisite conditions to determine the five scalar fields mentioned above and the constant multiplier $\lambda$. As noted at the end of Subsection \ref{sec:nondim}, a boundary condition for the Lagrange multiplier field $\mu$ is not needed.

Toward obtaining the analytical solutions of the spherically symmetric problem, we first eliminate $\Sperp$ and $\mu$ to yield a problem involving only $\Spar$. As we will see shortly, this yields a system depending on the functional parameters $\beta$ and $\gamma$ through only their sum $\beta+\gamma$. Thus, the spherical symmetry of the problem leads to a dimensional reduction in the parameter dependence, making it possible for us to determine the extremizers for all viable choices of the parameters $\beta$ and $\gamma$.

\subsection{Dimensional reduction in parameter dependence} \label{sec:dimred}
As outlined above, we express the boundary-value problem solely in terms of $\Spar$. This is done as follows. First, $\mu$ is eliminated by multiplying \eqref{EL_sph}$_2$ with $r$, differentiating the resulting equation, and subtracting \eqref{EL_sph}$_1$ from that equation. Then, $\Sperp$ is eliminated using \eqref{bvp_sph}$_1$, yielding the system
\be
\left.
\ba
\Spar''''+\frac{8\Spar'''}{r}+\Big(\frac{8}{r^2}+\lambda\Big)+\Big(-\frac{8}{r^3}+\frac{4\lambda}{r}\Big)\Spar'=0 \qquad \text{on} \qquad r_i< r< r_o,
\hspace{-110pt}&
\\[4pt]
\left.
\ba
\Spar&=0,
\\[4pt]
r\Spar''-(\beta+\gamma-4)\Spar'&=0,
\ea
\mskip4mu\right\}
&\qquad \text{at} \qquad r=r_i \quad\text{and} \quad r=r_o,
\\[4pt]
\int_{r=r_i}^{r_o}\Big(\Spar^2+2\Big(\frac{r\Spar'}{2}+\Spar\Big)^2\Big)\dr&=\frac{1}{4\pi}.
\ea\mskip3mu\right\}
\label{bvp_spar}
\ee 

We notice from \eqref{bvp_spar} that $\Spar$ and, by \eqref{bvp_sph}$_1$, $\Sperp$ depend on the functional parameters through $\beta+\gamma$. Moreover, we see from Figure \ref{fig:c2c3b} that one edge of the semi-infinite strip of viable $\beta$ and $\gamma$ values is the line $2\beta-\gamma=0$. Given these observations, it is convenient to work with
\be
p=\beta+\gamma \qquad \text{and} \qquad k=2\beta-\gamma. 
\label{pkcoords}
\ee
Recalling the set $\calV$ introduced in \eqref{viable} of viable parameter choices, we notice that its boundaries $\calV_i$, $i=1,2,3$, given by \eqref{Vbound}, can be represented in terms of $p$ and $k$ as 
\be
\ba
&\partial\calV_1=\big\{(p,k):p=5,\mskip4mu k\ge 0\big\},\\
&\partial\calV_2=\big\{(p,k):k=0,\mskip4mu 0< p< 5\big\},\\
&\partial\calV_3=\big\{(p,k):p=0,\mskip4mu k\ge 0\big\}.
\ea
\label{Vbound2} 
\ee  
These boundaries are depicted in Figure \ref{fig:pk}. 

We observe from \eqref{bvp_sph}$_1$ and \eqref{bvp_spar} that $\lambda$, $\Spar$, and $\Sperp$ do not depend on $k$. Consequently, having determined them for $0<p<5$ and any convenient choice of $k$, we can determine them across the entire semi-infinite strip depicted in Figure \ref{fig:pk}. 
\begin{figure}[t!]
    \centering
         \includegraphics[width=0.65\textwidth]{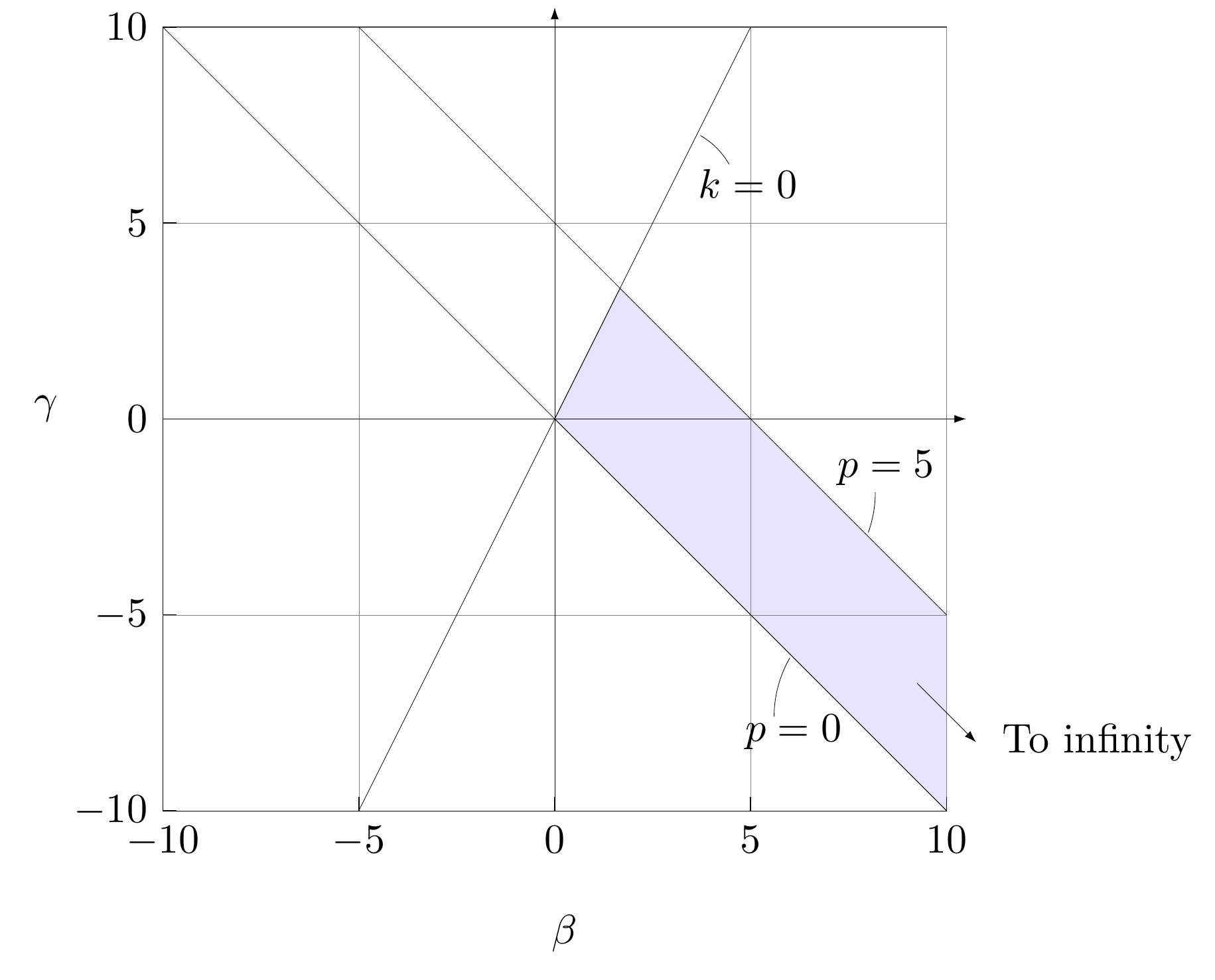}
    \caption{Depiction of the semi-infinite strip of viable parameters in terms of $p$ and $k$. In the spherically symmetric case, the determination of the extremizers and $\lambda$ for $0<p<5$ ascertains them across the entire strip.}
    \label{fig:pk}
\end{figure}

To see why this dimensional reduction in parameter dependence occurs, notice, by \eqref{transpose}, \eqref{GradS}, and \eqref{trSsph}, that for $\bfS$ spherically symmetric, 
\be
\left.\ba
&{\xg(\tr\bfS)=(\Spar'+2\Sperp')\bfe_r},
\\[4pt]
&{(\xg\bfS)^{\trans}=\Spar'\bfe_r\otimes\bfe_r\otimes\bfe_r+\Sperp'\bfe_r\otimes\bfPi}
\\
&\mskip100mu{+\frac{2(\Spar-\Sperp)}{r}(\bfe_{\phi}\otimes\sym(\bfe_r\otimes\bfe_{\phi})+\bfe_{\vartheta}\otimes\sym(\bfe_r\otimes\bfe_{\vartheta}))}.
\ea\mskip3mu\right\}
\label{simpli}
\ee
By \eqref{GradS} and \eqref{simpli}, the three terms in the integrand on the right-hand side of \eqref{E_nondim_iso} simplify to
\be
\left.\ba
&{|\xg(\tr\bfS)|^2=(\Spar'+2\Sperp')^2,}
\\[4pt]
&{|\xg\bfS|^2=\Spar'^2+2\Sperp'^2+\Big(\frac{2(\Spar-\Sperp)}{r}\Big)^2,}
\\[4pt]
&{\xg\bfS\cdot(\xg\bfS)^{\trans}=\Spar'^2+4\Sperp'\Big(\frac{\Spar-\Sperp}{r}\Big)+2\Big(\frac{\Spar-\Sperp}{r}\Big)^2.}
\ea\mskip3mu\right\}
\\[4pt]
\label{com}
\ee
Furthermore, in view of \eqref{bvp_sph}$_1$, \eqref{com}$_{2,3}$ reduce, respectively, to
\be
\left.\ba
|\xg\bfS|^2&=2(\Spar'^2+\Sperp'^2),
\\[4pt]
\xg\bfS\cdot(\xg\bfS)^{\trans}&=\frac{3\Spar'^2}{2}-2\Spar'\Sperp'.
\ea\mskip3mu\right\}
\label{com2}
\ee
With reference to \eqref{com}$_1$ and \eqref{com2}, we see that
\be
|\xg\bfS|^2-\frac{|\xg(\tr\bfS)|^2}{2}=\xg\bfS\cdot(\xg\bfS)^{\trans}.
\label{cond}
\ee
Hence, by \eqref{cond}, for $\bfS$ spherically symmetric and divergence-free, $E$, given by \eqref{E_nondim_iso}, simplifies to 
\be
E(\bfS)=\frac{1}{2}\int_{\calR} \Big(\frac{|\xg (\tr \bfS)|^2}{2}+(\beta+\gamma)\xg\bfS \cdot(\xg \bfS)^{\trans}\Big)\dv,
\label{Esph1}
\ee
which by \eqref{gradst} and a subsequent application of the divergence theorem, becomes
\be
E(\bfS)=\frac{\beta+\gamma}{2}\int_{\partial\calR}((\xg\bfS)^{\trans}[\bfS])\cdot\bfn\da+\frac{1}{4}\int_{\calR} |\xg (\tr \bfS)|^2\dv,
\label{Esph2}
\ee
where the notation $\boldsymbol{\calM}[\bfN]$ represents the action of a third-order tensor $\boldsymbol{\calM}$ on  a second-order tensor $\bfN$. From \eqref{Esph2}, it is evident that the parameters $\beta$ and $\gamma$ enter the formulation only through $p=\beta+\gamma$. Moreover, since that parameter appears on the contribution to $E$ from $\partial\calR$, it is evident why, of all the equations in the boundary-value problem \eqref{bvp_spar}, $p$ features only in the natural boundary condition \eqref{bvp_spar}$_3$.

In the following two subsections, we derive analytical expressions for the extremizer $\bfS$. We note an important point before proceeding. We have seen that the semi-infinite strip (depicted in Figure~\ref{fig:pk}) of parameter values for which the integrand $U$ of $E$ is positive-definite does not include the boundaries, namely $p=0$, $p=5$, and $k=0$, for general residual stress fields. However, for a non-trivial spherically symmetric residual stress field $\bfS$, we show in \ref{app:inclusive} that along these boundaries, $E(\bfS)>0$. For that reason, we will include these boundaries when using numerical methods to compute the constants that appear in the analytical expressions for the components of the stress and the Lagrange multiplier. 

\subsection{Analytical solutions for the extremizers for \boldmath{$p=0$}}\label{sec:p0}
The general solution $\Spar$ of \eqref{bvp_spar}$_1$ is given by
\be
\Spar(r)=\frac{c_0}{\omega^3r^3}\Big(c_1+\frac{c_4r^3}{r_o^3}+(c_2\omega r+c_3)\cos{\omega r}-(c_2-c_3\omega r)\sin{\omega r} \Big),
\label{srrgen}
\ee
where $c_0$ through $c_4$ are constants, with $c_i$, $i=1,\dots,4$, satisfying
\be
c_1^2+c_2^2+c_3^2+c_4^2=1,
\label{anorm}
\ee
and where $\omega$ is defined by
\be
\omega=\sqrt{\lambda}.
\label{omegasqrt}
\ee
It is readily verified that $\Spar$ as defined in \eqref{srrgen} satisfies 
\be
r\Spar''+4\srr'=\omega^2r\Big(\frac{c_0 c_1}{\omega^3r^3}+\frac{c_0 c_4}{\omega^3r_o^3}-\srr\Big).
\label{check}
\ee
The boundary conditions \eqref{bvp_spar}$_{2,3}$, with aid of \eqref{check}, yield
\be
\frac{c_0 r_i}{\omega}\Big(\frac{c_1}{r_i^3}+\frac{c_4}{r_o^3}\Big)=p\srr'\big|_{r_i}, \qquad
\frac{c_0 r_o}{\omega}\Big(\frac{c_1}{r_o^3}+\frac{c_4}{r_o^3}\Big)=p\srr'\big|_{r_o}.
\label{a1a4}
\ee
At $p=0$, \eqref{a1a4} reduces to the homogeneous system
\be
\frac{c_0 r_i}{\omega}\Big(\frac{c_1}{r_i^3}+\frac{c_4}{r_o^3}\Big)=0,\qquad
\frac{c_0 r_o}{\omega}\Big(\frac{c_1}{r_o^3}+\frac{c_4}{r_o^3}\Big)=0.
\label{a1a42}
\ee
For non-trivial solutions, $c_0\neq0$ and, thus, $c_1=c_4=0$, whereby
\eqref{srrgen} simplifies to
\be
\Spar(r)=\frac{c_0}{\omega^3r^3}((c_2\omega r+c_3)\cos{\omega r}-(c_2-c_3\omega r)\sin{\omega r} ).
\label{chance}
\ee
Using \eqref{anorm}, with $c_1=c_4=0$ substituted therein, and the trigonometric identity $\sin(\theta_1+\theta_2)=\sin{\theta_1}\cos{\theta_2}+\sin{\theta_2}\cos{\theta_1}$, \eqref{chance} can be written as
\be
\Spar(r)=A_{\scriptscriptstyle{\parallel}}(r)\sin(\omega r+\theta_{\scriptscriptstyle{\parallel}}(r)),
\label{srrapmp}
\ee
where $A_{\scriptscriptstyle{\parallel}}$ and $\theta_{\scriptscriptstyle{\parallel}}$ are defined by
\be
A_{\scriptscriptstyle{\parallel}}(r)=\frac{c_0\sqrt{1+\omega^2 r^2}}{\omega^3r^3} \qquad \text{and} \qquad \theta_{\scriptscriptstyle{\parallel}}(r)=\sin^{-1}\Big(\frac{c_2\omega r+c_3}{\sqrt{1+\omega^2r^2}}\Big).
\label{Ath}
\ee
Upon substituting \eqref{srrapmp} in \eqref{bvp_sph}$_1$, and employing manipulations resembling those used above, we find that
\be
\Sperp(r)=A_{\scriptscriptstyle{\perp}}(r)\cos(\omega r-\theta_{\scriptscriptstyle{\perp}}(r)),
\label{sttapmp}
\ee
where $A_{\scriptscriptstyle{\perp}}$ and $\theta_{\scriptscriptstyle{\perp}}$ are defined by
\be
A_{\scriptscriptstyle{\perp}}(r)=\frac{c_0\sqrt{\omega^4r^4-\omega^2r^2+1}}{2\omega^3r^3} \qquad \text{and} \qquad \theta_{\scriptscriptstyle{\perp}}(r)=\cos^{-1}\Big(\frac{c_2\omega r+c_3(1-\omega^2r^2)}{\sqrt{\omega^4r^4-\omega^2r^2+1}}\Big).
\label{Aperp}
\ee
Relations \eqref{srrapmp} and \eqref{sttapmp} clarify the forms of $\Spar$ and $\Sperp$: they are harmonic functions with radially varying amplitudes and phases. Furthermore, their frequency of oscillation is $\omega$.

To determine the four constants $c_0$, $c_2$, $c_3$, and $\omega$, we substitute \eqref{srrapmp} in \eqref{bvp_spar}$_{2\text{--}4}$ to obtain a system of six algebraic equations consisting of the derived relations and the relation \eqref{anorm}. Notice, however, that two of the four relations derived from \eqref{bvp_spar}$_{2,3}$ have already been utilized to establish that $c_1=c_4=0$ from \eqref{a1a42}. Consequently, we are left with four independent algebraic equations to solve for the four constants $c_0$, $c_2$, $c_3$, and $\omega$. 

Due to the non-linear dependence \eqref{srrapmp} and \eqref{sttapmp} of the stress components on $\omega=\sqrt{\lambda}$, it appears non-linearly in the algebraic equation system. Similarly, $c_2$ and $c_3$, in addition to being non-linearly related through \eqref{anorm}, appear non-linearly in the system due to the normalization constraint \eqref{bvp_spar}$_3$. Finally, $c_0$ appears non-linearly owing to the same constraint \eqref{bvp_spar}$_3$. Therefore, the constrained boundary-value problem satisfied by the extremizers is inherently non-linear, as highlighted towards the end of Subsection \ref{sec:nondim}. This non-linearity necessitates the utilization of numerical methods to determine the constants $c_0$, $c_2$, $c_3$, and $\omega$. From that viewpoint, it might be more appropriate to characterize the spherically-symmetric solutions obtained in this section as `semi-analytical'.

To compute $c_0$, $c_2$, $c_3$, and $\omega$, we use our own Newton--Raphson-based routine. Guided by the foregoing interpretation of the solutions being harmonic functions with radially varying amplitudes and phases, for initial guesses $\omega_0$ of $\omega$ for the iterative solution method, we use $\omega_0=N\pi/(r_o-r_i)$ for the $N^{\text{th}}$ solution. Furthermore, the dimensionless inner and outer radii $r_i$ and $r_o$ are now and hereafter taken to be $0.5$ and $1$, respectively. 

The plots of $c_0$ and $\omega$ versus $N$, $N=1,2,\dots,10$, are shown in the left panel of Figure \ref{fig:ai_p0} and those of $c_2$ and $c_3$ versus $N$, $N=1,2,\dots,10$, are shown in the right panel of Figure \ref{fig:ai_p0}. We notice from the former that $c_0$ and $\omega$ scale linearly with $N$, with $\omega_N$ close to $N\pi/(r_o-r_i)=2N\pi$, and from the latter that $c_2$ and $c_3$ seem to converge to 0 and 1, respectively.
\begin{figure}[t!]
    \centering
    \begin{subfigure}[b]{0.495\textwidth}
         \centering
         \includegraphics[width=\textwidth]{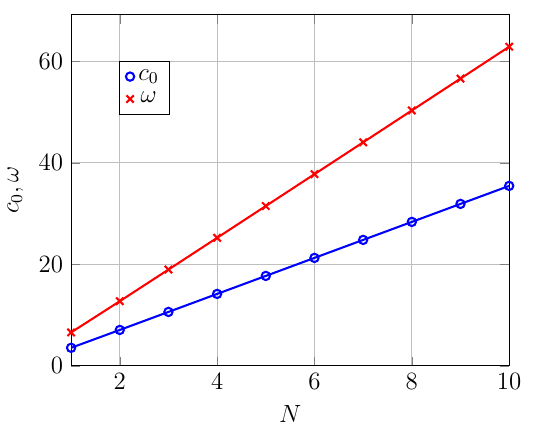}
     \end{subfigure}
     \begin{subfigure}[b]{0.495\textwidth}
         \centering
         \includegraphics[width=\textwidth]{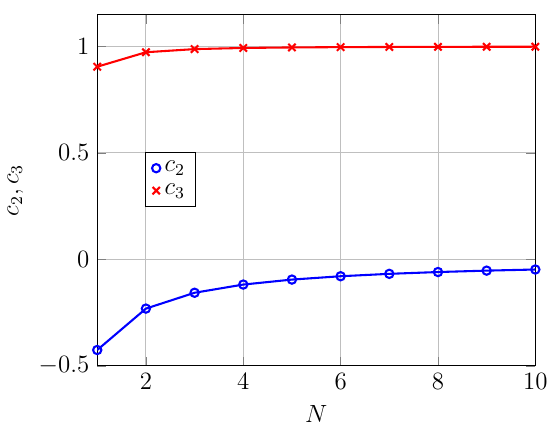}
     \end{subfigure}
    \caption{Left: Plots of $c_0$ and $\omega$ versus $N$ for $p=0$. Right: Plots of $c_2$ and $c_3$ versus $N$ for $p=0$.}
    \label{fig:ai_p0}
\end{figure}

The corresponding components $\Spar$ and $\Sperp$ of the extremizers are then computed using \eqref{srrapmp} and \eqref{sttapmp}, respectively. Plots of $\Spar$ and $\Sperp$ versus $r$ for the first four extremizers are shown in Figure \ref{fig:p0}, which reveals that the frequency of oscillation of $\Spar$ and $\Sperp$ increases with $N$, an observation in line with the left panel of Figure \ref{fig:ai_p0}. Furthermore, we notice that the amplitudes of $\Spar$ and $\Sperp$ decrease with $r$, in agreement with \eqref{Ath}$_1$ and \eqref{Aperp}$_1$, respectively. Finally, we find from \eqref{Ath}$_1$ that the leading term $\Apar^l$ of the power series expansion of $\Apar$ about $\omega=\infty$ is inversely proportional to $\omega^2$:
\be
\Apar^l(r)=\frac{c_0}{\omega^2r^2}.
\label{Aparl}
\ee
In conjunction with the observation from the left panel of Figure \ref{fig:ai_p0} that $\omega$ and $c_0$ scale linearly with $N$, \eqref{Aparl} implies that $\Apar$ scales with $1/N$ for large $N$. Similarly, the leading term $\Aperp^l$ of the power series expansion of $\Aperp$ about $\omega=\infty$ is inversely proportional to $\omega$:
\be
\Aperp^l(r)=\frac{c_0}{2\omega r}.
\ee
Thus, $\Aperp$ scales with $1/N^0$ for large $N$. These observations regarding the scalings of $\Apar$ and $\Aperp$ with $N$ are in agreement with Figure \ref{fig:p0}.
\begin{figure}[t!]
    \centering
    \begin{subfigure}[b]{0.46\textwidth}
         \centering
         \includegraphics[width=\textwidth]{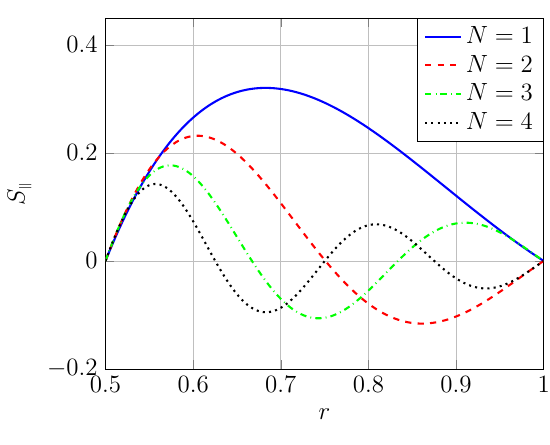}
     \end{subfigure}
     \hspace{5mm}
     \begin{subfigure}[b]{0.464\textwidth}
         \centering
         \includegraphics[width=\textwidth]{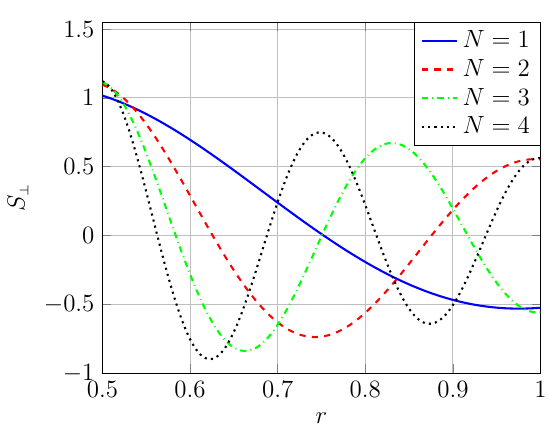}
     \end{subfigure}
    \caption{Plots of $\Spar$ versus $r$ (left) and $\Sperp$ versus $r$ (right) corresponding to the first four extremizers for $p=0$.}
    \label{fig:p0}
\end{figure}

\subsection{Analytical solutions for the extremizers for \boldmath{$0<p\leq 5$}}
Having obtained the solution for $p=0$, the solutions for the other viable values of $p$, namely $0<p\leq 5$, can be determined by setting up differential equations for $c_i$, $i=0,1,2,3,4$, and $\omega$. This is done by substituting \eqref{srrgen} in \eqref{bvp_spar}$_{2\text{--}4}$ and differentiating the resulting equations, and \eqref{anorm}, with respect to $p$,
yielding a system of equations of the form
\be
Mv=b,
\ee
where the $6\times6$ matrix $M$ and the $6\times1$ matrix $b$ depend on $c_i$, $i=0,1,2,3,4$, and $\omega$, and
\be
v=\Bigg[\frac{\text{d}c_0}{\text{d}p} \quad \frac{\text{d}c_1}{\text{d}p} \quad \frac{\text{d}c_2}{\text{d}p} \quad \frac{\text{d}c_3}{\text{d}p} \quad \frac{\text{d}c_4}{\text{d}p}\quad \frac{\text{d}\omega}{\text{d}p} \Bigg]^{\trans}.
\ee
The ensuing differential equations
\be
v=M^{-1}b
\label{ensue}
\ee
are numerically integrated, using the Matlab routine {\tt ode45}, from $p=0$ to $p=5$ to obtain $c_i$, $i=0,1,2,3,4$, and $\omega$. The initial conditions for \eqref{ensue} are taken from Subsection \eqref{sec:p0}. The solutions are then substituted in \eqref{srrgen} to obtain $\Spar$, which, in turn, is substituted in \eqref{bvp_sph}$_1$ to obtain $\Sperp$.

We plot $\omega$, after dividing it by $N$, versus $p$ in Figure \ref{fig:omega_p>0} for $N$ from $1$ through $4$. The corresponding $\Spar$ and $\Sperp$ are plotted against $r$ in Figure \ref{fig:p0to5} for six representative values of $p: 0,1,2,3,4$, and $5$. We notice from these figures that there is little change in $\omega$ and the stress profiles across the viable range of $p$, and, in fact, across the entire semi-infinite strip of viable parameter choices depicted in Figure \ref{fig:pk}. Since by \eqref{lambda2E} and \eqref{omegasqrt}, $\omega=\sqrt{\lambda}=\sqrt{2E}$, the value of $E$ corresponding to extremizers with fixed $N$ also changes little across the viable region.
\begin{figure}[t!]
    \centering
         \includegraphics[width=0.5\textwidth]{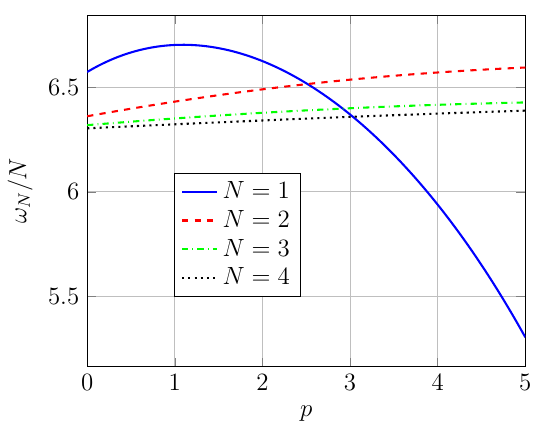}
    \caption{Plot of $\omega_N/N$ versus $p$ for $N$ from $1$ through $4$.}
    \label{fig:omega_p>0}
\end{figure}
\begin{figure}[t!]
    \centering
    \includegraphics[scale=0.61]{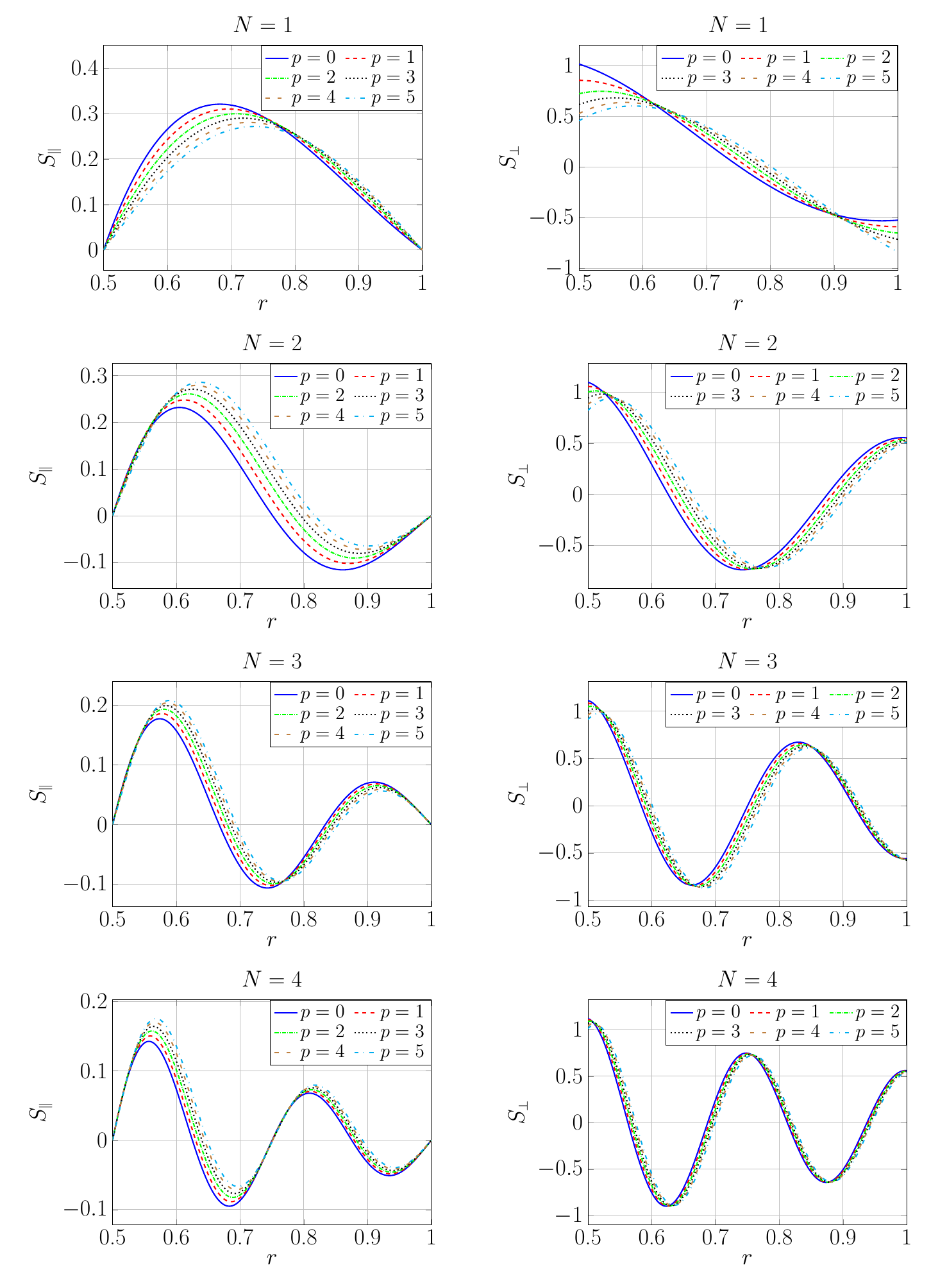}
    \caption{Plots of $\Spar$ versus $r$ and $\Sperp$ versus $r$ corresponding to the first four extremizers for $p=0,1,2,3,4,$ and $5$.}
    \label{fig:p0to5}
\end{figure}

\subsection{Analytical solutions for the Lagrange multiplier field}
From \eqref{EL_sph}$_2$, \eqref{bvp_sph}$_1$, and \eqref{srrgen}, it follows that the Lagrange multiplier field $\mu$ is given by
\be
\mu(r)=-\frac{c_0}{2\omega r^2}\Big(c_1-\frac{2c_4r^3}{r_o^3}-(1-\beta)((c_3+c_2\omega r)\cos{\omega r}-(c_2-c_3\omega r)\sin{\omega}r)\Big)
\label{murgen}
\ee
or, equivalently,
\be
\mu(r)=-\frac{c_0}{2\omega r^2}\Big(c_1-\frac{2c_4r^3}{r_o^3}+(1-\beta)A_{\mu}(r)\sin(\omega r+\theta_{\mu}(r))\Big),
\label{murshort}
\ee
where $A_{\mu}$ and $\theta_{\mu}$ are defined by
\be
A_{\mu}(r)=\sqrt{(1+\omega^2r^2)(c_2^2+c_3^2)} \qquad \text{and} \qquad \theta_{\mu}(r)=\sin^{-1}\Big(\frac{c_3+c_2\omega r}{A_{\mu}(r)}\Big).
\label{Amu}
\ee
Thus, $\mu$ depends on the functional parameters through $p=\beta+\gamma$ and $1-\beta$.
Consequently, unlike $\Spar$, $\Sperp$, and $\lambda$, it varies with $k$. 

From \eqref{murshort}, two parameter regimes of special functional forms of $\mu$ are immediately identified. First, recall from Subsection \ref{sec:p0} that at $p=0$, $c_1=c_4=0$ and, hence, that
\be
\mu(r)=-\frac{c_0(1-\beta) A_{\mu}(r)\sin(\omega r+\theta_{\mu}(r))}{2\omega r^2}.
\label{murshort2}
\ee 
So, at $p=0$, $\mu$ is harmonic, with radially varying amplitude and phase. In contrast, the second special form of $\mu$, obtained for $\beta=1$, is purely algebraic:
\be
\mu(r)=-\frac{c_0}{2\omega r^2}\Big(c_1-\frac{2c_4r^3}{r_o^3}\Big).
\ee
These two special parameter regimes are depicted in Figure \ref{fig:mur}. It is evident that at the intersection point of these two regimes, achieved at $\beta=1$ and $\gamma=-1$, $\mu$ vanishes identically. We will explain this phenomenon in detail in Subsection \ref{sec:Helmholtz}.
\begin{figure}[t!]
    \centering
         \includegraphics[width=0.65\textwidth]{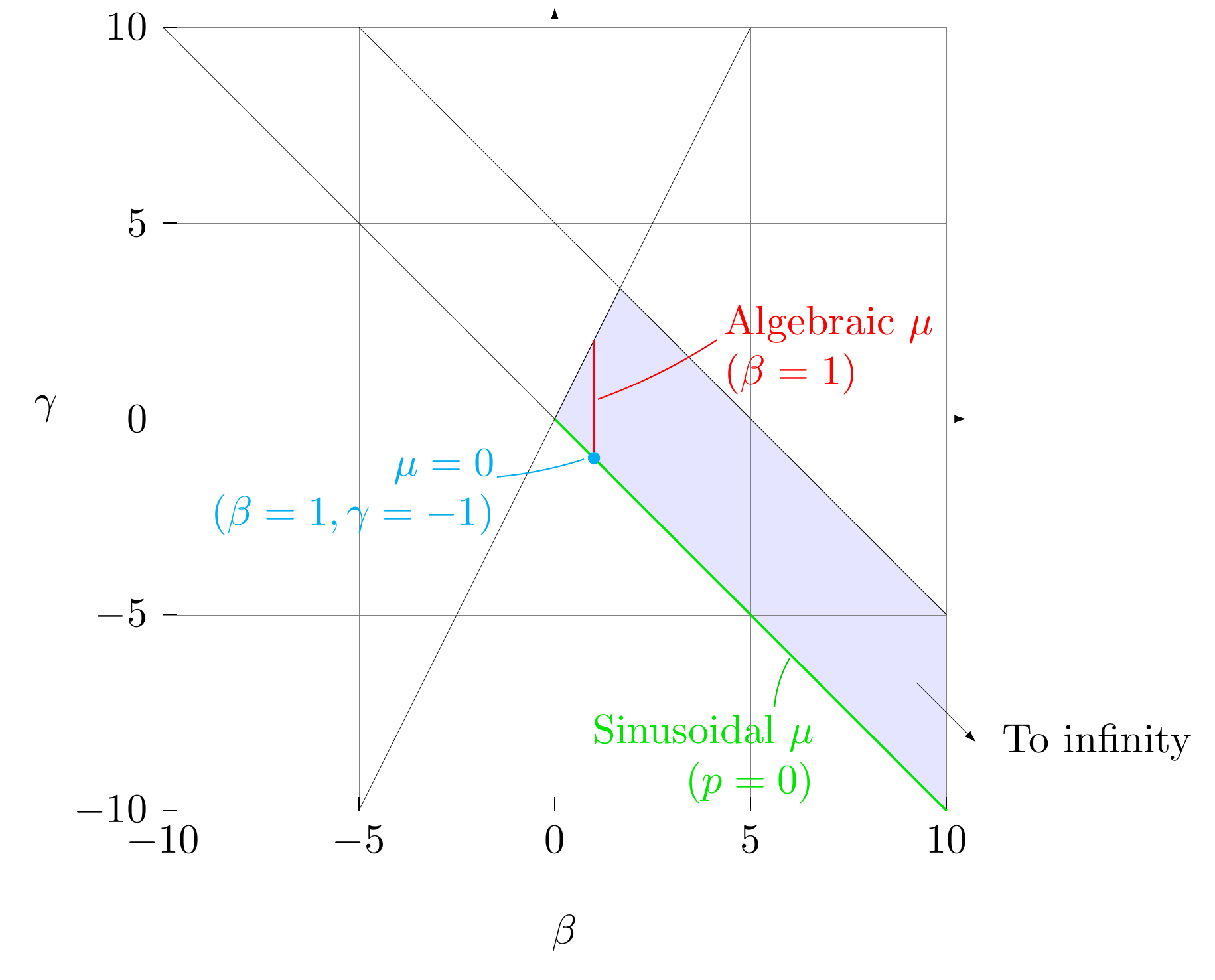}
    \caption{Parameter regimes for purely sinusoidal (green) and purely algebraic (red) $\mu$. At their intersection (cyan dot), $\mu$ vanishes.}
    \label{fig:mur}
\end{figure} 

To clarify the dependence of $\mu$ upon $k$, we invoke \eqref{pkcoords}, \eqref{EL_sph}$_2$, and \eqref{omegasqrt} to write
\be
\mu=kf_1+f_2,
\ee
where $f_1$ and $f_2$, defined by
\be
\left.
\ba
f_1(r)&=-\frac{2\Spar'(r)}{3}-\frac{r\Spar''(r)}{6} \qquad \text{and}
\\[4pt]
f_2(r)&=pf_1(r)+\frac{r\Spar''(r)}{2}+r\Sperp''(r)+\Spar'(r)+2\Sperp'(r)+\omega^2 r \Sperp(r),
\ea\mskip3mu\right\}
\label{f1f2a}
\ee
vary with $p$ but not with $k$. Thus, for a given $ k\gg\|f_2\|/\|f_1\|$, where $\|\cdot\|$ denotes an appropriate functional norm, $\mu\approx kf_1$. With reference to \eqref{bvp_sph}$_1$, \eqref{srrgen}, \eqref{omegasqrt}, and \eqref{Amu}, $f_1$ and $f_2$ can be expressed as
\be
f_1(r)=\frac{c_0 A_{\mu}(r)\sin(\omega r + \theta_{\mu}(r))}{6\omega r^2} \qquad \text{and} \qquad f_2(r)=-\frac{c_0}{\omega r^2}\Big(\frac{c_1}{2}-\frac{c_4r^3}{r_o^3}\Big)+(p-3)f_1(r).
\label{f1f2}
\ee
Since $f_1$ is harmonic, we conclude that for large $k$, $\mu$ exhibits oscillatory behaviour. 

We plot the density plots of $\mu$ for $0.5\le r\le 1$ and $0\le p\le 5$ corresponding to $N=1$ and $N=2$ in Figures \ref{fig:levsurmur1} and \ref{fig:levsurmur2}, respectively. The density plots are created for the representative choices $k=0, 3, 10$, and $100$ of $k$. In both the figures, we notice that at $k=0$ and $p=3$, which by \eqref{pkcoords} corresponds to the choice $\beta=1$, $\mu$ is purely algebraic. Similarly, we notice that at $k=3$ and $p=0$, which corresponds to the choices $\beta=1$ and $\gamma=-1$, $\mu$ vanishes. Furthermore, we notice that the plots corresponding to $k=10$ and $k=100$ have similar relative distributions of $\mu$ with respect to $r$ and $p$; however, the values attained by $\mu$ for $k=100$ are larger in magnitude in comparison to those corresponding to $k=10$. Moreover, the distributions corresponding to $k=100$ are nearly harmonic. As expected, these observations are in line with the foregoing discussion. 
\begin{figure}[t!]
    \centering
    \begin{subfigure}[b]{0.4\textwidth}
         \centering
         \includegraphics[width=\textwidth]{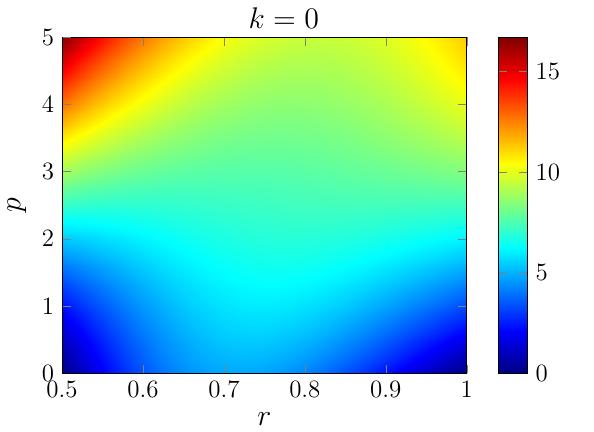}
     \end{subfigure}
     \begin{subfigure}[b]{0.4\textwidth}
         \centering
         \includegraphics[width=\textwidth]{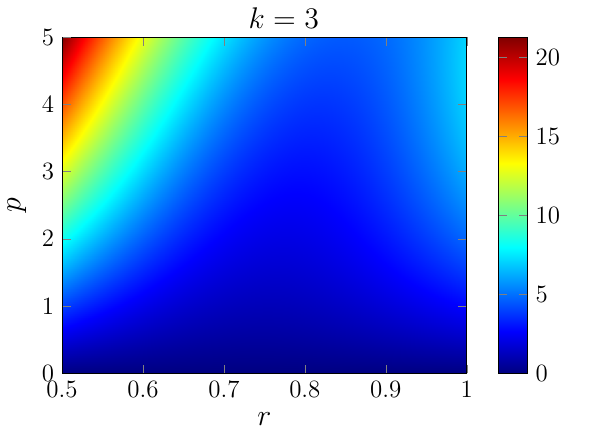}
     \end{subfigure}
    \begin{subfigure}[b]{0.4\textwidth}
         \centering
         \includegraphics[width=\textwidth]{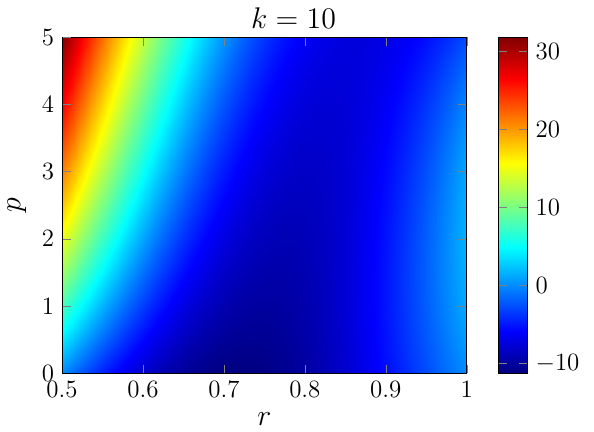}
     \end{subfigure}
     \begin{subfigure}[b]{0.4\textwidth}
         \centering
         \includegraphics[width=\textwidth]{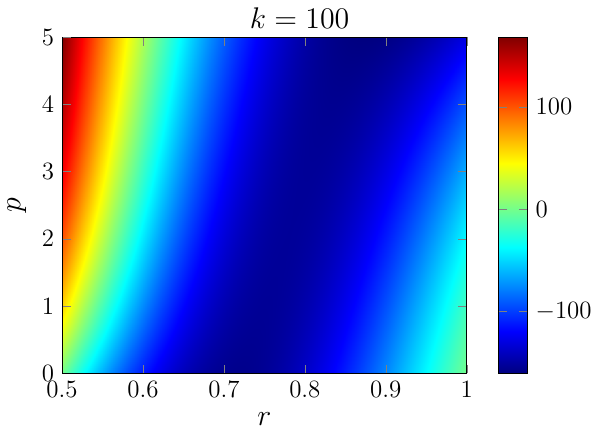}
     \end{subfigure}
    \caption{Density plots of $\mu$ for $0.5\le r\le 1$ and $0\le p\le 5$ corresponding to $N=1$, with $k=0,3,10$, and $100$.}
    \label{fig:levsurmur1}
    \begin{subfigure}[b]{0.4\textwidth}
         \centering
         \includegraphics[width=\textwidth]{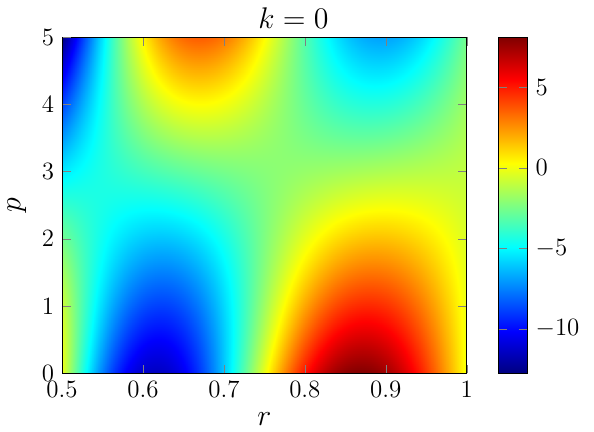}
     \end{subfigure}
     \begin{subfigure}[b]{0.4\textwidth}
         \centering
         \includegraphics[width=\textwidth]{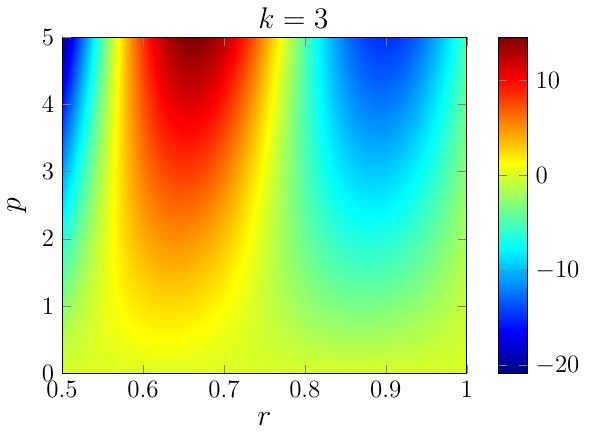}
     \end{subfigure}
    \begin{subfigure}[b]{0.4\textwidth}
         \centering
         \includegraphics[width=\textwidth]{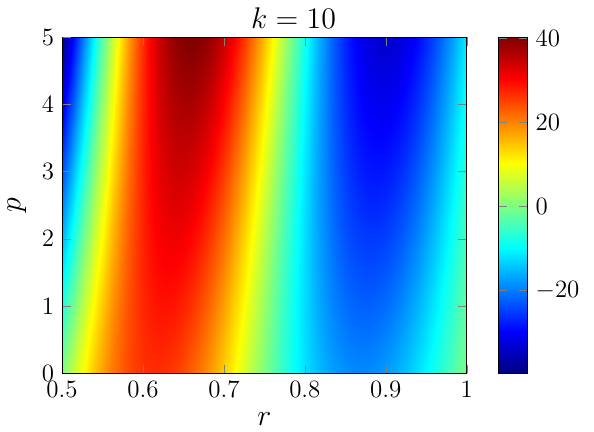}
     \end{subfigure}
     \begin{subfigure}[b]{0.4\textwidth}
         \centering
         \includegraphics[width=\textwidth]{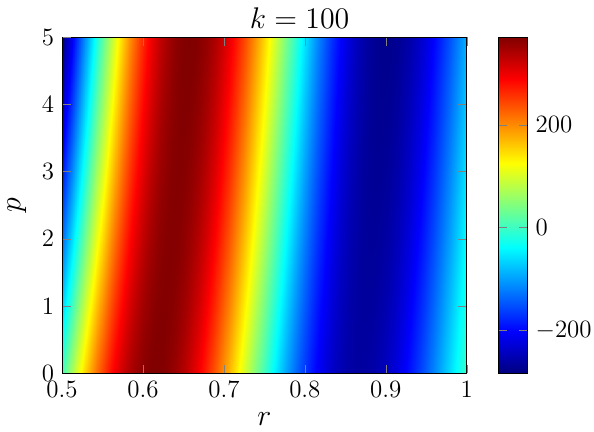}
     \end{subfigure}
    \caption{Density plots of $\mu$ for $0.5\le r\le 1$ and $0\le p\le 5$ corresponding to $N=2$, with $k=0,3,10$, and $100$.}
    \label{fig:levsurmur2}
\end{figure}

\subsection{Vanishing of the Lagrange multiplier field. Helmholtz equation} \label{sec:Helmholtz}
In the previous subsection, we observed that at $\beta=1$ and $\gamma=-1$, $\mu$ or, equivalently, $\bfmu$ vanishes.  It therefore follows from \eqref{bvp_nondim_iso}$_1$ and the property $\lambda>0$ that, in addition to \eqref{bvp_nondim_iso}$_{\text{2--5}}$, the extremizer $\bfS$ satisfies the Helmholtz equation
\be
\Delta\bfS+\lambda\bfS=\bf0 \qquad \text{in} \qquad \calR.
\ee
In view of the importance of the Helmholtz equation in physical sciences, we address the following question in this subsection: Does $\bfmu$ vanish at $\beta=1$ and $\gamma=-1$ in {\em absence} of spherical symmetry?

Recall that the Lagrange multiplier field $\bfmu$ serves to enforce the constraint $\xd\bfS=\bf0$. Granted this interpretation, the vanishing of $\bfmu$ must mean that an extremizer of the variational problem posed in Section \ref{sec:bvp}, with $\osix{A}$ therein being homogeneous and isotropic, must also be an extremizer upon lifting the foregoing constraint. In other words, the vanishing of $\bfmu$ implies that an extremizer of Section \ref{sec:bvp} is also an extremizer of $E$, as defined by \eqref{E_nondim_iso}, over the set
\be
\calS_u=\left\{\bfT:\bfT\in\Sym,\mskip2mu\bfT\bfn|_{\partial\calR}=\bfzero,\mskip2mu\int_{\calR}\bfT\cdot\mfC[\bfT]\dv<\infty,\mskip2mu E(\bfT)<\infty\right\},
\label{setSu}
\ee
subject to the condition
\be
\int_{\calR}\bfT\cdot\mfC[\bfT]\dv=1,
\label{normu}
\ee
in which $\mfC$ is the fourth-order identity tensor. From this viewpoint, to explain the vanishing of $\bfmu$, we must show that a solution of our original extremization problem of Section \ref{sec:bvp}, henceforth referred to as the `constrained extremization problem' for easy reference, is also a solution to the problem of extremizing $E$ in \eqref{setSu} subject to \eqref{normu}, henceforth referred to as the `unconstrained extremization problem'. 

To show that the solution to the constrained extremization problem is also a solution to the unconstrained extremization problem, we begin by observing, from \eqref{E_nondim_iso}, that for $\beta=1$ and $\gamma=-1$, $E$ takes the form
\be
E(\bfS)=\frac{1}{2}\int_{\calR} (|\xg\bfS|^2-\xg\bfS\cdot(\xg\bfS)^{\trans})\dv.
\label{E_special}
\ee
From \eqref{gtgt} and \eqref{gtgtt}, \eqref{E_special} reduces to 
\be
E(\bfS)=\frac{3}{4}\int_{\calR}|\skk(\xg\bfS)|^2\dv.
\label{Esk}
\ee
With reference to \eqref{Esk}, the first-order stationarity condition of $E$ can be stated as follows: If $\bfS$ is a stationary point of $E$, then for all admissible variations $\bfS_v$ of $\bfS$,
\be
\int_{\calR}\skk(\xg\bfS)\cdot\skk(\xg\bfS_v)\dv=0.
\label{equiv}
\ee
Thus, a solution $\bfS$ of the constrained extremization problem has, by \eqref{equiv}, the following, equivalent, characterization:
\be
\left.
\ba
\bfS&\in\calS,
\\[4pt]
\int_{\calR}\skk(\xg\bfS)\cdot\skk(\xg\bfS_c)\dv&=0, 
\\[4pt]
\int_{\calR}|\bfS|^2\dv&=1, 
\ea
\mskip3mu\right\}
\label{equiv1}
\ee
for all variations $\bfS_c$ satisfying
\be
\left.
\ba
\bfS_c&\in\calS,
\\[4pt]
\int_{\calR}\bfS\cdot\bfS_c\dv&=0,
\ea
\mskip3mu\right\}
\label{Sccond}
\ee
where the set $\calS$ was defined in \eqref{setS}. Similarly, a solution $\bfS$ of the unconstrained extremization problem has, by \eqref{equiv} and \eqref{normu}, the following, equivalent, characterization:
\be
\left.
\ba
\bfS&\in\calS_u,
\\[4pt]
\int_{\calR}\skk(\xg\bfS)\cdot\skk(\xg\bfS_u)\dv&=0, 
\\[4pt]
\int_{\calR}|\bfS|^2\dv&=1, 
\ea
\mskip3mu\right\}
\label{equiv2}
\ee
for all variations $\bfS_u$ satisfying
\be
\left.
\ba
\bfS_u&\in\calS_u,
\\[4pt]
\int_{\calR}\bfS\cdot\bfS_u\dv&=0,
\ea
\mskip3mu\right\}
\label{Sun}
\ee
where the set $\calS_u$ is given by \eqref{setSu}. We next aim to show that a spherically symmetric solution $\bfS$ of \eqref{equiv1} necessarily satisfies \eqref{equiv2}. To that end, we must demonstrate that \eqref{equiv2} holds for such a solution and all variations $\bfS_u$ satisfying \eqref{Sun}. Since, by \eqref{setS} and \eqref{setSu}, $\calS$ is a subset of $\calS_u$, it follows that any choice of $\bfS$ that satisfies \eqref{equiv1} trivially satisfies \eqref{equiv2}$_{1,3}$. Hence, we must confirm, in particular, that $\bfS$ satisfies \eqref{equiv2}$_{2}$.

Since $\calR$ is simply-connected and has a smooth boundary $\partial\calR$, by Proposition \ref{corollary} presented in \ref{app:decomposition}, any admissible variation $\bfS_u$ can be decomposed as
\be
\bfS_u=\tilde{\bfS}_c+\sym\xg\bfv,
\label{Sudec}
\ee
where the symmetric second-order tensor field $\tilde{\bfS}_c$ satisfies 
\be
\left.
\ba
\xd\tilde{\bfS}_c=\bf0\qquad&\text{on}\qquad\calR,
\\[4pt]
\tilde{\bfS}_c\bfn=\bf0\qquad&\text{on}\qquad\partial\calR,
\ea
\mskip3mu\right\}
\label{Sccond1}
\ee
and $\bfv$ is a differentiable vector field on $\calR$. Notice, by the divergence theorem, that
\be
\int_{\calR}\tilde{\bfS}_c\cdot(\sym\xg\bfv)\dv=\int_{\partial\calR}(\tilde{\bfS}_c\bfn)\cdot\bfv \da-\int_{\calR}\xd\tilde{\bfS}_c\cdot\bfv\dv=0,
\ee
from which it is evident that $\tilde{\bfS}_c$ and $\sym\xg\bfv$ are orthogonal with respect to the $L^2(\calR)$ scalar product. 

We next show that $\tilde{\bfS}_c$ satisfies \eqref{Sccond}. Since $\bfS_u$ belongs to the set $\calS_u$ defined by \eqref{setSu}, it is square-integrable and $E(\bfS_u)<\infty$. By \eqref{Sudec}, the same must be true for $\tilde{\bfS_c}$. Thus, 
\be
\int_{\calR}|\tilde{\bfS}_c|^2\dv<\infty \qquad \text{and} \qquad E(\tilde{\bfS}_c)<\infty.
\label{Scinf}
\ee
Furthermore, by the equivalence of the norm induced by $E$ and the $H^1(\calR)$ norm on $\calS$ as established in Proposition \ref{hatEH1}, it follows that $\bfv$ must be twice differentiable. The latter attribute of $\bfv$ will soon prove useful. Next, invoking \eqref{Sun}$_2$ and \eqref{Sudec}, and using the divergence theorem, the property \eqref{equiv1}$_1$, and the definition \eqref{setS}, we find that
\be
\ba
0=\int_{\calR} \bfS\cdot\bfS_u\dv&=\int_{\calR} \bfS\cdot\tilde{\bfS}_c\dv\,+\int_{\calR} \bfS\cdot(\sym\xg\bfv)\dv\\
&=\int_{\calR} \bfS\cdot\tilde{\bfS}_c\dv\,+\int_{\partial\calR}(\bfS\bfn)\cdot\bfv\da-\int_{\calR}\xd\bfS\cdot\bfv\dv\\
&=\int_{\calR} \bfS\cdot\tilde{\bfS}_c\dv.
\ea
\label{Sccond2}
\ee
Considering \eqref{Sccond1}, \eqref{Scinf}, and \eqref{Sccond2}, in conjunction with the definition \eqref{setS} of $\calS$, we confirm that $\tilde\bfS_c$ fulfills all conditions specified in \eqref{Sccond2} and, thus, qualifies as an admissible variation of the constrained extremization problem. Consequently, by \eqref{equiv1}$_2$, it follows that
\be
\int_{\calR}\skk(\xg\bfS)\cdot\skk(\xg\tilde{\bfS}_c)\dv=0.
\label{sksc}
\ee

We next consider the integral
\be
I=\int_{\calR}\skk(\xg\bfS)\cdot\skk(\xg(\sym\xg\bfv))\dv.
\label{I}
\ee
Since $\bfS$ is spherically symmetric, we notice from \eqref{GradS} that 
\begin{multline}
\skk(\xg\bfS)=\Big(\Sperp'-\frac{\Spar-\Sperp}{r}\Big)(2\bfe_{\varphi}\otimes\bfe_{\varphi}\otimes\bfe_r+2\bfe_{\vartheta}\otimes\bfe_{\vartheta}\otimes\bfe_r-\bfe_{\vartheta}\otimes\bfe_r\otimes\bfe_{\varphi}-\bfe_{r}\otimes\bfe_{\varphi}\otimes\bfe_{\varphi}\\
-\bfe_{\vartheta}\otimes\bfe_r\otimes\bfe_{\vartheta}-\bfe_{r}\otimes\bfe_{\vartheta}\otimes\bfe_{\vartheta}).
\label{skgs}
\end{multline}
On invoking the representation
\be
\bfv(r,\varphi,\vartheta)=v_r(r,\varphi,\vartheta)\,\bfe_r(\varphi,\vartheta)+v_{\varphi}(r,\varphi,\vartheta)\,\bfe_{\varphi}(\varphi,\vartheta)+v_{\vartheta}(r,\varphi,\vartheta)\,\bfe_{\vartheta}(\varphi,\vartheta)
\ee
of $\bfv$ relative to the spherical basis, it can be shown that
\be
\ba
\skk(\xg(\sym\xg\bfv))=\,&h_1(2\bfe_{\varphi}\otimes\bfe_{\varphi}\otimes\bfe_r-\bfe_{r}\otimes\bfe_{\varphi}\otimes\bfe_{\varphi}-\bfe_{\varphi}\otimes\bfe_r\otimes\bfe_{\varphi})
\\
+&h_2(2\bfe_{\vartheta}\otimes\bfe_{\vartheta}\otimes\bfe_r-\bfe_{r}\otimes\bfe_{\vartheta}\otimes\bfe_{\vartheta}-\bfe_{\vartheta}\otimes\bfe_r\otimes\bfe_{\vartheta})
\\
+&~\text{other components},
\ea
\label{gradv}
\ee
where `other components' refers to those directions in the spherical basis of a third-order tensor that are absent in the expression \eqref{skgs} of $\skk(\xg\bfS)$. Moreover, the auxiliary functions $h_1$ and $h_2$ of $v_r$, $v_{\varphi}$, and $v_{\vartheta}$ in \eqref{gradv} are given by
\be
\left.\ba
h_1(v_r,v_{\varphi},v_{\vartheta})&=\frac{1}{6r^2}\frac{\partial}{\partial\varphi}\Big(v_{\varphi}-\frac{\partial v_r}{\partial\varphi}+r\frac{\partial v_{\varphi}}{\partial r}\Big),
\\[4pt]
h_2(v_r,v_{\varphi},v_{\vartheta})&=\frac{\cot{\varphi}}{6r^2}\Big(v_{\varphi}-\frac{\partial v_r}{\partial\varphi}+r\frac{\partial v_{\varphi}}{\partial r}\Big)
+\frac{1}{6r^2\sin{\varphi}}\frac{\partial}{\partial\vartheta}\Big(v_{\vartheta}-\frac{1}{\sin{\varphi}}\frac{\partial v_r}{\partial\vartheta}+r\frac{\partial v_{\vartheta}}{\partial r}\Big).
\ea\mskip3mu\right\}
\label{h1h2a}
\ee
Since the elements of the spherical basis of a third-order tensor are mutually orthogonal, by \eqref{skgs}, the `other components' mentioned in \eqref{gradv} do not contribute in the scalar product of $\skk(\xg\bfS)$ and $\skk(\xg(\sym\xg\bfv))$ and, hence, to the integral $I$ in \eqref{I}.

For further analysis, we introduce two additional auxiliary functions $f$ and $g$ of $v_r$, $v_{\varphi}$, and $v_{\vartheta}$, given by
\be
\left.\ba
&f(v_r,v_{\varphi},v_{\vartheta})=v_{\varphi}-\frac{\partial v_r}{\partial\varphi}+r\frac{\partial v_{\varphi}}{\partial r},
\\[4pt]
&g(v_r,v_{\varphi},v_{\vartheta})=v_{\vartheta}-\frac{1}{\sin{\varphi}}\frac{\partial v_r}{\partial\vartheta}+r\frac{\partial v_{\vartheta}}{\partial r},
\ea\mskip3mu\right\}
\ee
so that, by \eqref{h1h2a}, 
\be
h_1=\frac{1}{6r^2}\frac{\partial f}{\partial\varphi}, \qquad h_2=\frac{f\cot{\varphi}}{6r^2}+\frac{1}{6r^2\sin{\varphi}}\frac{\partial g}{\partial\vartheta}.
\label{h1h2}
\ee
From \eqref{gradv} and \eqref{skgs}, it follows that
\be
\skk(\xg\bfS)\cdot\skk(\xg(\sym\xg\bfv))=6\Big(\Sperp'-\frac{\Spar-\Sperp}{r}\Big)(h_1+h_2)
\ee
or, from \eqref{h1h2}, that
\be
\skk(\xg\bfS)\cdot\skk(\xg(\sym\xg\bfv))=\frac{1}{r^2}\Big(\Sperp'-\frac{\Spar-\Sperp}{r}\Big)\Big(\frac{\partial f}{\partial\varphi}+f\cot{\varphi}+\frac{1}{\sin{\varphi}}\frac{\partial g}{\partial\vartheta}\Big).
\ee
Thus, by \eqref{I},
\begin{align}
I&=\int_{r_i}^{r_o}\frac{1}{r^2}\Big(\Sperp'-\frac{\Spar-\Sperp}{r}\Big)\int_{0}^{\pi}\int_0^{2\pi} \Big(\frac{\partial f}{\partial\varphi}+f\cot{\varphi}+\frac{1}{\sin{\varphi}}\frac{\partial g}{\partial\vartheta}\Big)r^2\sin{\varphi}\dvar\dvarphi\dr
\notag\\[4pt]
&=\int_{r_i}^{r_o}\Big(\Sperp'-\frac{\Spar-\Sperp}{r}\Big)\Big(\int_{0}^{\pi}\dvarphi\int_0^{2\pi}\frac{\partial g}{\partial\vartheta}\dvar+\int_{0}^{2\pi}\dvar\int_0^{\pi}\Big(\sin{\varphi}\frac{\partial f}{\partial\varphi}+\cos{\varphi}f\Big)\dvarphi\Big)\dr
\notag\\[4pt]
&=\int_{r_i}^{r_o}\Big(\Sperp'-\frac{\Spar-\Sperp}{r}\Big)\Big(\int_{0}^{\pi}\dvarphi\int_0^{2\pi}\frac{\partial g}{\partial\vartheta}\dvar+\int_{0}^{2\pi}\dvar\int_0^{\pi}\frac{\partial}{\partial\varphi}(f\sin{\varphi})\dvarphi\Big)\dr.
\end{align}
Since $\bfv$ is twice differentiable as noted earlier, $g$ is continuous at $\vartheta=0$, and it follows that
\be
\int_0^{2\pi}\frac{\partial g}{\partial\vartheta}\dvar=0.
\ee
Again, since $\bfv$ is twice differentiable, $f$ is finite at $\varphi=0$ and $\varphi=\pi$, and we find that
\be
\int_0^{\pi}\frac{\partial}{\partial\varphi}(f\sin{\varphi})\dvarphi=0.
\ee
Thus,
\be
I=\int_{\calR}\skk(\xg\bfS)\cdot\skk(\xg(\sym\xg\bfv))\dv=0.
\label{skgv}
\ee
Finally, combining \eqref{Sudec}, \eqref{sksc}, and \eqref{skgv} yields
\be
\int_{\calR}\skk(\xg\bfS)\cdot\skk(\xg\bfS_u)\dv=0.
\ee
It therefore follows from \eqref{equiv2} that $\bfS$ is a solution to the unconstrained extremization problem. Since it is a solution to both the constrained and unconstrained extremization problems, we conclude that $\bfmu=\bf0$.

It is important to note that both the spherical symmetry of $\bfS$ and the particular form of $E$, namely \eqref{E_special}, play decidedly critical roles towards the vanishing of $\bfmu$. Owing to the former, $\bfmu$ vanishes only on a single point in the viable parameter space; owing to the latter, $\bfmu$ may not vanish at $\beta=1$ and $\gamma=-1$ in the absence of spherical symmetry. 

We also note that in the absence of the constraint $\xd\bfS=\bf0$, \eqref{cond} does not hold and, consequently, that the dependence of $E$ on the parameters does not reduce to $p=\beta+\gamma$. Thus, the solution of the unconstrained extremization problem depends, in general, on both $p$ and $k$. This explains why $\bfmu$, as the reaction of the constraint $\xd\bfS=\bf0$, depends on both $p$ and $k$ while $\bfS$ and the reaction $\lambda$ to the normalization condition \eqref{normu} do not depend on $k$.

\section{Illustrative examples of fitting of spherically symmetric residual stress fields} \label{sec:examples}
As practical illustrations of the results established in Subsections \ref{sec:span} and \ref{sec:H1span}, where it was shown that the extremizers $\bfS_N$, $N\in\mathbb{N}$, serve as bases for $\bar{\calS}$ in the $L^2(\calR)$ norm and for $\calS$ in the $H^1(\calR)$ norm, respectively, we next use the sequences derived in Section \ref{sec:spherical} to fit the residual stress fields in (i) a spherical shell subjected to a non-uniform temperature distribution and (ii) a shrink-fitted spherical shell.

Before attending to these examples, we observe that if a stress field $\bfSigma$ belongs to $\calS$ then, since $\bfS_N, N\in\mathbb{N}$, span $\cal{S}$, there is a sequence $(b_N)_{N\in\mathbb{N}}$ such that
\be
\bfSigma=\sum_{N=1}^{\infty}b_N\bfS_N.
\label{ai}
\ee
Recalling from Subsection \ref{sec:ortho} the property that $\bfS_N, N\in\mathbb{N}$, are mutually orthonormal with respect to the scalar product $\langle\cdot,\cdot\rangle_{\sbbC}$ and noting that for $\mfC$ equal to the fourth-order identity tensor, $\langle\cdot,\cdot\rangle_{\sbbC}$ is simply the $L^2(\calR)$ scalar product, we conclude that the $\bfS_N$, $N\in\mathbb{N}$, computed in Section \ref{sec:spherical} are mutually orthonormal with respect to the $L^2(\calR)$ scalar product. It therefore follows from \eqref{ai} that
\be
b_N=\int_{\calR}\bfSigma\cdot\bfS_N\dv, \qquad N\in\mathbb{N}.
\label{aN}
\ee
The $n$-term approximation $\bfSigma_n$ to $\bfSigma$ is given by
\be
\bfSigma_n=\sum_{N=1}^n b_N \bfS_N,
\label{nterm}
\ee
with $b_N$, $N\in\mathbb{N}_n$, as defined in \eqref{aN}, where the set $\mathbb{N}_{n}$ is defined through \eqref{NN0}. We notice that the orthogonality of $\bfS_N$, $N\in\mathbb{N}$, in the $L^2(\calR)$ scalar product implies that the coefficient $b_N$ for a given $N\in\mathbb{N}_{n}$ does not change if $n$ is increased. Hence, for instance, given the $n$-term approximation $\bfSigma_n$, if we wish to compute the $(n+1)$-term approximation $\bfSigma_{n+1}$, we do not need to compute each coefficient $b_N,N\in\mathbb{N}_{n+1}$, afresh; we merely need to compute $b_{n+1}$. 

To quantify the closeness of $\bfSigma_n$ to $\bfSigma$, we define the relative $L^2(\calR)$ approximation error corresponding to the $n$-term approximation as
\be
e_{n_{L^2}}=\Big(\frac{\int_{\calR}|\bfSigma-\bfSigma_n|^2\dv}{\int_{\calR}|\bfSigma|^2\dv}\Big)^{\frac{1}{2}}.
\label{en}
\ee
Similarly, we define the relative $H^1(\calR)$ approximation error as
\be
e_{n_{H^1}}=\Big(\frac{\int_{\calR}(|\bfSigma-\bfSigma_n|^2+|\xg(\bfSigma-\bfSigma_n)|^2)\dv}{\int_{\calR}(|\bfSigma|^2+|\xg\bfSigma|^2)\dv}\Big)^{\frac{1}{2}}.
\label{ebarn}
\ee
Recall that the fields $\bfSigma$ and $\bfSigma_n$, the operator $\xg$, and the volume measure $\dv$ are dimensionless and, thus, $e_{n_{H^1}}$ in \eqref{ebarn} is consistently dimensionless. 

\subsection{Thermoelastic residual stress field}\label{sec:therm}
Consider an unloaded spherical shell occupying the region
\be
\calR=\{\bfx\in\calE:r_i<r=|\bfx-\bfo|<r_o\}.
\ee
Assume that the shell is made from a homogeneous, isotropic, linear elastic material with dimensionless bulk and shear moduli $\kappa$ and $\mu$, respectively, and dimensionless coefficient of thermal expansion $\alpha$\footnote{\label{constants}If $\kappa_p$, $\mu_p$, and $\alpha_p$ denote the corresponding physical constants, respectively, then the corresponding dimensionless constants are defined through $\kappa=\kappa_p/\varsigma$, $\mu=\mu_p/\varsigma$, and $\alpha=\alpha_pT_0$, where $\varsigma$ was introduced in \eqref{normalization} and $T_0$ is a reference physical temperature.}. The shell is subjected to a spherically symmetric dimensionless temperature-difference field $T$ satisfying $\Delta T \neq 0$, generating a spherically symmetric residual stress field of the form $\bfSigma=\mathit{\Sigma}_{\scriptscriptstyle{\parallel}}\bfe_r+\mathit{\Sigma}_{\scriptscriptstyle{\perp}}\bfPi$, with $\bfPi$ as defined in \eqref{PiOm}$_1$, in the shell. If $
\bfu=u\bfe_r$ denotes the resulting spherically symmetric displacement field, then $\bfSigma$ satisfies, in addition to \eqref{bvp_sph}$_{1,2}$, the constitutive equation
\be
\bfSigma=\kappa (\tr\bfE)\idem+2\mu\bfE^0,
\label{thermoS}
\ee
where $\bfE$ and $\bfE^0$ are given by
\be
\bfE=\sym\xg\bfu-\alpha T\idem \qquad \text{and} \qquad \bfE^0=\bfE-\frac{\tr\bfE}{3}\idem.
\label{udef}
\ee
Using the expression \eqref{Gradmu} for the gradient of a spherically symmetric vector field in \eqref{thermoS} yields
\be
\left.\ba
&{u(r)=-\frac{r((3\kappa-2\mu)\mathit{\Sigma}_{\scriptscriptstyle{\parallel}}(r)-(3\kappa+4\mu)\mathit{\Sigma}_{\scriptscriptstyle{\perp}}(r)-18\alpha\kappa\mu T(r))}{18\kappa\mu},}
\\[4pt]
&{u'(r)=\phantom{-}\frac{(3\kappa+\mu)\mathit{\Sigma}_{\scriptscriptstyle{\parallel}}(r)-(3\kappa-2\mu)\mathit{\Sigma}_{\scriptscriptstyle{\perp}}(r)+9\alpha\kappa\mu T(r)}{9\kappa\mu}.}
\ea\mskip3mu\right\}
\label{utherm}
\ee
Differentiating \eqref{utherm}$_1$, subtracting \eqref{utherm}$_2$ from the resulting equation, and eliminating $\mathit{\Sigma}_{\scriptscriptstyle{\perp}}$ using the equilibrium equation \eqref{bvp_sph}$_1$, we find that
\be
r\mathit{\Sigma}_{\scriptscriptstyle{\parallel}}''+4\mathit{\Sigma}_{\scriptscriptstyle{\parallel}}'=-\frac{36\alpha\kappa\mu T'}{3\kappa+4\mu}.
\label{Spartherm1}
\ee
We take $T$ to be of the simple linear form
\be
T(r)=\frac{cr}{r_o},
\label{dimlesstemp}
\ee
where $c\neq 0$ is a prescribed dimensionless constant. Notice that for $T$ defined in \eqref{dimlesstemp},
\be
\Delta T=\frac{2c}{r_or}\neq 0 \qquad\text{on}\qquad r_i<r<r_o.
\ee
Consequently, the resulting thermal strain $\alpha T\idem$ is incompatible, thus inducing a residual stress $\bfSigma$ in the shell. We find, by integrating \eqref{Spartherm1} and using the boundary conditions \eqref{bvp_sph}$_3$, that 
\be
\mathit{\Sigma}_{\scriptscriptstyle{\parallel}}(r)=\frac{9c\alpha\kappa\mu(r-r_i)(r_o-r)(r_i^2r_o^2+rr_ir_o(r_i+r_o)+r^2(r_i^2+r_ir_o+r_o^2))}{(3\kappa+4\mu)r_or^3(r_i^2+r_ir_o+r_o^2)}
\label{Spartherm}
\ee
and, by \eqref{bvp_sph}$_1$, that
\be
\mathit{\Sigma}_{\scriptscriptstyle{\perp}}(r)=\frac{9c\alpha\kappa\mu(r_i^3r_o^3-3r^4(r_i^2+r_ir_o+r_o^2)+2r^3(r_i^3+r_i^2r_o+r_ir_o^2+r_o^3))}{2(3\kappa+4\mu)r_or^3(r_i^2+r_ir_o+r_o^2)}.
\label{Sperptherm}
\ee

For illustration, we choose a spherical shell made of aluminium, with physical (i.e., dimensional) inner and outer radii \qty{0.5}{\metre} and \qty{1}{\metre}, respectively. The physical bulk and shear moduli of aluminium are taken from the literature \citep{samsonov} to be \qty{7.6e10}{\newton\per\metre\squared} and \qty{2.7e10}{\newton\per\metre\squared}, respectively. Furthermore, we choose the constant $\varsigma$ introduced in $\eqref{normalization}$ to be equal to the shear modulus, that is, $\varsigma=$\qty{2.7e10}{\newton\per\metre\squared}. Then, by \eqref{change}$_{1}$, \eqref{Ldef}, and Footnote \ref{constants}, with the choice $k_0=6/(7\pi)$ in \eqref{Ldef}, the inner and outer dimensionless radii and the dimensionless bulk and shear moduli become 
\be
r_i=0.5, \qquad r_o=1, \qquad \kappa=2.8, \qquad \text{and} \qquad \mu=1,
\label{rkmu}
\ee
respectively. 

Finally, to obtain the dimensionless thermal expansion coefficient $\alpha$, we take the characteristic reference temperature $T_0$ to be the melting point $\qty{832}{K}$ of aluminium \citep{samsonov}  and the physical thermal expansion coefficient $\alpha_p$ to be \qty{2.1e-5}{\per\kelvin} \citep{samsonov} from the literature. It then follows from Footnote \ref{constants} that
\be
\alpha=1.75\times10^{-2}.
\label{tec}
\ee 
Furthermore, we choose the physical temperature-difference field $T_p$ to be
\be
T_p(r)=\frac{T_0}{9}\frac{r}{r_o},
\label{barT}
\ee
whereby the dimensionless temperature-difference field is given by
\be
T(r)=\frac{\bar{T}(r)}{T_0}=\frac{1}{9}\frac{r}{r_o}.
\label{dimlesstemp2}
\ee
Note that the choice \eqref{barT} ensures that \eqref{dimlesstemp2} is consistent with \eqref{dimlesstemp}, as long as $c$ is given by
\be
c=\frac{1}{9}.
\label{cinT}
\ee 
We also note that the choice \eqref{barT} implies that if the unstressed shell is at the room temperature, say, 300 K, then the temperature in the residually stressed shell varies linearly from 346.2 K on the inner boundary to 392.4 K on the outer boundary of the shell.

Using \eqref{rkmu}, \eqref{tec}, and \eqref{cinT} in \eqref{Spartherm} and \eqref{Sperptherm}, we obtain $\bfSigma$.  We plot $\bfSigma$ and its three-term approximations, obtained using \eqref{nterm} and \eqref{aN}, against $r$ in Figure \ref{fig:sEx1}. The approximations correspond to three choices of the parameter combinations: $\beta=0,\gamma=0$; $\beta=1,\gamma=0$; and $\beta=1,\gamma=1$. We find that the approximations for all three parameter combinations are reasonably good, considering they are obtained by taking only $n=3$ terms in the expansion \eqref{nterm}. For $n\ge 4$, the approximations are visually indistinguishable from $\bfSigma$, which is why we have chosen to plot three-term approximations in Figure \ref{fig:sEx1}. 
\begin{figure}[t!]
    \centering
    \begin{subfigure}[b]{0.49\textwidth}
         \centering
         \includegraphics[width=\textwidth]{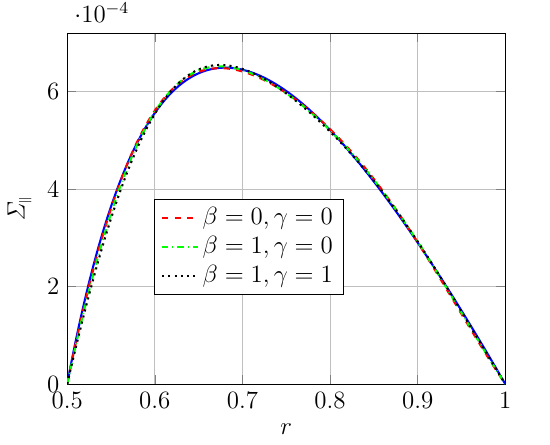}
     \end{subfigure}
     \begin{subfigure}[b]{0.49\textwidth}
         \centering
         \includegraphics[width=\textwidth]{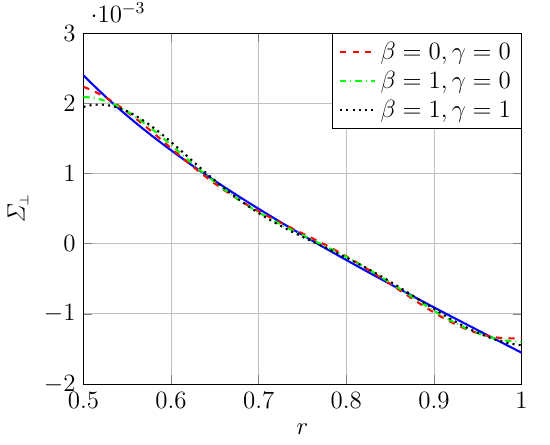}
     \end{subfigure}
    \caption{Fitting of the thermoelastic residual stress field with components given by \eqref{Spartherm} and \eqref{Sperptherm}. Left: Plots of $\mathit{\Sigma}_{\scriptscriptstyle{\parallel}}$ (solid blue curve) and its three-term approximations versus $r$. Right: Plots of $\mathit{\Sigma}_{\scriptscriptstyle{\perp}}$ (solid blue curve) and its three-term approximations versus $r$. The approximations correspond to the three parameter choices reported in the legends.}
    \label{fig:sEx1}
\end{figure}

The approximation errors $e_{n_{L^2}}$ and $e_{n_{H^1}}$, introduced in \eqref{en} and \eqref{ebarn}, respectively, are plotted against $n$ on a log-log scale in Figure \ref{fig:thermoen}. We find, by measuring the slope of the curves at large $n$, that $e_{n_{L^2}}$ and $e_{n_{H^1}}$ decay approximately as $1/n^{1.5}$ and $1/n^{0.5}$, respectively, for all three parameter combinations. To understand the reason for these decay rates, we introduce $s_{n_{L^2}}=e_{n_{L^2}}^2$, and notice, with reference to \eqref{bvp_nondim_iso}$_5$, \eqref{nterm}, and \eqref{en}, that $s_{n_{L^2}}$ can be written as 
\be
s_{n_{L^2}}=1-\frac{\sum_{i=1}^n b_i^2}{\|\bfSigma\|_{L^2(\calR)}^2},
\ee
where the coefficients $b_i,i\in\mathbb{N}_n,$ are given by \eqref{aN}. Accordingly, for sufficiently large values of $n$, 
\be
\frac{\text{d}s_{n_{L^2}}}{\text{d}n}\approx s_{n+1_{L^2}}-s_{n_{L^2}}=-\frac{b_{n+1}^2}{\|\bfSigma\|_{L^2(\calR)}^2}.
\label{den}
\ee
Therefore, the corresponding rate of decay of $s_{n_{L^2}}$ depends on that of the magnitude of the coefficient $b_N$. To determine the rate of decay of $|b_N|$, we plot $|b_N|, N\in\mathbb{N}_{100}$, versus $N$ on a log-log scale in Figure \ref{fig:aN}. In that figure, we see that there are two subsequences within the sequence $(|b_N|)_{N\in\mathbb{N}_{100}}$, both of which decay approximately as $1/N^2$ for large $N$. Plugging this in \eqref{den} and integrating, we find that $s_{n_{L^2}}$ decays approximately as $1/n^3$ for large $n$. Accordingly, $ e_{n_{L^2}}=s_{n_{L^2}}^{1/2}$ decays approximately as $n^{-3/2}$ for large $n$, as revealed by the left panel of Figure \ref{fig:thermoen}. Using similar heuristics, we find that for large $n$,
\be
\frac{\text{d}s_{n_{H^1}}}{\text{d}n}\approx s_{n+1_{H^1}}-s_{n_{H^1}}=-k b_{n+1},
\ee
where $s_{n_{H^1}}=e_{n_{H^1}}^2$ and $k$ is a constant. It follows that $e_{n_{H^1}}$ decays approximately as $n^{-1/2}$ for large $n$, an observation in agreement with that from the right panel of Figure \ref{fig:thermoen}.
\begin{figure}[t!]
    \centering
         \includegraphics[width=\textwidth]{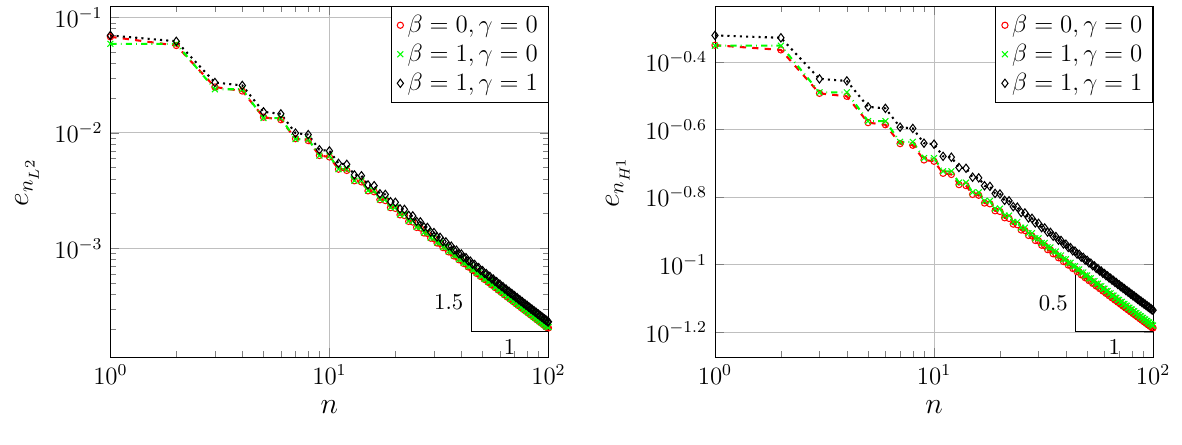}
    \caption{Plots of the approximation errors $e_{n_{L^2}}$ (left) and $e_{n_{H^1}}$ (right panel) versus $n$ corresponding to the thermoelastic residual stress field with components given by \eqref{Spartherm} and \eqref{Sperptherm}. The approximation errors correspond to the three parameter choices reported in the legends.}
    \label{fig:thermoen}
\end{figure}

An interesting feature borne out by Figure \ref{fig:thermoen} is the step-like decays of the error measures $e_{n_{L^2}}$ and $e_{n_{H^1}}$. Such a decay indicates that the even-numbered extremizers, namely $\bfS_2,\bfS_4,\bfS_6,\dots$, have relatively small contributions in the approximation of $\bfSigma$. This is confirmed by Figure \ref{fig:aN}, which reveals that while, on average, both odd- and even-numbered coefficients decay as $1/N^2$ for large $N$, the latter are relatively smaller, magnitude-wise. This contrast is particularly evident for the parameter choice $\beta=1$ and $\gamma=0$, where the even-numbered coefficients have already decayed to values close to the machine precision by about $N=20$ and seem to saturate after that.
\begin{figure}[t!]
    \centering
         \includegraphics[width=0.5\textwidth]{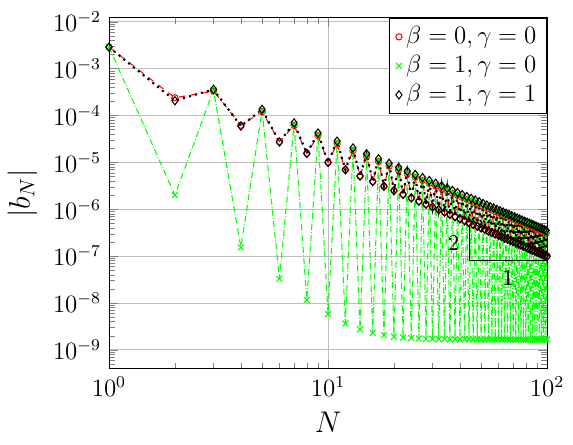}
    \caption{Plot of $|b_N|$, with the coefficient $b_N$ as given by \eqref{aN}, versus $N$ for the thermoelastic residual stress field corresponding to the three parameter choices reported in the legend.}
    \label{fig:aN}
\end{figure} 

The relatively larger contribution of the odd-numbered extremizers can be explained by noting, for instance, that $\mathit{\Sigma}_{\scriptscriptstyle{\parallel}}$ (left panel of Figure \ref{fig:sEx1}) is relatively symmetric about $r=(r_i+r_o)/2$, and so are the radial components of the odd-numbered extremizers (left panels of the first and third rows in Figure \ref{fig:p0to5}). In contrast, the radial components of the even-numbered extremizers (left panels of the second and fourth rows in Figure \ref{fig:p0to5}) are relatively anti-symmetric with respect to $r=0.75$ and, thus, make relatively smaller contributions. Similarly, $\mathit{\Sigma}_{\scriptscriptstyle{\perp}}$ (right panel of Figure \ref{fig:sEx1}) is relatively anti-symmetric about $ r=0.75$, and so are the azimuthal components of the odd-numbered extremizers (right panels of the first and third rows in Figure \ref{fig:p0to5}). In contrast, the azimuthal components of the even-numbered extremizers (right panels of the second and fourth rows in Figure \ref{fig:p0to5}) are relatively symmetric with respect to $r=0.75$.

\subsection{Shrink-fit residual stress field}\label{sec:shrink}
For the second example, we consider the residual stress field in a shell obtained by shrink-fitting an inner spherical shell of inner radius $r_i$ and notional outer radius $r_m$ and an outer spherical shell of notional inner radius $r_m$ and outer radius $r_o$. The two constituent shells have a small radial interference of $\delta$ and are taken to be made of the same homogeneous material obeying linear isotropic elasticity with dimensionless bulk and shear moduli $\kappa$ and $\mu$, respectively. 

Following the same procedure as that in the previous example, we find that the radial stress in the inner shell $\mathit{\Sigma}_{\scriptscriptstyle{\parallel}i}$ satisfies
\be
\left.\ba
r\mathit{\Sigma}_{\scriptscriptstyle{\parallel}i}''+4\mathit{\Sigma}_{\scriptscriptstyle{\parallel}i}'=0 \qquad &\text{on} \qquad r_i< r< r_m,
\\[4pt]
\mathit{\Sigma}_{\scriptscriptstyle{\parallel}i}=0 \qquad &\text{on} \qquad r=r_i,
\\[4pt]
\mathit{\Sigma}_{\scriptscriptstyle{\parallel}i}=-p_0 \qquad &\text{on} \qquad r=r_m,
\ea\mskip3mu\right\}
\label{srri}
\ee
where the pressure $p_0$ at the notional interface $r_m$ is unknown as yet. Similarly, the radial stress in the outer shell $\mathit{\Sigma}_{\scriptscriptstyle{\parallel}o}$ satisfies
\be
\left.\ba
r\mathit{\Sigma}_{\scriptscriptstyle{\parallel}o}''+4\mathit{\Sigma}_{\scriptscriptstyle{\parallel}o}'=0 \qquad &\text{on} \qquad r_m< r< r_o,
\\[4pt]
\mathit{\Sigma}_{\scriptscriptstyle{\parallel}o}=-p_0 \qquad &\text{on} \qquad r=r_m,
\\[4pt]
\mathit{\Sigma}_{\scriptscriptstyle{\parallel}o}=0 \qquad &\text{on} \qquad r=r_o.
\ea\mskip3mu\right\}
\label{srro}
\ee
The quantities $\mathit{\Sigma}_{\scriptscriptstyle{\parallel}i}$ and $\mathit{\Sigma}_{\scriptscriptstyle{\parallel}o}$ that enter \eqref{srri}$_{2,3}$ and \eqref{srro}$_{2,3}$, respectively, are found in terms of $p_0$ from \eqref{srri} and \eqref{srro} to be
\be
\mathit{\Sigma}_{\scriptscriptstyle{\parallel}i}(r)=-p_0\Big(\frac{1}{r_i^3}-\frac{1}{r_m^3}\Big)^{-1}\Big(\frac{1}{r_i^3}-\frac{1}{r^3}\Big)
\qquad\text{and}\qquad
\mathit{\Sigma}_{\scriptscriptstyle{\parallel}o}(r)=-p_0\Big(\frac{1}{r_m^3}-\frac{1}{r_o^3}\Big)^{-1}\Big(\frac{1}{r^3}-\frac{1}{r_o^3}\Big).
\label{srrio}
\ee
Finally, $p_0$ is found by using the compatibility condition 
\be
u_o(r_m)-u_i(r_m)=\delta,
\label{pdelta}
\ee
where $u_i$ and $u_o$ denote the displacement fields of the inner and outer shells, respectively. By the same calculations that led to \eqref{utherm}, it follows that
\be
u_i(r)=-\frac{p_0r_m^3(3r_i^3\kappa+4r^3\mu)}{12\kappa\mu r^2(r_m^3-r_i^3)},\qquad
u_o(r)=\frac{p_0r_m^3(3r_o^3\kappa+4r^3\mu)}{12\kappa\mu r^2(r_o^3-r_m^3)}.
\label{uiuo}
\ee
Evaluating $u_i$ and $u_o$ at the interface by substituting $r=r_m$ in \eqref{uiuo} and using \eqref{pdelta}, we find that $p_0$ is given by
\be
p_0=\frac{12\mskip1mu\delta\kappa\mu(r_o^3-r_m^3)(r_m^3-r_i^3)}{(3\kappa+4\mu)r_m^4(r_o^3-r_i^3)}.
\label{pint}
\ee
We then substitute \eqref{pint} in \eqref{srrio} to obtain $\mathit{\Sigma}_{\scriptscriptstyle{\parallel}i}$ and $\mathit{\Sigma}_{\scriptscriptstyle{\parallel}o}$. Thereafter, by \eqref{bvp_sph}$_1$, it follows that
\be
\mathit{\Sigma}_{\scriptscriptstyle{\perp}i}(r)=-p_0\Big(\frac{1}{r_i^3}-\frac{1}{r_m^3}\Big)^{-1}\Big(\frac{1}{r_i^3}+\frac{1}{2r^3}\Big),\qquad
\mathit{\Sigma}_{\scriptscriptstyle{\perp}o}(r)=p_0\Big(\frac{1}{r_m^3}-\frac{1}{r_o^3}\Big)^{-1}\Big(\frac{1}{r_o^3}+\frac{1}{2r^3}\Big),
\label{sttio}
\ee
with $p_0$ as given by \eqref{pint}. Notice, from \eqref{sttio}, that $\mathit{\Sigma}_{\scriptscriptstyle{\perp}}$ has a jump
\be
[\![ \mathit{\Sigma}_{\scriptscriptstyle{\perp}} ]\!]=\mathit{\Sigma}_{\scriptscriptstyle{\perp}o}(r_m)-\mathit{\Sigma}_{\scriptscriptstyle{\perp}i}(r_m)=\frac{3\mskip1mu p\mskip1mu r_m^3(r_o^3-r_i^3)}{2(r_o^3-r_m^3)(r_m^3-r_i^3)}
\ee
across the interface, with the consequence that $\bfSigma$ is discontinuous. Therefore, $\bfSigma$ does not belong to $\calS$ since its gradient is not square-integrable and, thus, $E(\bfSigma)$ is unbounded. However, it belongs to $\bar{\calS}$, which, as shown in Subsection \ref{sec:span}, is spanned by $\bfS_N,N\in\mathbb{N}$, in the $L^2(\calR)$ norm. We illustrate below that $\bfSigma$ belongs to the span of $\bfS_N,N\in\mathbb{N}$.

We choose the physical dimensions and material properties of the shrink-fitted system to be the same as those in the example in Subsection \ref{sec:therm}. Accordingly, $r_i=0.5$, $r_o=1$, $\kappa=3$, and $\mu=1$. Furthermore, we choose the radius $r_m$ of the interface to be the average of $r_i$ and $r_o$; so, $r_m=0.75$. Finally, we choose the radial interference $\delta$ to be one percent of the outer radius; hence, $\delta=0.01$. Plugging these values in \eqref{srrio} and \eqref{sttio}, with $p_0$ computed using \eqref{pint}, we obtain $\bfSigma$.

The plots of $\bfSigma$ and its ten- and hundred-term approximations obtained using \eqref{nterm} and \eqref{aN} against $r$ are shown in the top and bottom panels, respectively, of Figure \ref{fig:sEx2}. The approximations correspond to three choices of the parameter combinations: $\beta=0,\gamma=0$; $\beta=1,\gamma=0$; and $\beta=1,\gamma=1$. It is evident that the hundred-term approximations are closer to $\bfSigma$ than the ten-term approximations for all three parameter combinations; indeed, per the theory, upon taking infinitely many terms, the approximations must exactly match with $\bfSigma$. The convergence rate with respect to the number of terms $n$ in the approximation is notably slower than in the previous example, where the three-term approximations were nearly visually indistinguishable from the true stress. This is evidently because the true stress $\bfSigma$ is discontinuous in the current example but is smooth in the previous example. We also notice the expected Gibbs phenomenon in the right panels of Figure \ref{fig:sEx2}, especially in the bottom-right panel.
\begin{figure}[t!]
    \centering
         \includegraphics[width=\textwidth]{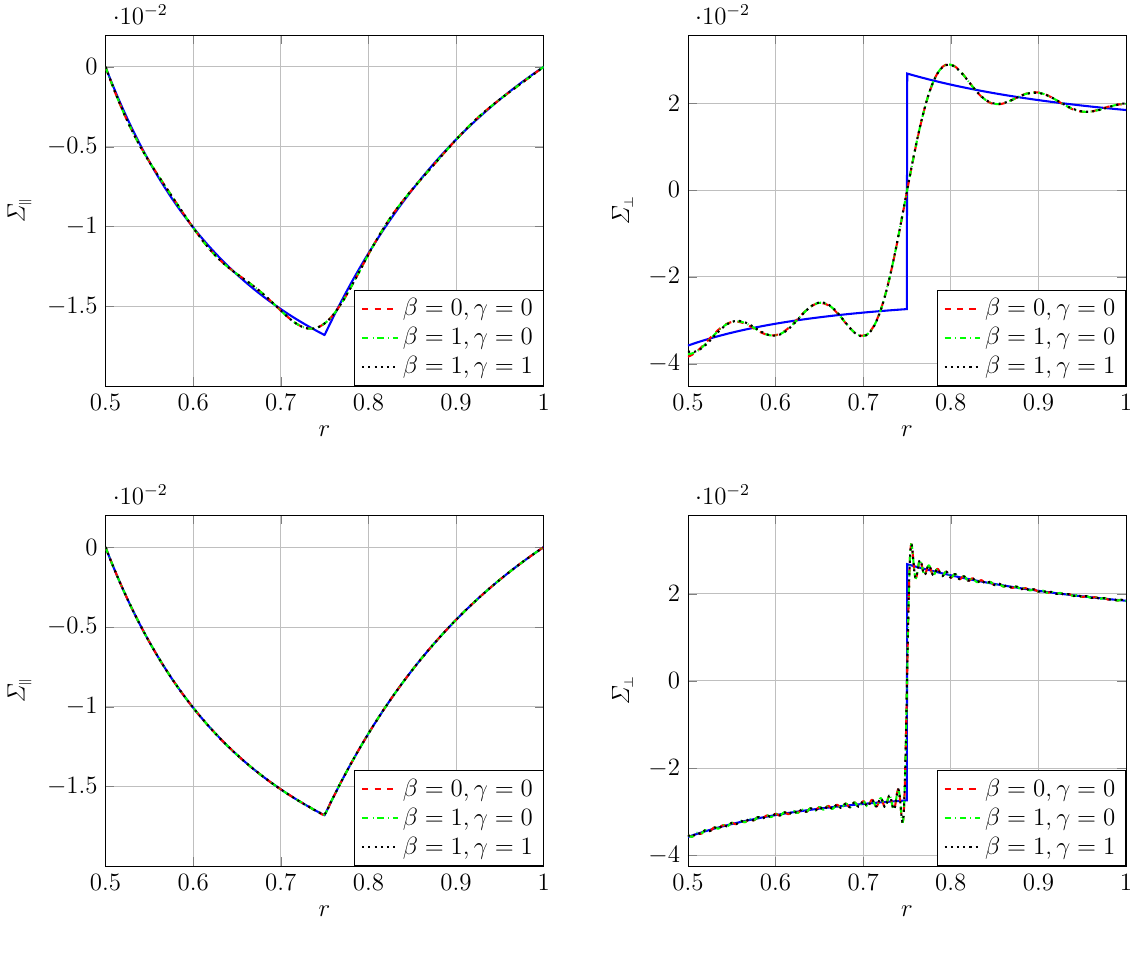}
    \caption{Fitting of the shrink-fit residual stress field with components given by \eqref{srrio} and \eqref{sttio}. Top: Plots of $\mathit{\Sigma}_{\scriptscriptstyle{\parallel}}$ and $\mathit{\Sigma}_{\scriptscriptstyle{\perp}}$ (solid blue curves) and their ten-term approximations versus $r$. Bottom: Plots of $\mathit{\Sigma}_{\scriptscriptstyle{\parallel}}$ and $\mathit{\Sigma}_{\scriptscriptstyle{\perp}}$ (solid blue curves) and their hundred-term approximations versus $r$. The approximations correspond to the three parameter choices reported in the legends.}
    \label{fig:sEx2}
\end{figure}

We plot the approximation error $e_{n_{L^2}}$ against $n$ on a log-log scale in Figure \ref{fig:shrinken}. We find, by measuring the slope of $e_{n_{L^2}}$ at large $n$, that $e_{n_{L^2}}$ decays approximately as $n^{-1/2}$ for all three parameter combinations. 
\begin{figure}[t!]
    \centering
         \includegraphics[width=0.5\textwidth]{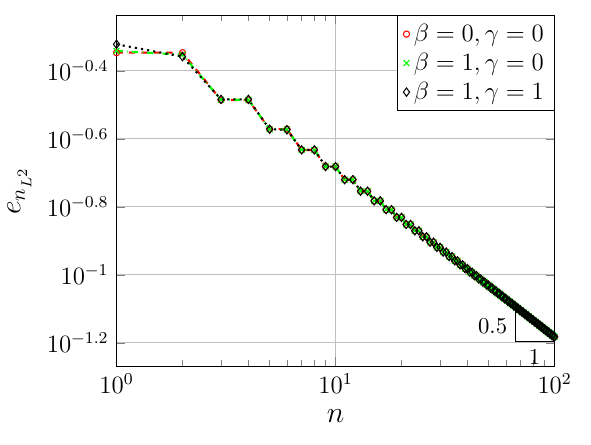}
    \caption{Plot of the approximation error $e_{n_{L^2}}$ versus $n$ corresponding to the shrink-fit residual stress field with components given by \eqref{srrio} and \eqref{sttio}. The approximation error corresponds to the three parameter choices reported in the legends.}
    \label{fig:shrinken}
\end{figure}

Reminiscent of the previous example, we again see a step-like decay in the approximation error $e_{n_{L^2}}$ in Figure \ref{fig:shrinken}, indicating that the odd-numbered modes, namely $\bfS_1,\bfS_3,\bfS_5,\dots$, contribute relatively more than the even-numbered modes. This is confirmed by plotting the magnitudes of the coefficients $b_N,N\in\mathbb{N}_{100}$, in Figure \ref{fig:aNEx2}: for all parameter choices considered, we notice that the odd-numbered coefficients have relatively greater magnitudes than those of the even-numbered coefficients. Also, for the parameter choice $\beta=1$ and $\gamma=0$, we notice that the even-numbered sequence, namely $b_2,b_4,b_6,\dots$, consists of two further subsequences of relatively different magnitudes: $b_2,b_6,b_{10},\dots$, and $b_4,b_8,b_{12},\dots$. The same is observed for the parameter choice $\beta=1$ and $\gamma=1$. Furthermore, we find that for all parameter choices considered, the odd-numbered sequences decay approximately as $1/N$ for large $N$. Similarly, the even-numbered sequences, or subsequences thereof, decay approximately as $1/N^2$ for large $N$.
\begin{figure}[t!]
    \centering
         \includegraphics[width=0.5\textwidth]{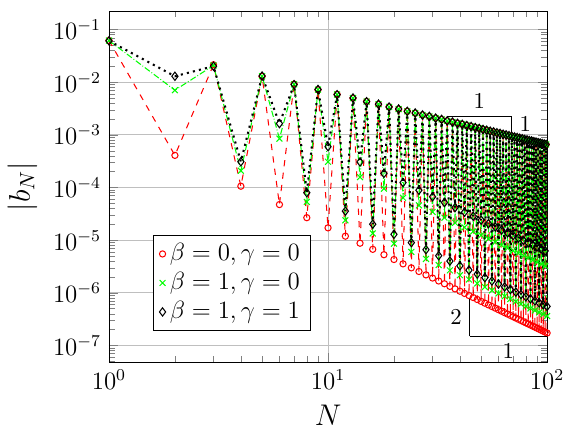}
    \caption{Plot of $|b_N|$, with the coefficient $b_N$ as given by \eqref{aN}, versus $N$ for the shrink-fit residual stress field corresponding to the three parameter choices reported in the legend.}
    \label{fig:aNEx2}
\end{figure} 
\begin{figure}[t!]
    \centering
         \includegraphics[width=1\textwidth]{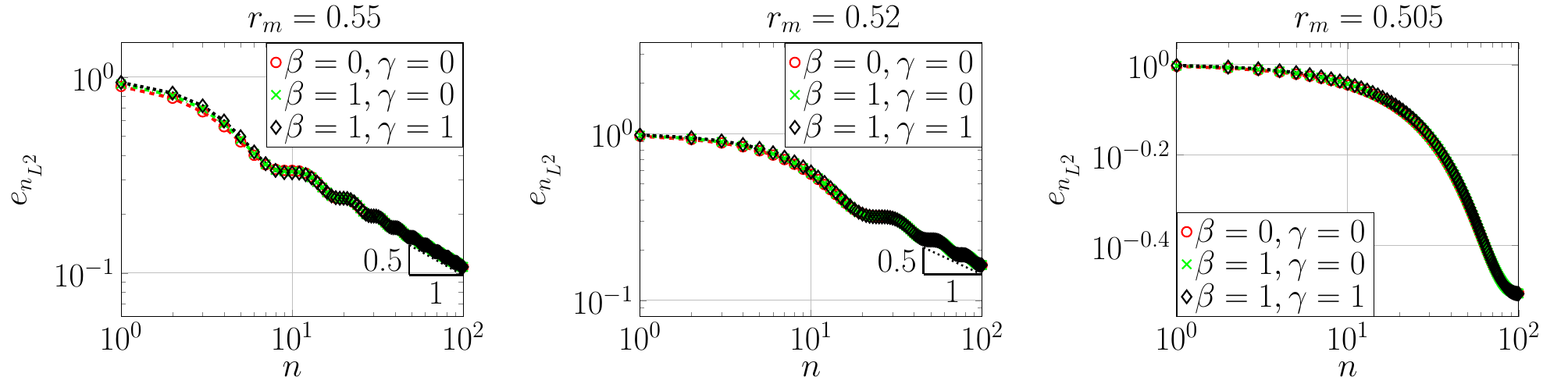}
    \caption{Plots of the approximation error $e_{n_{L^2}}$ versus $n$ for the shrink-fit residual stress fields corresponding to three representative choices of the interface radius $r_m$: $0.55$, $0.52$, and $0.505$.}
    \label{fig:rmInf}
\end{figure} 

Finally, we investigate the influence of the radius $r_m$ of the interface on the convergence properties of the sequences. We plot the approximation error $e_{n_{L^2}}$ against $n$ on a log-log scale in Figure \ref{fig:rmInf} for three other representative values of $r_m$: $0.55$, $0.52$, and $0.505$. We notice from the left panel of Figure \ref{fig:rmInf} corresponding to $r_m=0.55$ that for all three parameter choices, starting from around $n=10$, $e_{n_{L^2}}$ exhibits oscillatory behaviour about a line with slope $-0.5$; the amplitude of these oscillations reduces as $n$ increases. We thus conclude that for all three parameter choices, the approximation error $e_{n_{L^2}}$ corresponding to $r_m=0.55$ decays, on average, as $n^{-1/2}$ for large $n$. Note that this decay rate is the same as that observed for $r_m=0.75$ from Figure \ref{fig:shrinken}. We arrive at the same conclusion from the middle panel of Figure \ref{fig:rmInf} corresponding to $r_m=0.52$. However, in this case, the oscillatory behaviour of $e_{n_{L^2}}$ begins at around $n=20$. Finally, in the right panel of Figure \ref{fig:rmInf} corresponding to $r_m=0.505$, we observe that while $e_{n_{L^2}}$ decays monotonically, it does not exhibit oscillatory behavior within the considered range $1\leq n \leq 100$ of $n$. It is anticipated that even for $r_m=0.505$, $e_{n_{L^2}}$ exhibits oscillatory behaviour with shrinking amplitude, with an average decay rate of $n^{-1/2}$, past a sufficiently large value $n_{\text{osc}}>100$ of $n$. 

Both examples considered in this section reveal that in the spherically symmetric case, the extremizers corresponding to different viable choices of parameters $\beta$ and $\gamma$ exhibit similar qualitative behaviour with regard to capturing given residual stress fields. On this evidence, as stated in the introduction, there is no reason to prefer the sequence obtained by \citet{tiwari2020basis} or, in fact, any sequence of the family obtained herein. However, it might be worth examining in future studies whether certain sequences are more optimal than others toward a given objective.

\section{Conclusions} \label{sec:conclusions}
We obtained a family of sequences that span the set of all square-integrable residual stress fields defined on a given bounded three-dimensional region $\calR$ and its boundary. These sequences are obtained by extremizing the most general positive-definite, quadratic functional $E$ of the stress-gradient. This extremization is carried out over a subset, corresponding to a certain prescribed norm, of the set $\calS$, which consists of all residual stress fields corresponding to a finite $E$. Each choice of a functional with a positive-definite integrand yields a complete sequence.

The sequences exhibit several desirable properties. The elements of each sequence are mutually orthonormal in the scalar product induced by the fourth-order coefficient tensor $\mfC$ in the normalization condition and mutually orthogonal in the scalar product induced by the sixth-order coefficient tensor $\osix{A}$ in the integrand of $E$. Additionally, each sequence spans the set $\calS$ in the $H^1(\calR)$ norm, and the set $\bar{\calS}$ in the $L^2(\calR)$ norm, $\bar{\calS}$ being the $L^2(\calR)$ completion of $\calS$.

Stipulating that $\osix{A}$ be homogeneous and isotropic and $\mfC$ be proportional to the identity tensor, we found that the resulting boundary-value problem involves three independent constant parameters. Upon eliminating, without loss of generality, one of these parameters through a suitable normalization, we noticed that choices of the remaining normalized parameters which ensure that the integrand of $E$ is positive-definite belong to a semi-infinite strip in the parameter space.

Within this homogeneous, isotropic case, we next chose $\calR$ to be a sphere and found analytical solutions for the spherically symmetric sequences. Surprisingly, we found that the dependence of these sequences over the semi-infinite strip of viable parameters collapses to one of its finite edges. On one extremity of this edge, the sequences are sinusoidal functions with radially varying amplitudes and phases. Furthermore, as we transition across this edge, the sequences change only slightly, implying that the solutions possess a nearly harmonic structure across the entire strip. Surprisingly, at an interior point of the strip, the constraint $\xd\bfS=\bf0$ is satisfied trivially, with the consequence that the Euler--Lagrange equation reduces to the Helmholtz equation. 

Lastly, we approximated two standard spherically symmetric residual stress fields using three different spherically symmetric sequences, demonstrating that the stress fields indeed lie within the span of all three sequences, consistent with the theory. Moreover, we found that the three sequences exhibit similar results with regard to capturing the candidate stresses. Notably, the orthogonality property of the sequences was found to facilitate the development of a computationally efficient framework for obtaining approximations of the candidate stresses.

The generality of our results opens avenues for future exploration into the implications of different choices of the coefficient tensors $\osix{A}$ and $\mfC$. In particular, we anticipate that certain choices of $\osix{A}$ and $\mfC$ may be more optimal than others toward a given objective. Note that we can tweak $\osix{A}$ and $\mfC$ in two distinct ways: by changing their symmetry class or by altering their spatial distribution. 

For instance, to capture residual stresses primarily concentrated near the surface, an appropriate selection of $\osix{A}$ and $\mfC$ may yield sequences with similar characteristics. A typical example of boundary-confined residual stress is found in shot peening, a process that induces compressive residual stresses primarily within a boundary layer significantly narrower---typically ranging from ten to a hundred times smaller---than the overall dimensions of the specimen \citep{wang1998compressive}. Similarly, when representing residual stresses in a component crafted from a material with a particular anisotropy, like single-crystal turbine blades casted from face-centered cubic nickel superalloys \citep{arakere2001fretting} or spider silk assembled from axially oriented $\beta$-crystallites \citep{knight2002biological}, choices of $\osix{A}$ and $\mfC$ corresponding to the same symmetry class may perhaps be more optimal. Yet another example illustrating the suitability of localized and anisotropic sequences is that provided by the nanocrystalline diamond films grown on glass, which are residually stressed due to the mismatch between the coefficients of thermal expansion of diamond and glass, and which have columnar architecture.
\drop{
\section*{CRediT authorship contribution statement}
\textbf{Sankalp Tiwari}: Conceptualization, Formal analysis, Investigation, Methodology,  Writing -- original draft, Writing -- review \& editing. \textbf{Eliot Fried}: Conceptualization, Formal analysis, Investigation, Methodology,  Writing -- original draft, Writing -- review \& editing.

\section*{Declaration of competing interest}
The authors declare that they have no known competing financial interests or personal relationships that could have appeared to influence the work reported in this paper.

\section*{Data availability}
No data was used for the research described in the article.

\section*{Acknowledgements}
The authors express their gratitude for support provided by the Okinawa Institute of
Science and Technology Graduate University, funded by the Cabinet Office of the Government of Japan.
}

\appendix
\section{Proof of equivalence of the norms \boldmath{$\|\cdot\|_{\sbbA}$} and \boldmath{$\|\cdot\|_{\sbbC}$} with the \boldmath{$H^1(\calR)$} and \boldmath{$L^2(\calR)$} norms, respectively}\label{app:equiv}
\begin{proposition}\label{hatEH1}
The norm $\|\cdot\|_{\sbbA}$ is equivalent to the $H^1(\calR)$ norm on $\calS$.
\end{proposition}
\begin{proof}
Note that the $H^1(\calR)$ norm of $\bfS$ is defined by
\be
\|\bfS\|_{H^1(\calR)}=\Big(\int_{\calR}|\bfS|^2\dv+\int_{\calR}|\xg\bfS|^2\dv\Big)^{\frac{1}{2}}.
\label{H1norm}
\ee 
Since, as shown by \citet{hoger1986determination}, residual stresses have zero mean, by Poincar\'e's inequality there exists a constant $C$, that depends only on $\calR$, such that 
\be
\Big(\int_{\calR}|\bfS|^2\dv\Big)^{\frac{1}{2}}\le C \Big(\int_{\calR}|\xg\bfS|^2\dv\Big)^{\frac{1}{2}} 
\label{poincare}
\ee
for all $\bfS\in\calS$. From \eqref{H1norm} and \eqref{poincare}, it follows that 
\be
\int_{\calR}|\xg\bfS|^2\dv\le\|\bfS\|_{H^1(\calR)}^2\le(C^2+1)\int_{\calR}|\xg\bfS|^2\dv.
\label{Ehatineq1}
\ee

Next, by the Cauchy--Schwarz inequality, 
\be
\xg\bfS\cdot\osix{A}\mskip2mu[\xg\bfS]\le (|\osix{A}\mskip2mu[\xg\bfS]|^2)^{\frac{1}{2}}(|\xg\bfS|^2)^{\frac{1}{2}}.
\label{CS}
\ee
Next, we exploit the inequality of arithmetic and geometric means to infer that
\be
|\osix{A}\mskip2mu[\xg\bfS]|^2\le (27A_{\text{m}})^2 |\xg\bfS|^2,
\label{AMGM}
\ee
where $A_{\text{m}}$ denotes maximum of the absolute values of the Cartesian components of $\osix{A}$:
\be
A_{\text{m}}=\text{max}(|A_{ijklmn}|), \qquad i,j,k,l,m,n\in\{1,2,3\}.
\ee
Combining \eqref{CS} and \eqref{AMGM} yields
\be
\xg\bfS\cdot\osix{A}\mskip2mu[\xg\bfS]\le 27A_{\text{m}}|\xg\bfS|^2.
\label{CSAMGM}
\ee
Denoting the supremum of $A_{\text{m}}$ over $\calR$ as $A_{\text{ms}}$, we find from \eqref{CSAMGM} that
\be
\int_{\calR}\xg\bfS\cdot\osix{A}\mskip2mu[\xg\bfS]\dv\le 27\int_{\calR}A_{\text{m}}|\xg\bfS|^2\dv\le 27A_{\text{ms}}\int_{\calR}|\xg\bfS|^2\dv.
\label{oneside1}
\ee
Furthermore, by the assumption \eqref{Aposdef}$_1$ on $\osix{A}$, there exists a constant $a_{\text{inf}}$ such that 
\be
a_{\text{inf}}=\inf_{\bfS_g\in\calS_g}\int_{\calR}\xg\bfS_g\cdot\osix{A}\mskip2mu[\xg\bfS_g]\dv,
\label{ainf}
\ee
where
\be
\calS_g=\left\{\bfT:\bfT\in\calS, \int_{\calR}|\xg\bfT|^2\dv=1\right\}.
\label{calSg}
\ee
We notice, by the definition \eqref{calSg} of $\calS_g$, that if $\bfS\in\calS$, then 
\be
\frac{\xg\bfS}{(\int_{\calR}|\xg\bfS|^2\dv)^{\frac{1}{2}}}\in\calS_g.
\ee
Accordingly, by \eqref{ainf}, it follows that for any $\bfS\in\calS$,
\be
\int_{\calR}\frac{\xg\bfS}{(\int_{\calR}|\xg\bfS|^2\dv)^{\frac{1}{2}}}\cdot\osix{A}\mskip2mu\left[\frac{\xg\bfS}{(\int_{\calR}|\xg\bfS|^2\dv)^{\frac{1}{2}}}\right]\dv\ge a_{\text{inf}}
\ee
or, equivalently,
\be
\int_{\calR}\xg\bfS\cdot\osix{A}\mskip2mu[\xg\bfS]\dv\ge a_{\text{inf}}\int_{\calR}|\xg\bfS|^2\dv.
\label{cinfres}
\ee
Combining \eqref{oneside1} and \eqref{cinfres}, and invoking the definition \eqref{normsAC}$_1$ of the norm $\|\cdot\|_{\sbbA}$, we see that
\be
\sqrt{a_{\text{inf}}}\mskip2mu\Big(\int|\xg\bfS|^2\dv\Big)^{\frac{1}{2}} \le \|\bfS\|_{\sbbA} \le \sqrt{27A_{\text{ms}}}\mskip2mu\Big(\int|\xg\bfS|^2\dv\Big)^{\frac{1}{2}}
\ee
for each $\bfS\in\calS$, whereby, referring to \eqref{Ehatineq1}, we find that
\be
\sqrt{\frac{a_{\text{inf}}}{C^2+1}}\mskip2mu\|\bfS\|_{H^1(\calR)} \le \|\bfS\|_{\sbbA} \le \sqrt{27A_{\text{ms}}}\mskip2mu\|\bfS\|_{H^1(\calR)}
\label{Ehatineq2}
\ee
for each $\bfS\in\calS$. Thus, the norm $\|\cdot\|_{\sbbA}$ is equivalent to the $H^1(\calR)$ norm on $\calS$, establishing the proposition.
\end{proof}

\begin{proposition}\label{CL2}
The norm $\|\cdot\|_{\sbbC}$ is equivalent to the $L^2(\calR)$ norm on $\calS$.
\end{proposition}
\begin{proof}
Note that the $L^2(\calR)$ norm of $\bfS$ is defined by
\be
\|\bfS\|_{L^2(\calR)}=\Big(\int_{\calR}|\bfS|^2\dv\Big)^{\frac{1}{2}}.
\label{L2norm}
\ee
By the Cauchy--Schwarz inequality, 
\be
\bfS\cdot\mfC[\bfS]\le (|\mfC[\bfS]|^2)^{\frac{1}{2}}(|\bfS|^2)^{\frac{1}{2}}.
\label{CS2}
\ee
By the inequality of arithmetic and geometric means, it follows that
\be
|\mfC[\bfS]|^2\le (9C_{\text{m}})^2 |\bfS|^2,
\label{AMGM2}
\ee
where $C_{\text{m}}$ denotes maximum of the absolute values of the Cartesian components of $\mfC$:
\be
C_{\text{m}}=\text{max}(|C_{ijkl}|), \qquad i,j,k,l\in\{1,2,3\}.
\ee
Combining \eqref{CS2} and \eqref{AMGM2} yields
\be
\bfS\cdot\mfC[\bfS]\le 9C_{\text{m}}|\bfS|^2.
\label{CSAMGM2}
\ee
Denoting the supremum of $C_{\text{m}}$ over $\calR$ as $C_{\text{ms}}$, we find from \eqref{CSAMGM2} that
\be
\int_{\calR}\bfS\cdot\mfC[\bfS]\dv\le 9\int_{\calR}C_{\text{m}}|\bfS|^2\dv\le 9C_{\text{ms}}\int_{\calR}|\bfS|^2\dv.
\label{oneside}
\ee
Furthermore, by the assumption \eqref{Aposdef}$_2$ on $\mfC$, there exists a constant $c_{\text{inf}}$ such that 
\be
c_{\text{inf}}=\inf_{\bfS_m\in\calS_m}\int_{\calR}\bfS_m\cdot\mfC[\bfS_m]\dv,
\label{cinf2}
\ee
where
\be
\calS_m=\left\{\bfT:\bfT\in\calS, \int_{\calR}|\bfT|^2\dv=1\right\}.
\label{calSn}
\ee
We notice, by the definition \eqref{calSn} of $\calS_n$, that if $\bfS\in\calS$, then 
\be
\frac{\bfS}{{(\int_{\calR}|\bfS|^2\dv)}^{\frac{1}{2}}}\in\calS_n.
\ee
Accordingly, by \eqref{cinf2}, it follows that for any $\bfS\in\calS$,
\be
\int_{\calR}\frac{\bfS}{(\int_{\calR}|\bfS|^2\dv)^{\frac{1}{2}}}\cdot\mfC\left[\frac{\bfS}{(\int_{\calR}|\bfS|^2\dv)^{\frac{1}{2}}}\right]\dv\ge c_{\text{inf}}
\ee
or, equivalently,
\be
\int_{\calR}\bfS\cdot\mfC[\bfS]\dv\ge c_{\text{inf}}\int_{\calR}|\bfS|^2\dv.
\label{cinfres2}
\ee
Combining \eqref{oneside} and \eqref{cinfres2}, and invoking the definitions \eqref{normsAC}$_2$ and \eqref{L2norm} of $\|\cdot\|_{\sbbC}$ and $\|\cdot\|_{L^2(\calR)}$, respectively, we see that
\be
\sqrt{c_{\text{inf}}}\mskip2mu\|\bfS\|_{L^2(\calR)} \le \|\bfS\|_{\sbbC} \le 3\sqrt{C_{\text{ms}}}\mskip2mu\|\bfS\|_{L^2(\calR)}.
\ee
Thus, the norm $\|\cdot\|_{\sbbC}$ is equivalent to the $L^2(\calR)$ norm on $\calS$, establishing the proposition.
\end{proof}
\section{Existence of a solution in Subsection \ref{sec:inf}}\label{app:existence}
In this Appendix, we establish the following theorem:
\begin{theorem}\label{existence_theorem}
A minimizer of $E$, given by \eqref{E_nondim}, in the set 
\be
\calP=\left\{\bfT:\bfT\in\calS,\mskip2mu \int_{\calR}\bfT\cdot\mfC[\bfT]\,\text{{\em d}}v=1,\mskip2mu \int_{\calR}\bfS_N\cdot\mfC[\bfT]\,\text{{\em d}}v=0,\mskip1mu N\in\mathbb{N}_{N_0}\right\},
\label{setP}
\ee
with $\calS$ given by \eqref{setS}, exists.     
\end{theorem}

We establish the theorem in several steps. Notice, first, that since $E$ is positive-definite, the problems of minimizing $E$ and 
\be
\|\bfS\|_{\sbbA}=\sqrt{2E(\bfS)}=\Big(\int_{\calR}\xg\bfS\cdot\osix{A}\mskip2mu[\xg\bfS]\dv\Big)^{\frac{1}{2}}
\label{hatE}
\ee
are equivalent. We next establish that the values of $\|\cdot\|_{\sbbA}$ evaluated on $\mathcal{P}$ have a greatest lower bound $\displaystyle \inf_{\calP}\|\cdot\|_{\sbbA}>0$. 
\begin{proposition} \label{glbprop}
The values of the functional $\|\cdot\|_{\sbbA}$ evaluated on $\mathcal{P}$ have a greatest lower bound $\displaystyle \inf_{\calP}\|\cdot\|_{\sbbA}>0$.    
\end{proposition}
\begin{proof}
We notice from the positive definiteness \eqref{Aposdef}$_1$ of $\osix{A}$ over the set $\calS$ that the values of the functional $\|\cdot\|_{\sbbA}$ evaluated over any subset of $\calS$ that does not contain the zero second-order tensor field are bounded below by a positive number. Furthermore, since the elements of the set $\calP$ satisfy the normalization condition $\int_{\calR}\bfS\cdot\mfC[\bfS]\dv=1$, we conclude that $\calP\subset\calS$ does not contain the zero second-order tensor field. Thus, the values of $\|\cdot\|_{\sbbA}$ over $\calP$ are bounded below by a positive number. The completeness of the real number line guarantees the existence of the supremum $\displaystyle \inf_{\calP}\|\cdot\|_{\sbbA}$ of all such positive numbers, thus establishing the proposition.
\end{proof}

By Proposition \ref{glbprop}, there exists a minimizing sequence $(\bfSigma_n)_{n\in\mathbb{N}}$ in $\mathcal{P}$ such that
\be
\lim_{n\to\infty} \|\bfSigma_n\|_{\sbbA}=\inf_{\calP}\|\cdot\|_{\sbbA}.
\label{glb}
\ee
Since $\calP$ is a subset of $\calS$, it follows that $(\bfSigma_n)_{n\in\mathbb{N}}$ belongs to $\calS$.  Then, because $(\bfSigma_n)_{n\in\mathbb{N}}$ is bounded in the norm $\|\cdot\|_{\sbbA}$, by Proposition \ref{hatEH1}, it is bounded in the $H^1(\calR)$ norm. Thus, there exists a subsequence $(\bfSigma_{n_k})_{n_k\in\mathbb{N}}$ that converges weakly in the $H^1(\calR)$ norm to some $\bfSigma_0\in H^1(\calR)$, that is,
\be
(\bfSigma_{n_k})_{n_k\in\mathbb{N}}\rightharpoonup\bfSigma_0 \qquad\text{in}\qquad H^1(\calR).
\label{H1conv}
\ee
Furthermore, since the set $H^1(\calR)$ is compactly embedded in the set $L^2(\calR)$, $(\bfSigma_{n_k})_{n_k\in\mathbb{N}}$ converges strongly in the $L^2(\calR)$ norm to $\bfSigma_0$, that is,
\be
(\bfSigma_{n_k})_{n_k\in\mathbb{N}}\rightarrow\bfSigma_0 \qquad\text{in}\qquad L^2(\calR).
\label{L2conv}
\ee
We next use \eqref{H1conv} and \eqref{L2conv} to show that $\bfSigma_0$ belongs to $\calP$.

\begin{proposition}\label{S0inP}
    $\bfSigma_0$ belongs to $\calP$.
\end{proposition}
\begin{proof}
To show that $\bfSigma_0$ belongs to $\calP$, we must show, with reference to \eqref{setS} and \eqref{setP}, that (i) $\bfSigma_0\in\Sym$, (ii) $\xd\bfSigma_0=\bf0$, (iii) $\bfSigma_0\bfn|_{\partial\calR}=\bf0$, (iv) $\int_{\calR}\bfSigma_0\cdot\mfC[\bfSigma_0]\dv=1$, (v) $E(\bfSigma_0)<\infty$, and (vi) $\int_{\calR}\bfS_N\cdot\mfC[\bfSigma_0]\dv=0$, $N\in\mathbb{N}_{N_0}$. We show each of these properties below.

\begin{enumerate}[(i)]
    \item Since each $\bfSigma_{n_k}\in\Sym,n_k\in\mathbb{N}$, it follows that
    \be
    \int_{\calR}\bfSigma_{n_k}\cdot\boldsymbol{\mathit{\Omega}}_{\calR}\dv=0
    \ee
    for all square-integrable skew-symmetric tensor fields $\boldsymbol{\mathit{\Omega}}_{\calR}$ on $\calR$. Then, by \eqref{L2conv} and the continuity of the scalar product, we find that
    \be
    \int_{\calR}\bfSigma_0\cdot\boldsymbol{\mathit{\Omega}}_{\calR}\dv=0.
    \ee
    Thus, $\bfSigma_0\in\Sym$.

    \item Let $\boldsymbol{\phi}$ be a smooth, compactly supported vector field on $\calR$ and consider 
    \be
    I=\Big|\int_{\calR}\xd\bfSigma_0\cdot\boldsymbol{\phi}\dv\Big|.
    \ee
    We show below that $I$ is bounded from above by zero and, thus, that $I=0$. 
    
    We begin by noting that since $\bfSigma_0\in H^1(\calR)$, $\xd\bfSigma_0$ is square-integrable. Thus, by H\"{o}lder's inequality, $I$ is well-defined. Since each $\bfSigma_{n_k},n_k\in\mathbb{N}$, in the minimizing sequence is divergence-free and $\boldsymbol{\phi}$ is compactly supported on $\calR$, it follows, by the divergence theorem and H\"{o}lder's inequality, that
    \begin{align}
\Big|\int_{\calR}\xd\bfSigma_0\cdot\boldsymbol{\phi}\dv\Big|&=\Big|\int_{\calR}(\xd\bfSigma_0-\xd\bfSigma_{n_k})\cdot\boldsymbol{\phi}\dv\Big|
    \notag\\[4pt]
    &=\Big|\int_{\calR}(\bfSigma_0-\bfSigma_{n_k})\cdot\xg\boldsymbol{\phi}\dv\Big|\notag\\[4pt]
    &\le\Big(\int_{\calR}|(\bfSigma_0-\bfSigma_{n_k})|^2\dv\Big)^{\frac{1}{2}}\Big(\int_{\calR}|\xg\boldsymbol{\phi}|^2\dv\Big)^{\frac{1}{2}}.
    \label{divineq}
    \end{align}
    Since $\boldsymbol{\phi}$ is smooth, $ \int_{\calR}|\xg\boldsymbol{\phi}|^2\dv$ is finite and, thus, by \eqref{L2conv}, the last term in \eqref{divineq} goes to zero. Consequently,
    \be
    \int_{\calR}\xd\bfSigma_0\cdot\boldsymbol{\phi}\dv=0.
    \ee
    Finally, since $\boldsymbol{\phi}$ is arbitrary, and the set of smooth functions with compact support is dense in the set $L^2(\calR)$, we conclude that $ \xd\bfSigma_0=\bf0$.
    
    \item Let $\boldsymbol{\gamma}$ be a smooth vector field on $\calR$ and consider 
    \be
    I_b=\Big|\int_{\partial\calR}(\bfSigma_0\bfn)\cdot\boldsymbol{\gamma}\da\Big|.
    \ee
    We show below that $I_b$ is bounded from above by zero and, thus, that $I_b=0$.
    
    Notice first that, by the divergence theorem, 
    \begin{align}
    0&=\int_{\calR}\xd(\bfSigma_0-\bfSigma_{n_k})\cdot\boldsymbol{\gamma}\dv
    \notag\\[4pt]
    &=\int_{\partial\calR}((\bfSigma_0-\bfSigma_{n_k})\bfn)\cdot\boldsymbol{\gamma}\da-\int_{\calR}(\bfSigma_0-\bfSigma_{n_k})\cdot\xg\boldsymbol{\gamma}\dv.
    \label{moot}
    \end{align}
    Since $\bfSigma_0\in H^1(\calR)$ and $\bfSigma_{n_k} \in H^1(\calR),n_k\in\mathbb{N}$, the integrals on the second line of \eqref{moot} are well-defined. By ${\bfSigma_{n_k}\bfn=\bf0}$, \eqref{moot}, and H\"{o}lder's inequality, it follows that
    \begin{align}
     \big|\int_{\partial\calR}(\bfSigma_0\bfn)\cdot\boldsymbol{\gamma}\da\big|&=\big|\int_{\partial\calR}((\bfSigma_0-\bfSigma_{n_k})\bfn)\cdot\boldsymbol{\gamma}\da\big|
     \notag\\[4pt]
     &=\big|\int_{\calR}(\bfSigma_0-\bfSigma_{n_k})\cdot\xg\boldsymbol{\gamma}\dv\big|
    \notag\\[4pt]
    &\le\Big(\int_{\calR}|(\bfSigma_0-\bfSigma_{n_k})|^2\dv\Big)^{\frac{1}{2}}\Big(\int_{\calR}|\xg\boldsymbol{\gamma}|^2\dv\Big)^{\frac{1}{2}}.   
    \label{mootoo}
    \end{align}
    Again, by \eqref{L2conv}, the product on the final line of \eqref{mootoo} goes to zero. Since $\boldsymbol{\gamma}$ is arbitrary, we conclude that $\bfSigma_0\bfn|_{\partial\calR}=\bf0$.
    
    \item Since, from \eqref{L2conv}, the sequence $(\bfSigma_{n_k})_{n_k\in\mathbb{N}}$ converges strongly to $\bfSigma_0$ in the $L^2(\calR)$ norm, by the equivalence of the $L^2(\calR)$ norm and the norm $\|\cdot\|_{\sbbC}$ as established in Proposition \ref{CL2}, we find that 
    \be
    (\bfSigma_{n_k})_{n_k\in\mathbb{N}}\rightarrow\bfSigma_0 \qquad\text{in}\qquad \|\cdot\|_{\sbbC}.
    \label{Cnormconv}
    \ee
    Then, since for each $\bfSigma_{n_k},n_k\in\mathbb{N}$, $\|\bfSigma_{n_k}\|_{\sbbC}=1$ and because norm is a continuous function, it follows from \eqref{Cnormconv} that $\|\bfSigma_0\|_{\sbbC}=1$ or, equivalently, $\int_{\calR}\bfSigma_0\cdot\mfC[\bfSigma_0]\dv=1$.

    \item From items (i)--(iii) above, we conclude that $\bfSigma_0$ satisfies the same conditions as $\bfS$ in \eqref{intro}. Thus, $\bfSigma_0$ is a residual stress field. Accordingly, the arguments from \eqref{poincare} through \eqref{Ehatineq2} hold for $\bfSigma_0$. Then, because $\bfSigma_0\in H^1(\calR)$, it follows from \eqref{Ehatineq2} and \eqref{hatE} that $ E(\bfSigma_0)<\infty$.

    \item Since for each $ \bfSigma_{n_k},n_k\in\mathbb{N}$, $ \int_{\calR}\bfS_N\cdot\mfC[\bfSigma_{n_k}]\dv=0$, $N\in\mathbb{N}_{N_0}$, and the scalar product is a continuous function, it follows from \eqref{Cnormconv} that $\int_{\calR}\bfS_N\cdot\mfC[\bfSigma_0]\dv=0$, $N\in\mathbb{N}_{N_0}$.
\end{enumerate}
\end{proof}

For the next proposition, we use the following characterization, due to \citet{ekeland1999convex}, of weakly lower semi-continuous functions.
\begin{theorem}\label{Temam}
If a functional defined on a convex set is strongly lower semi-continuous with respect to a norm and is convex, then it is weakly lower semi-continuous with respect to that norm.
\end{theorem}

Let $\calP_c$ be the set 
\be
\calP_c=\left\{\bfT:\bfT\in\calS,\mskip2mu \int_{\calR}\bfT\cdot\mfC[\bfT]\dv\le 1,\mskip2mu \int_{\calR}\bfS_N\cdot\mfC[\bfT]\dv=0,\mskip1mu N\in\mathbb{N}_{N_0}\right\}.
\label{setPbar}
\ee
Notice that $\calP_c$ is convex. We then have the following proposition. 

\begin{proposition} \label{wls}
The functional $\|\cdot\|_{\sbbA}$ is weakly lower semi-continuous in the $H^1(\calR)$ norm on $\calP_c$.
\end{proposition}
\begin{proof}
Since, by Proposition \ref{hatEH1}, $\|\cdot\|_{\sbbA}$ is equivalent to the $H^1(\calR)$ norm on $\calS$ and $\calP_c$ is a subset of $\calS$, $\|\cdot\|_{\sbbA}$ is continuous and, thus, strongly lower semi-continuous in the $H^1(\calR)$ norm on $\calP_c$. Furthermore, since $\|\cdot\|_{\sbbA}$ is a norm, it is a convex functional. By Theorem \ref{Temam}, the functional $\|\cdot\|_{\sbbA}$ is weakly lower semi-continuous in the $H^1(\calR)$ norm on $\calP_c$.
\end{proof}

\begin{proposition}
    $\displaystyle\|\bfSigma_0\|_{\sbbA}=\inf_{\calP}\|\cdot\|_{\sbbA}$.
\end{proposition}
\begin{proof}
Since $\calP$ is a subset of $\calP_c$, every element of $\calP$ is also an element of $\calP_c$. Consequently, $(\bfSigma_{n_k})_{n_k\in\mathbb{N}}$ and, by Proposition \ref{S0inP}, $\bfSigma_0$ belong to $\calP_c$. Since $(\bfSigma_{n_k})_{n_k\in\mathbb{N}}\rightharpoonup \bfSigma_0$ in the $H^1(\calR)$ norm, it then follows from Proposition \ref{wls} that
\be
\|\bfSigma_0\|_{\sbbA}\le\inf_{\calP}\|\cdot\|_{\sbbA}.
\label{Ele0}
\ee
However, since, by \eqref{glb}, $\inf_{\calP}\|\cdot\|_{\sbbA}$ is a lower bound of the functional $\|\cdot\|_{\sbbA}$ on $\calP$ and, by Proposition \ref{S0inP}, $\bfSigma_0$ is an element of $\calP$, it follows that
\be
\|\bfSigma_0\|_{\sbbA}\ge\inf_{\calP}\|\cdot\|_{\sbbA}.
\label{Ege0}
\ee
Thus, by \eqref{Ele0} and \eqref{Ege0}, 
\be
\|\bfSigma_0\|_{\sbbA}=\inf_{\calP}\|\cdot\|_{\sbbA}.
\ee
\end{proof}

We have, therefore, established that a minimizer of $\|\cdot\|_{\sbbA}$ or, equivalently, of $E$ over $\calP$, namely $\bfSigma_0$, exists, as claimed in Theorem \ref{existence_theorem}. 

\section{At the boundaries of the strip in Figure \ref{fig:pk}, \boldmath{$E(S)>0$} for a non-zero spherically symmetric \boldmath{$S$}}\label{app:inclusive}
Recall, from \eqref{Vbound}, that the points on the boundaries of the semi-infinite strip of parameter choices which ensure that the integrand $U$ of $E$ is positive-definite are {\em not} viable for a general residual stress field. In this Appendix, we show that for a non-zero spherically symmetric residual stress field $\bfS$, $E$ satisfies $E(\bfS)>0$ at each such boundary point.

Consider, first, the parameters in the set 
\be
\partial\calV_1=\big\{(\beta,\gamma):\beta+\gamma=5,\mskip4mu \beta\ge 5/3\big\}.
\ee
Substituting $\beta+\gamma=5$ in \eqref{W5} yields
\be
E(\bfS)=\frac{1}{2}\int_{\calR}\Big(5|\bfH_2|^2+\frac{3\beta-5}{2}|\bfH_3|^2\Big)\dv.
\label{Eapp3}
\ee
Since $\beta\ge5/3$, the second term in \eqref{Eapp3} is greater than or equal to zero. Thus, for $E$ to vanish, $\bfH_2$ must vanish. By \eqref{GradS} and \eqref{H2def}, $\bfH_2=\bf0$ implies that $\Sperp'=2\Spar'$. Differentiating \eqref{bvp_sph}$_1$, substituting $\Sperp'=2\Spar'$ and integrating the resulting expression, we find that $\Spar$ has the form $\Spar(r)=\hat{c}r^2+\tilde{c}$, where $\hat{c}$ and $\tilde{c}$ are constants. In conjunction with the boundary conditions \eqref{bvp_sph}$_3$, this implies that $\Spar=0$. From the equilibrium equation \eqref{bvp_sph}$_1$, it then follows that $\Sperp=0$. Thus, $E(\bfS)=0$ on $\partial\calV_1$ only if $\bfS=\bf0$.

Consider, next, the parameters in the set 
\be
\partial\calV_2=\big\{(\beta,\gamma):\gamma=2\beta,\mskip4mu 0< \beta< 5/3\big\}.
\ee 
Since, as seen in Subsection \ref{sec:dimred}, for the spherically symmetric case, the integrand $U$ of $E$ depends on the parameters only through the combination $p=\beta+\gamma$. Accordingly, the parameter $k$ does not influence the positive-definiteness of $W$. In other words, given $0<p<5$, $U$ is positive-definite for all real $k$. In particular, it is positive-definite for $k=0$ or, equivalently, on $\partial\calV_2$.  
We conclude that $E(\bfS)=0$ on $\partial\calV_2$ only if $\bfS=\bf0$.

Consider, finally, the parameters in the set 
\be
\partial\calV_3=\big\{(\beta,\gamma):\beta+\gamma=0,\mskip4mu \beta\ge 0\big\}.   
\ee
From \eqref{Esph1}, we see that for $\beta+\gamma=0$,
\be
E(\bfS)=\frac{1}{4}\int_{\calR}|\xg(\tr\bfS)|^2\dv.
\ee
For $E$ to vanish, it follows that $\xg(\tr\bfS)=\bf0$ in $\calR$. Thus, $\tr\bfS$ is constant in $\calR$. However, a residual stress field must have zero mean and, thus, $\tr\bfS=\Spar+2\Sperp=0$. Using this in the equilibrium condition \eqref{bvp_sph}$_1$ and integrating, we find that $\Spar$ has the form $\Spar(r)=\hat{c}/r^3$, where $\hat{c}$ is a constant. In conjunction with the boundary conditions \eqref{bvp_sph}$_3$, the foregoing result implies $\Spar=0$. Since $\Spar+2\Sperp=0$, it follows that $\Sperp=0$. Thus, $\bfS=\bf0$. 

We conclude that, for the spherically symmetric case, $E(\bfS)=0$ on $\partial\calV_1$, $\partial\calV_2$, and $\partial\calV_3$ only if $\bfS=\bf0$.

\section{Proof of the decomposition in Eq.~\eqref{Sudec}} \label{app:decomposition}
We finally show that any square-integrable symmetric second-order tensor field on a simply-connected region with a smooth boundary can be decomposed uniquely into a square-integrable residual stress field and the symmetric part of the gradient of a differentiable vector field. We will establish this result with the aid of the following theorem due to \citet{maggiani2015compatible}.

\begin{theorem}\label{Maggiani}
\noindent A square-integrable second-order symmetric tensor field $\bfT$ on a simply-connected domain $\calR$ with smooth boundary $\partial\calR$ can be uniquely decomposed as
\be
\bfT=(\emph{Curl}\mskip2mu(\emph{Curl}\mskip2mu\bfF))^{\trans}+\emph{sym}\mskip2mu\emph{Grad}\mskip2mu\bfu,
\ee
where $\bfu$ is a square-integrable vector field on $\calR$ with square-integrable gradient, $\bfF$ and $(\emph{Curl}\mskip2mu(\emph{Curl}\mskip2mu\bfF))^{\trans}$ are square-integrable symmetric second-order tensor fields on $\calR$, $\emph{Div}\mskip2mu\bfF=\bf0$ on $\calR$, and $\bfF\bfn=\bf0$ on $\partial\calR$. 
\end{theorem}

\begin{proposition}\label{corollary}
A square-integrable second-order symmetric tensor field $\bfT$ over a simply-connected domain $\calR$ with smooth boundary $\partial\calR$ can be uniquely decomposed as
\be
\bfT=\bfG+\emph{sym}\mskip2mu\emph{Grad}\mskip2mu\bfw,
\ee
where $\bfw$ is a square-integrable vector field on $\calR$ with square-integrable gradient, $\bfG$ is a square-integrable symmetric second-order tensor field on $\calR$, $\emph{Div}\mskip2mu\bfG=\bf0$ on $\calR$, and $\bfG\bfn=\bf0$ on $\partial\calR$.
\end{proposition}
\begin{proof}
By Theorem \ref{Maggiani}, $\bfT$ can be uniquely decomposed as
\be
\bfT=(\xc(\xc\bfF))^{\trans}+\sym\xg\bfu,
\label{TFu}
\ee
with $\xd\bfF=\bf0$ and $\bfF\bfn=\bf0$. Let
\be
\hat{\bfF}=(\xc(\xc\bfF))^{\trans}.
\label{hatF}
\ee
Notice that $\hat{\bfF}$ is symmetric by Theorem \ref{Maggiani}. \citet{fosdick2005stokes} have shown that a symmetric $\hat{\bfF}$ of the form given in \eqref{hatF} satisfies the conditions:
\be
\left.
\ba
\xd\hat{\bfF}&=\bf0 \qquad \text{on} \qquad \calR,
\\[4pt]
\int_{\partial\calR}\hat{\bfF}\bfn\da&=\bf0,
\\[4pt]
\int_{\partial\calR}(\bfr-\bfo)\times\hat{\bfF}\bfn\da&=\bf0.
\ea
\label{hatFcond}\mskip3mu\right\}
\ee

Let $\hat{\bfG}$ be such that
\be
\left.\ba
\left.\ba
&\hat{\bfG}=\sym\xg\hat{\bfw},
\\[4pt]
&\xd\hat{\bfG}=\bf0, 
\ea\mskip4mu\right\}\qquad &\text{on} \qquad \calR,
\\[4pt]
\hat{\bfG}\bfn=\hat{\bfF}\bfn \mskip75mu\qquad &\text{on} \qquad \partial\calR.
\\[4pt]
\ea\mskip3mu\right\}
\label{Gw}
\ee
With reference to \eqref{Gw}, $\hat{\bfG}$ admits a physical interpretation as follows: It is the strain in a homogeneous isotropic linear elastic body with stiffness tensor equal to the fourth-order identity tensor, such that the body is in equilibrium with the surface traction $\hat{\bfF}\bfn$. Furthermore, we see from \eqref{hatFcond} that $\hat{\bfF}\bfn$ is a self-equilibrating traction field. Hence, a unique $\hat{\bfG}$ and $\hat{\bfw}$ satisfying \eqref{Gw} exist. Let 
\be
\bfG=\hat{\bfF}-\hat{\bfG}.
\label{Gdefapp}
\ee
By \eqref{hatFcond}$_1$ and \eqref{Gw}$_{2,3}$, it follows that
\be
\left.\ba
\xd\bfG=\bf0 \qquad &\text{on} \qquad \calR,
\\[4pt]
\bfG\bfn=\bf0 \qquad &\text{on} \qquad \partial\calR.
\ea\mskip3mu\right\}
\label{Gprop}
\ee
On invoking \eqref{hatF}, \eqref{Gw}$_1$, and \eqref{Gdefapp}, \eqref{TFu} takes the form
\be
\bfT=\bfG+\sym\xg\bfw,
\ee
where 
\be
\bfw=\bfu+\hat{\bfw}.
\label{bfw}
\ee
Since $\hat{\bfF}$, $\hat{\bfG}$, $\bfu$, and $\hat{\bfw}$ are uniquely determined, $\bfG$ and $\bfw$ are uniquely determined by \eqref{Gdefapp} and \eqref{bfw}, thus confirming the proposition.
\end{proof}

\bibliographystyle{apalike}
\bibliography{arXiv}
\end{document}